%% file: main.tex
\documentclass[sigconf,dvipsnames]{acmart}
\setcopyright{none}

\AtBeginDocument{%
  \providecommand\BibTeX{{%
    \normalfont B\kern-0.5em{\scshape i\kern-0.25em b}\kern-0.8em\TeX}}}

\input{localdefs.tex}

\title{Properties of Inconsistency Measures for Databases}

\begin{document}




\author{Ester Livshits}
\email{esterliv@cs.technion.ac.il}
\affiliation{%
  \institution{Technion}
  \city{Haifa}
  \country{Israel}
  \postcode{3200003}
}

\author{Rina Kochirgan}
\email{rina.k@campus.technion.ac.il}
\affiliation{%
  \institution{Technion}
  \city{Haifa}
  \country{Israel}
  \postcode{3200003}
}

\author{Segev Tsur}
\email{segevtsur@campus.technion.ac.il}
\affiliation{%
  \institution{Technion}
  \city{Haifa}
  \country{Israel}
  \postcode{3200003}
}

\author{Ihab F. Ilyas}
\email{ilyas@uwaterloo.ca}
\affiliation{%
  \institution{University of Waterloo}
  \city{Waterloo}
  \state{ON}
  \country{Canada}
}

\author{Benny Kimelfeld}
\email{bennyk@cs.technion.ac.il}
\affiliation{%
  \institution{Technion}
  \city{Haifa}
  \country{Israel}
  \postcode{3200003}
}

\author{Sudeepa Roy}
\email{sudeepa@cs.duke.edu}
\affiliation{%
  \institution{Duke University}
  \city{Durham}
  \state{NC}
  \country{USA}
}

\renewcommand{\shortauthors}{Livshits et al.)}

\begin{abstract}
   How should we quantify the inconsistency of a database that violates integrity constraints?
  Proper measures are important for various tasks, such as progress
  indication and action prioritization in cleaning systems, and
  reliability estimation for new datasets.  To choose an appropriate
  inconsistency measure,
  it is important to identify the desired properties in the application and understand which of these is
  guaranteed or at least expected in practice.  For example, in some
  use cases the inconsistency should reduce if constraints are
  eliminated; in others it should be stable and avoid jitters and
  jumps in reaction to small changes in the database. 
  We embark on
  a systematic investigation of properties for database inconsistency
  measures. We investigate a collection of basic measures that have
  been proposed in the past in both the Knowledge Representation and
  Database communities, analyze their theoretical properties, and
  empirically observe their behaviour in an experimental study.  We
  also demonstrate how the framework can lead to new inconsistency
  measures by introducing a new measure that, in contrast to the
  rest, satisfies all of the properties we consider and can be
  computed in polynomial time.
\end{abstract}





\maketitle

\input{introduction.tex}

\input{preliminaries.tex}

\input{measures.tex}

\input{postulate.tex}

\input{update.tex}

\input{experiments}

\input{conclusions.tex}

\paragraph*{Acknowledgements}
This work was supported by the
German Research Foundation (DFG) Project 412400621 (DIP program) and the Israel Science
  Foundation (ISF), Grant 768/19. The work of Sudeepa Roy was supported by NSF awards IIS-1552538, IIS-1703431, IIS-2008107, and NIH award R01EB025021. We also thank the reviewers of this paper for very helpful comments and suggestions.

\bibliographystyle{ACM-Reference-Format}
\bibliography{main}

\newpage
\appendix
\input{appendix.tex}

\end{document}

%% file: localdefs.tex
\usepackage{graphicx}
\usepackage{balance}
\usepackage{url}
\usepackage{pifont}
\newcommand{\xmark}{\ding{55}}%
\usepackage{xcolor}
\usepackage{colortbl}
\usepackage{enumitem}
\usepackage{ifthen}
\usepackage{subfig}
\usepackage{pgfplots}
\usepackage{pgfplotstable}

\def\rhorel{\mathbin{\rho}}
\def\e#1{\emph{#1}}

\input{macros.tex}

\def\set#1{\{#1\}}
\newtheorem{theorem}{Theorem}
\newtheorem{proposition}{Proposition}
\newcommand{\eat}[1]{}

\newcommand{\red}[1]{{\color{red}#1}}
\newcommand{\added}[1]{{\color{black}#1}}

\definecolor{Gray}{gray}{0.9}

\newcommand{\mrevcolor}{\color{red}}
\newcommand{\mrev}[1]{{\mrevcolor#1}}

\newcommand{\revonecolor}{\color{violet}}
\newcommand{\revone}[1]{{\revonecolor#1}}

\newcommand{\revtwo}[1]{{\revtwocolor#1}}

\newcommand{\revthree}[1]{{\revthreecolor#1}}

\def\removerevcolor{true}

\ifdefined\removerevcolor
\renewcommand{\revone}[1]{#1}
\renewcommand{\revtwo}[1]{#1}
\renewcommand{\revthree}[1]{#1}
\renewcommand{\mrev}[1]{#1}
\renewcommand{\revonecolor}{}

\renewcommand{\mrevcolor}{}
\fi

\def\scs{\mathbf{S}}

\def\sig{\mathit{sig}}
\def\tids{\mathit{ids}}
\newcommand{\eqdef}{\mathrel{{\mathop:{\!=}}}}
\def\dbs{\mathrm{DB}}

\def\P{\mathcal{P}}

\def\att#1{\mathrm{#1}}
\def\rel#1{\mathit{#1}}

\newcommand{\algname}[1]{{\sf #1}}

{\vskip 1em}

%% file: macros.tex
\def\la{\leftarrow}
\def\ra{\rightarrow}
\def\angs#1{\mathord{\langle{#1\rangle}}}
\def\del#1#2{\angs{-#1}(#2)}
\def\upd#1#2#3{\angs{#1\la #2}(#3)}
\def\ins#1#2{\angs{+#1}(#2)}
\def\R{\mathcal{R}}
\def\C{\mathbf{C}}

\def\str{^{\star}}

\def\I{\mathcal{I}}

\newcommand{\cut}[1]{}

\def\MI{\mathsf{MI}}
\def\MC{\mathsf{MC}}

\def\SI{\mathsf{SelfInconsistencies}}

\def\Id{\I_{\mathsf{d}}}
\def\Imr{\I_{\R}}

\def\Ilmr{\I_{\R}^{\mathsf{lin}}}
\def\Imi{\I_{\MI}}
\def\Imc{\I_{\MC}}
\def\Ip{\I_{\mathsf{P}}}
\def\Cfd{\C_{\mathsf{FD}}}
\def\Cdc{\C_{\mathsf{DC}}}

\newtheorem{example}{Example}
\newtheorem{lemma}{Lemma}
\newcommand{\dom}{{\mathsf{Val}}}

\def\rel#1{\textsc{#1}}
\def\att#1{\textsf{#1}}

\def\val#1{\mathtt{#1}}

\newenvironment{repeatresult}[2]
{\vskip0.5em\par\noindent\nobreakspace{\bfseries #1 #2.\!}\em}
{\vskip1em}

\newenvironment{repproposition}[1]{\begin{repeatresult}{Proposition}{#1}}{\end{repeatresult}}

\newenvironment{reptheorem}[1]{\begin{repeatresult}{Theorem}{#1}}{\end{repeatresult}}

\newboolean{long}   

\def\textforshortversion#1{
\ifthenelse{\boolean{long}}{}{#1}}

\def\textforlongversion#1{
\ifthenelse{\boolean{long}}{#1}{}}


%% file: introduction.tex
\section{Introduction}\label{sec:introduction}
Inconsistency of databases may arise in a variety of situations for a
variety of reasons. Database records may be collected from imprecise
sources (social encyclopedias/networks, sensors attached to
appliances, cameras, etc.)  via imprecise procedures (natural-language
processing, signal processing, image analysis, etc.), or be integrated
from different sources with conflicting information or formats.
Common principled approaches to handling inconsistency consider
databases that violate \e{integrity constraints}
, but can nevertheless
be \e{repaired} by means of \e{operations} that revise the database
and resolve inconsistency~\cite{DBLP:conf/pods/ArenasBC99}.

Instantiations differ in the supported types of integrity constraints,
allowed operations, and optimization goals. The constraints may be
Functional Dependencies (FDs) or the more general Equality-Generating
Dependencies (EGDs)
or, more generally, Denial Constraints (DCs),
and they may be referential (foreign-key) constraints or the more
general inclusion dependencies~\cite{DBLP:conf/icdt/AfratiK09}.  A
repair operation can be a \e{deletion} of a tuple, an \e{insertion} of
a tuple, or an \e{update} of an attribute value.  Operations may be
associated with different \e{costs}, representing levels of trust in
data items or extent of
impact~\cite{DBLP:conf/icdt/LopatenkoB07}. Optimization goals can be
minimization of the costs incurred \cite{DBLP:conf/icdt/KolahiL09,
  DBLP:conf/pods/LivshitsKR18} or producing the most probabable clean
database from a distribution
\cite{DBLP:journals/pvldb/RekatsinasCIR17} that fully (hard
constraints) or partially (soft constraints) fixes the violated
integrity constraints.  Various approaches and systems have been
proposed for data cleaning and, in particular, data repairing
(e.g.,~\cite{DBLP:journals/pvldb/EbaidEIOQ0Y13,DBLP:journals/pvldb/GeertsMPS13,DBLP:journals/pvldb/RekatsinasCIR17}
to name a few).

We explore the question of \e{how to measure database inconsistency}.
This question may arise in different situations.  A measure of
inconsistency can be used for estimating the potential usefulness and
cost of incorporating databases for downstream
analytics~\cite{DBLP:conf/edbt/PapottiNK15}.  We can also use an
inconsistency measure for implementing a progress bar for data
repairing.  Indeed, the importance of incorporating progress
indicators in interactive systems has been well recognized and studied
in the Human-Computer Interaction (HCI)
community~\cite{Myers:1985:IPP:317456.317459,
  Conn:1995:TAT:223904.223928,Harrison:2007:RPB:1294211.1294231,Harrison:2010:FPB:1753326.1753556,Villar:2013:AIM:2542634.2542637}.
\revthree{Note that an inconsistency measure can be used in any cleaning system,  where the measure gives one of several indications of progress (namely those that are associated with integrity constraints).} Beyond the indication of progress, a measure of inconsistency can be
used for prioritizing and recommending actions in data
repairing---address the tuples that have the highest
\e{responsibility} to the inconsistency level (e.g., Shapley value for
inconsistency~\cite{DBLP:journals/ai/HunterK10,DBLP:conf/ijcai/YunVCB18,DBLP:journals/corr/abs-2009-13819})
or the ones that might result in the greatest reduction in
inconsistency.

Example measures include the number of violations in the database, the
number of tuples involved in violations, and the number of operations
needed to reach consistency.  To choose an appropriate inconsistency
measure for a specific use case, it is important to identify the
desired properties of the application and understand which of these is
guaranteed or expected in practice.  For example, to effectively
communicate progress indication in repairing, the measure should
feature certain characteristics. While it should react to changes in
the database and show progress when the database is ``cleaner'' with
respect to the integrity constraints, it should minimize jitters,
jumps and sudden changes in reaction to small changes in the database,
and in turn, should well correlate with the ``expected waiting
time''---an important aspect in machine-user interaction.  It should
also be computationally tractable in interactive
systems. Luo et al.~\cite{DBLP:conf/sigmod/LuoNEW04} enunciate the
importance of these properties, referring to them as ``acceptable
pacing'' and ``minimal overhead," respectively, in progress indicators
for database queries.


\begin{figure*}[t]
\subfloat[\small \bf A clean database $D_0$.]{
\begin{small}
\begin{tabular}{r|c|c|c|c|c|c|}
\hline
& Id &	Type &	Name &	Continent &	Country &	Municipality	\\
\hline
$f_1$ & 00AA &	Small airport &	Aero B Ranch &		NAm &	US &			Leoti	\\

$f_2$ & 7FA0 &	heliport &	Florida Keys Memorial Hospital Heliport  &	NAm  &	US  &		Key West\\

$f_3$ & 7FA1 &	Small airport &	Sugar Loaf Shores Airport  & NAm  &	US  &	Key West\\

$f_4$ & KEYW &	Medium airport &	Key West International Airport  &	NAm &	US &	Key West\\

$f_5$ & KNQX &	Medium airport &	Naval Air Station Key West/Boca Chica Field  &	NAm &	US  &			Key West\\\hline
\end{tabular}
\end{small}
\label{fig:clean}
}\vspace{-0.5em}
\subfloat[\small \bf A noisy database $D_1$ (bold values are changed from the clean database. \revtwo{The underlines are explained in Section~\ref{sec:measure}}).]{
\begin{small}
\begin{tabular}{r|c|c|c|c|c|c|}
\hline
& Id &	Type &	Name &	Continent &	Country &	Municipality	\\
\hline
$f_1$ & 00AA &	Small airport &	Aero B Ranch &		NAm &	US &			Leoti	\\

$f_2$ & 7FA0 &	heliport &	Florida Keys Memorial Hospital Heliport  &	{\bf Am}  &	{\bf USA}  &		Key West\\

$f_3$ & 7FA1 &	Small airport &	Sugar Loaf Shores Airport  & \underline{NAm}  &	\underline{US}  &	Key West\\

$f_4$ & KEYW &	Medium airport &	Key West International Airport  &	\underline{NAm} &	{\bf USA} &	Key West\\

$f_5$ & KNQX &	Medium airport &	Naval Air Station Key West/Boca Chica Field  &	{\bf Am} &	\underline{US}  &			Key West\\\hline
\end{tabular}
\end{small}
\label{fig:noisy-1}
}\vspace{-0.5em}
\subfloat[\small \bf A noisy database $D_2$ (bold values are changed from the clean database).]{
\begin{small}
\begin{tabular}{r|c|c|c|c|c|c|}
\hline
& Id &	Type &	Name &	Continent &	Country &	Municipality	\\
\hline
$f_1$ & 00AA &	Small airport &	Aero B Ranch &		NAm &	US &			Leoti	\\

$f_2$ & 7FA0 &	heliport &	Florida Keys Memorial Hospital Heliport  &	{\bf Am}  &	{\bf USA}  &		Key West\\

$f_3$ & 7FA1 &	Small airport &	Sugar Loaf Shores Airport  & NAm  &	US  &	Key West\\

$f_4$ & KEYW &	Medium airport &	Key West International Airport  &	NAm &	{\bf USA} &	Key West\\

$f_5$ & KNQX &	Medium airport &	Naval Air Station Key West/Boca Chica Field  &	NAm &	US  &			Key West\\\hline
\end{tabular}
\end{small}
\label{fig:noisy-2}
}
\vspace{-0.6em}
\caption{\bf A clean and two noisy Airport databases from Example~\ref{eg:running}.}
\vspace{-2mm}
\label{fig:running}
\end{figure*}

In this paper, we explore desirable properties of inconsistency measures
and demonstrate the analysis and design of specific inconsistency
measures with respect to these properties.  As a guiding principle, we
adopt the approach of \e{rationality postulates} of \e{inconsistency
  measures} over knowledge bases that have been investigated in depth
by the Knowledge Representation (KR) and Logic
communities~\cite{DBLP:conf/ijcai/KoniecznyLM03,DBLP:journals/jolli/Knight03,DBLP:journals/jiis/GrantH06,DBLP:conf/kr/HunterK08,DBLP:journals/ai/HunterK10,DBLP:journals/ijar/GrantH17,DBLP:journals/ki/Thimm17}.
Yet, the studied measures and postulates fall short of capturing our
desiderata for various reasons. First, inconsistency is typically
measured over a knowledge base of logical sentences (formulas without
free variables).  In databases, we reason about \e{tuples} (facts) and
fixed integrity constraints, and inconsistency typically refers to the
tuples rather than the constraints (which are treated as exogenous
information). In particular, while a collection of sentences, and even a single sentence, might
form a contradiction, a set of tuples can be inconsistent only in the
presence of integrity constraints.  Hence, as recently
acknowledged~\cite{DBLP:journals/corr/abs-1804-08834,DBLP:journals/corr/abs-1904-03403},
it is of importance to seek inconsistency measures that are closer to
database applications. 
More fundamentally, in order to capture
the repairing \e{process} and corresponding changes to the database,
the measures should be aware of the underlying \e{repairing operations}
(e.g., tuple insertion, deletion, or revision).

The following example illustrates the concept of a noisy
  database \revtwo{and will be used throughout the paper to demonstrate the
  differences among several inconsistency measures.}

\begin{example}\label{eg:running}
We have an 
Airport database\footnote{Simplified version of \url{https://ourairports.com/data/}.} with the schema
$$\rel{Airport}(\underline{\att{Id}},\att{Type},\att{Name},\att{Continent},\att{Country},\att{Municipality})$$
and the FDs ``$\att{Municipality}\rightarrow
\att{Continent}\,\,\att{Country}$''
and ``$\att{Country}\rightarrow \att{Continent}$.''
A clean database $D_0$ where both FDs hold is given in Figure~\ref{fig:clean}. Figures~\ref{fig:noisy-1} and \ref{fig:noisy-2} depict two noisy versions of the database that violate one or both FDs: $D_1$ is obtained from $D_0$ by modifying four values, and $D_2$ is obtained from $D_0$ by modifying three values. The changed values are shown in {\bf bold}. 
\revtwo{When the clean database is not available, it is not easy to detect which of $D_1$ and $D_2$ is `more noisy'. An inconsistency measure aims to formally quantify the inconsistency level in a noisy database, given a set of constraints. We will revisit this example to illustrate different measures in Section~\ref{sec:measure} (Table~\ref{tab:running_measures}).}  \qed
\end{example}

To illustrate some of the inconsistency measures, let us consider the case where all constraints are
anti-monotonic (i.e., consistency cannot be violated by deleting
tuples). For example, FDs and the more general DCs are anti-monotonic, whereas referential key constraints are not.
One simple measure is the \e{drastic measure} $\Id$, which is $1$ if
the database is inconsistent, and $0$
otherwise~\cite{DBLP:journals/ki/Thimm17}. Of course, this measure
hardly communicates progress of cleaning, as it does not change until the very
end. \revtwo{Clearly, we have $\Id(D_0)=0$ while $\Id(D_1)=\Id(D_2)=1$ for the databases of Figure~\ref{fig:running}; hence, this measure does not differentiate the database $D_1$ from $D_2$, motivating the need to study more fine-grained measures.}

\revtwo{
What about the measure $\Ip$ that counts the
\e{problematic tuples}, which are the tuples that participate in
\emph{(minimal) witnesses of
inconsistency}~\cite{DBLP:conf/kr/HunterK08,DBLP:journals/ai/HunterK10}? This measure suffers from a disproportional reaction to repairing
operations: the removal of a single tuple (e.g., a misspelling
of a country name) may cause a drastic reduction in inconsistency (for
example, if a single tuple is involved in most
violations of the constraints). The measure $\Imi$, that simply counts
the constraint violations (i.e., the minimal inconsistent subsets), suffers from the same problem. }
\revtwo{
As another example, consider the measure $\Imc$ that counts the
\emph{maximal consistent subsets}. This measure suffers from
various problems: adding constraints can cause a \e{reduction} in
inconsistency, it may fail to reflect repairing operations (i.e., the number of maximal consistent subsets may stay the same after applying a repairing operation), and, again, it may react disproportionally
to small changes. Moreover, it is hard to compute (\#P-complete)
already for simple FDs~\cite{DBLP:conf/pods/LivshitsK17}. 
}

A recent effort in the database community proposed measures
based on the concept of a \e{minimal repair}---the
minimal number of repairing operations needed to obtain
consistency~\cite{DBLP:conf/sum/Bertossi18}. We refer to this measure
as $\Imr$. 
We show that $\Imr$ satisfies the aforementioned rationality
criteria that we formally define later on, and so, we provide a formal
justification to its adoption. Yet, it is again intractable (NP-hard)
even for simple sets of
FDs~\cite{DBLP:conf/icdt/KolahiL09, DBLP:conf/pods/LivshitsKR18}. \mrev{Interestingly, we are able to show that a \e{linear relaxation} of this measure, that we propose in this paper as a new inconsistency measure and refer to as $\Ilmr$, provides a combination of rationality and tractability.}

\noindent
\textbf{Contributions.~} \mrev{Given the numerous choices of inconsistency
measures, we ask and address the question of what
properties are desired in such a measure, so that one can make an
informed choice.} In particular,
we make the following contributions.
\cut{ We make a step towards
  formalizing the behavior of inconsistency measures such as the
  aforementioned ones.  }
\par
{\bf (1) Formalize and analyze desired properties of inconsistency measures: } 
We define four properties
of inconsistency measures in the context of a \e{repair system} (a space of possibly weighted repairing operations) with the following intuitive meaning:
\vspace{-0.2em}
\begin{enumerate}[label=(\roman*)]
\itemsep0em
    \item {\bf Positivity}: the
measure is strictly positive if and only if the database is
inconsistent,
\item {\bf Monotonicity}: inconsistency cannot be reduced if the constraints get stricter,
\item {\bf Continuity}: a single operation
can have a limited relative impact on inconsistency, and
\item {\bf Progression}: we can always find an operation
that reduces inconsistency.
\end{enumerate}
\vspace{-0.3em}
Moreover, we use these properties to analyze a  collection of
inconsistency measures adapted from the KR and Logic literature (Section~\ref{sec:measure}).
Specifically, we examine the inconsistency measures against the properties, and show that most of the measures violate at least one of these properties (see Table~\ref{table:postulates}). However, the measure $\Imr$ stands out satisfying all the properties (Section~\ref{sec:postulates}). 
\par
{\bf (2) Computational Complexity:} As an additional desideratum for practical purposes, we consider the complexity of computing the measures. While some measures are tractable, the measure $\Imr$, that satisfies all the properties, is known to be intractable already for FDs~\cite{DBLP:conf/pods/LivshitsKR18}. We provide a stronger result showing that computing $\Imr$ is hard already for the case of a single
EGD  (Section~\ref{sec:complexity}). 
\par
{\bf (3) Propose a new measure satisfying all desiderata:}
For tuple deletions as repairing operations, we propose a new measure $\Ilmr$ relaxing the concept of $\Imr$. We show that $\Ilmr$ satisfies all
four properties and is computable in polynomial time, even for the
general case of arbitrary sets of denial constraints; hence, it provides a combination of rationality and tractability (Section~\ref{sec:rational}).

\par
{\bf (4) Empirical evaluation of measures:} 
We empirically evaluate the behavior of these measures on multiple
datasets with different integrity constraints and inconsistency levels
(Section~\ref{sec:experiments}). \mrev{We show that the measures that behave well in theory also exhibit a good empirical
      behavior under different repair models.} 

\noindent
\textbf{Other related work.~}  \revtwo{As mentioned earlier,
  inconsistency measures were extensively studied in the KR
  community~\cite{DBLP:journals/ki/Thimm17}. Rationality
  postulates for inconsistency measures have been studied in the
  database context, for example by Parisi and
  Grant~\cite{DBLP:journals/corr/abs-1904-03403}.  Contrasting with
  that work, we devise properties that account for operational aspects
  of repairs and repair systems, and we study the computational
  complexity of the measures. The properties and complexity analysis
  lead us to tractable and rational measures that have not been
  considered in the past (e.g., $\Ilmr$). Also contrasting with that
  work, we conduct a thorough empirical evaluation over candidate
  measures to explore their behavior in practice; to the best of our
  knowledge, we are the first to conduct such an experimental study.}

  Some other relevant
  studies have focused on different aspects of inconsistency measures
  for databases. Martinez et al.~\cite{10.1007/978-3-540-75256-1_12}
  introduced axioms that are inherently based on a numerical
  interpretation of the database values.  Grant and
Hunter~\cite{DBLP:conf/ecsqaru/GrantH11} considered repairing (or
\e{resolution}) operators, and focused on the trade-off between
inconsistency reduction and \e{information loss}: an operation is
beneficial if it causes a high reduction in inconsistency alongside a
low loss of information.  \eat{Instead, our focus here is on measuring
  \e{progress} of repairing, and it is an interesting future direction
  to understand how the two relate to each other and/or can be
  combined.} Another complementary problem is that of associating
\e{individual facts} with portions of the database inconsistency
(e.g., the Shapley value of the fact) and using these portions to
define preferences among
repairs~\cite{DBLP:conf/ijcai/YunVCB18,DBLP:journals/corr/abs-2009-13819}.

\revthree{In the bigger context, note that inconsistency measures
  quantify the extent to which constraints are violated.  Data
  cleaning in general goes beyond constraint resolution to challenges
  such as outlier detection, record linkage, entity resolution and so
  on (see, e.g.,~\cite{DBLP:journals/pvldb/AbedjanCDFIOPST16} and
  references therein).  Yet, one can use an inconsistency measure in
  any cleaning system, adopting any correction mechanism whatsoever,
  for providing a specific view of the progress: the extent of
  constraint violation (but not necessarily general dirt). Measuring
  the status for other kinds of dirt is an important future challenge that goes
  beyond the scope of this paper.  }

\revtwo{The remainder of the paper is organized as follows. We present
  preliminary concepts and terminology in Section~\ref{sec:prelim}. We
  consider inconsistency measures in Section~\ref{sec:measure} and
  their properties in Section~\ref{sec:postulates}. In
  Section~\ref{sec:rational}, we discuss complexity aspects and
  propose a new rational and tractable measure. We describe an experimental evaluation
  in Section~\ref{sec:experiments}. Finally, we make concluding
  remarks and discuss directions for future work in
  Section~\ref{sec:conclusions}.  }

%% file: preliminaries.tex
\section{Preliminaries}\label{sec:prelim}
We first give the basic terminology and concepts.

\textbf{Relational model.~} A \e{relational schema} (or just
\e{schema} for short) $\scs$ has a finite set of \e{relation symbols}
$R$, each associated with a \e{relation signature} $\sig(R)$. In turn, 
a relation signature $\alpha$ is a sequence $(A_1,\dots,A_k)$ of
distinct \e{attributes} $A_i$, and $k$ is the \e{arity} of
$\alpha$. 
If $\sig(R)$ has arity $k$, then we say that $R$ is $k$-ary.  A
\e{fact} $f$ (over $\scs$) is an expression of the form
$R(c_1,\dots,c_k)$, where $R$ is a $k$-ary relation symbol of $\scs$,
and $c_1$,\dots, $c_k$ are \e{values}.
If $f=R(c_1,\dots,c_k)$ is a fact and
$\sig(R)=(A_1,\dots,A_k)$, then we refer to the value $c_i$ as $f.A_i$.


A \e{database} $D$ over $\scs$ is a mapping from a finite set
$\tids(D)$ of \e{record identifiers} to facts over $\scs$. The set of all databases over
$\scs$ is denoted by $\dbs(\scs)$.  We denote by $D[i]$ the fact that
$D$ maps to the identifier $i$. A database $D$ is a \e{subset} of a
database $D'$, denoted $D\subseteq D'$, if
$\tids(D)\subseteq\tids(D')$ and $D[i]=D'[i]$ for all $i\in\tids(D)$.

An \e{integrity constraint} is a first-order sentence over $\scs$. A
database $D$ satisfies a set $\Sigma$ of integrity constraints,
denoted $D\models\Sigma$, if $D$ satisfies every constraint
$\sigma\in\Sigma$. If $\Sigma$ and $\Sigma'$ are sets of integrity
constraints, then we write $\Sigma\models\Sigma'$ to denote that
$\Sigma$ \e{entails} $\Sigma'$; that is, every database that satisfies
$\Sigma$ also satisfies $\Sigma'$. We also write $\Sigma\equiv\Sigma'$
to denote that $\Sigma$ and $\Sigma'$ are \e{equivalent}, that is,
$\Sigma\models\Sigma'$ and $\Sigma'\models\Sigma$. By a \e{constraint
  system} we refer to a class $\C$ of integrity constraints (e.g., the
class of all functional dependencies).

As a special case, a \e{Functional Dependency} (FD) $R:X\rightarrow Y$, where $R$ is a relation symbol and $X,Y\subseteq\sig(R)$, states that every two facts that agree on (i.e., have equal values in) every attribute of $X$ should also agree on  $Y$. The more general \e{Equality Generating Dependency} (EGD) has the form $\forall\vec x\left[\varphi_1(\vec x)\wedge\dots\wedge\varphi_k(\vec x)\rightarrow (y_1=y_2)\right]$ where each $\varphi_j(\vec x)$ is an atomic formula over the schema and $y_1$ and $y_2$ are variables in $\vec x$. Finally, a \e{Denial Constraint} (DC) has the form $\forall\vec x\neg\left[\varphi_1(\vec x)\wedge\dots\wedge\varphi_k(\vec x)\wedge\psi(\vec x)\right]$ where each $\varphi_j(\vec x)$ is an atomic formula and $\psi(\vec x)$ is a conjunction of atomic comparisons over  $\vec x$.

\begin{example}\label{eg:db-constraints}
  The schema of our running example consists of a single 5-ary relation symbol $\rel{Airport}$. All databases of Figure~\ref{fig:running} have five facts with
  $\tids(D) = \{f_1, \cdots, f_5\}$. The constraint set $\Sigma$ consists of two
  FDs as shown in Example~\ref{eg:running}.
\end{example}

\textbf{Repair systems.~}
Let $\scs$ be a schema. A \e{repairing operation} (or just
\e{operation}) $o$ transforms a database $D$ over $\scs$ into another
database $o(D)$ over $\scs$, that is,
$o:\dbs(\scs)\rightarrow\dbs(\scs)$. An example is \e{tuple deletion},
denoted $\del{i}{\cdot}$, parameterized by a tuple identifier $i$ and
applicable to $D$ if $i\in\tids(D)$; the result
$\del{i}{D}$ is obtained from $D$ by deleting the tuple identifier $i$
(along with the corresponding fact $D[i]$).  Another example is
\e{tuple insertion}, denoted $\ins{f}{\cdot}$, parameterized by a fact
$f$; the result $\ins{f}{D}$ is obtained from $D$ by adding $f$ with a
new tuple identifier. (For convenience, this is the minimal integer
$i$ such that $i\notin\tids(D)$.) A third example is \e{attribute
  update}, denoted $\upd{i.A}{c}{\cdot}$, parameterized by a tuple
identifier $i$, an attribute $A$, and a value $c$, and applicable to
$D$ if $i\in\tids(D)$ and $A$ is an attribute of the fact $D[i]$; the
result $\upd{i.A}{c}{D}$ is obtained from $D$ by setting 
$D[i].A$ to $c$. By convention, if $o$ is not applicable to $D$, then
it keeps $D$ intact, that is, $o(D)=D$.

A \e{repair system} (over a schema $\scs$) is a collection of
repairing operations with an associated \e{cost} of applying to a
given database. For example, a smaller change of value might be less
costly than a greater one~\cite{DBLP:conf/lid/GardeziBK11}, and some
facts might be more costly than others to
delete~\cite{DBLP:conf/icdt/LopatenkoB07,DBLP:conf/pods/LivshitsKR18}
or
update~\cite{DBLP:conf/icdt/KolahiL09,DBLP:conf/pods/LivshitsKR18,DBLP:journals/is/BertossiBFL08};
changing a person's zip code might be less costly than changing the
person's country, which, in turn, might be less costly than deleting
the entire person's record.  Formally, a repair system $\R$ is a pair
$(O,\kappa)$ where $O$ is a set of operations and
$\kappa:O\times\dbs(\scs)\ra[0,\infty)$ is a cost function that
assigns the cost $\kappa(o,D)$ to applying $o$ to $D$. We require that
$\kappa(o,D)=0$ if and only if $D=o(D)$; that is, the cost is nonzero
when, and only when, an actual change occurs.

For a repair system $\R$, we denote by $\R\str=(O\str,\kappa\str)$ the repair system 
of all \e{sequences} of operations from $\R$, where the cost
of a sequence is the sum of costs of the individual operations
thereof. 
\textforlongversion{
Formally, for $\R=(O,\kappa)$, the repair system $\R\str$ is
$(O\str,\kappa\str)$, where $O\str$ consists of all compositions
$o=o_m\circ\dots\circ o_1$, with $o_j\in O$ for all $j=1,\dots,m$,
defined inductively by $o_m\circ\dots\circ
o_1(D)=o_m(o_{m-1}\circ\dots\circ o_1(D))$ and
$\kappa\str(o_m\circ\dots\circ
o_1,D)=\kappa(o_m,o_{m-1}\circ\dots\circ
o_1(D))+\kappa\str(o_{m-1}\circ\dots\circ o_1,D)$.
}
Let $\C$ be a constraint system and $\R$ a repair system. We say that
$\C$ is \e{realizable by} $\R$ if it is always possible to make a
database satisfy the constraints of $\C$ by repeatedly applying operations
from $\R$. Formally, $\C$ is realizable by $\R$ if for every database
$D$ and a finite set $\Sigma\subseteq\C$ there is a sequence $o$ in
$\R\str$ such that $o(D)\models\Sigma$.  An example of $\C$ is the
system $\Cfd$ of all FDs $R:X\ra Y$.
An example of $\R$ is the \e{subset} system,
denoted $\R_\subseteq$, where $O$ is the set of all tuple deletions
(hence, the result is always a subset of the original database), and
$\kappa$ is determined by a special \e{cost} attribute,
$\kappa(\del{i}{D})=D[i].\mathtt{cost}$, if a cost attribute exists, and otherwise,
$\kappa(\del{i}{D})=1$ (every tuple has unit cost for  deletion). Observe that $\R_\subseteq$ realizes $\C$,
since the latter consists of anti-monotonic constraints.

\begin{example}\label{eg:repair}
The database $D_1$ of Figure~\ref{fig:noisy-1} may be obtained from $D_0$ by applying the following sequence of attribute updates (here, we assume that the identifier of a fact $f_i$ is $i$):
{\small
\begin{align*}
    &E_1=\upd{2.\att{Continent}}{\val{Am}}{D_0}
    &E_2=\upd{2.\att{Country}}{\val{USA}}{E_1}\\
    &E_3=\upd{4.\att{Country}}{\val{USA}}{E_2}
    &D_1=\upd{5.\att{Continent}}{\val{Am}}{E_3}
\end{align*}
}
If we consider tuple deletions and insertions, we may obtain $D_1$ from $D_0$ by applying the following sequence of operations:
{\small
\begin{align*}
    &E_1=\del{2}{D_0}
    &E_2=\ins{f_2}{E_1}\quad\quad\quad
    &E_3=\del{4}{E_2}\\
    &E_4=\ins{f_4}{E_3}
    &E_5=\del{5}{E_4}\quad\quad\quad
    &D_1=\ins{f_5}{E_5}
\end{align*}
}
where $f_2,f_4$ and $f_5$ are the facts of $D_1$. We may also assign a cost to each operation. For example, if both tuple deletions and attribute updates are allowed, we may associate a higher cost with the operation $\del{4}{D}$ that deletes an entire fact compared to the operation $\upd{4.\att{Country}}{\val{USA}}{D}$ that updates a single attribute value.
\end{example}

%% file: measures.tex
\section{Inconsistency Measures}\label{sec:measure}

\begin{table*}[t]
  \caption{\revtwo{The inconsistency measure values on the databases of our
  running example.}}
  \begin{small}
\begin{tabular}{|c||c|c||c|c||}
\hline
 & \multicolumn{2}{|c||}{Noisy database $D_1$} & \multicolumn{2}{|c||}{Noisy database $D_2$}\\\hline
{\bf Measure} & {\bf value} & {\bf explanation} & {\bf value} & {\bf explanation} \\\hline\hline
$\Id$  &1 & inconsistent database &1 & inconsistent database \\\hline
$\Imr$ (deletions) & 3 & e.g., remove $\{f_2, f_4, f_5\}$ or $\{f_3, f_4, f_5\}$  & 2 & e.g., remove $\{f_2, f_3\}$ or $\{f_2, f_4\}$\\\hline
$\Imr$ (updates) & 4 & change the values shown in {\bf bold} or with an \underline{underline}& 3 & change the values shown in {\bf bold} \\\hline
$\Imi$ & 7 & $|\{\{f_2, f_3\}, \{f_2, f_4\},\{f_2, f_5\}, \{f_3, f_4\},\{f_3, f_5\}, \{f_4, f_5\},\{f_1, f_5\}\}|$   & 5 & $|\{f_2, f_3\}, \{f_2, f_4\},\{f_2, f_5\}, \{f_3, f_4\},\{f_4, f_5\}\}|$ \\\hline
$\Ip$ & 5 & all tuples are involved in violations & 4 & all tuples except $f_1$ are involved in violations \\\hline
$\Imc$ & 3 & $|\{\{f_1, f_2\}, \{f_1, f_3\}, \{f_1, f_4\}, \{f_5\}\}| - 1$ & 2 &  $|\{\{f_1, f_2\}, \{f_1, f_4\}, \{f_1, f_3, f_5\}\}| - 1$  \\\hline
$\Ilmr$ & 2.5 & e.g., assign 0.5 to all $f_1, \cdots, f_5$ (see Section~\ref{sec:lmr}) & 2 & e.g., assign 0.5 to all $f_2, \cdots, f_5$ (see Section~\ref{sec:lmr}) \\\hline
\end{tabular}
\end{small}
\label{tab:running_measures}
\end{table*}

Let $\scs$ be a schema and $\C$ a constraint system. An
\e{inconsistency measure} is a function $\I$ that maps a finite set
$\Sigma\subseteq\C$ of integrity constraints and a database $D$ to a
number $\I(\Sigma,D)\in[0,\infty)$.
Intuitively, a high $\I(\Sigma,D)$ implies that $D$ is far from
satisfying $\Sigma$. We make two standard requirements:
\begin{itemize}[topsep=0pt,itemsep=-1ex,partopsep=1ex,parsep=1ex]
\item $\I$ is zero on consistent databases; that is, $\I(\Sigma,D)=0$
  whenever $D\models\Sigma$;
\item $\I$ is invariant under logical equivalence of constraints; that
  is, $\I(\Sigma,D) = \I(\Sigma',D)$ whenever $\Sigma\equiv\Sigma'$.
\end{itemize}

Next, we discuss several examples of inconsistency measures. Some of these
measures (namely, $\Id$, $\Imi$, $\Ip$ and $\Imc$) are adapted from the
survey of Thimm~\cite{DBLP:journals/ki/Thimm17} to the context of relational databases.
\revtwo{Table~\ref{tab:running_measures} summarizes the values of all measures on the noisy databases $D_1$ and $D_2$ of Example~\ref{eg:running}.}
\par
\revtwo{\textbf{Drastic inconsistency measure $\Id$.}}
The simplest
measure is the \e{drastic inconsistency value}, denoted $\Id$, which is the
indicator function of inconsistency.
\[
\Id(\Sigma,D) \eqdef 
\begin{cases}
0 & \mbox{if $D\models\Sigma$;}\\
1 & \mbox{otherwise.}
\end{cases}
\]
\revtwo{In Figure~\ref{fig:running}, we have that $D_0 \models \Sigma$, while $D_1, D_2\not\models\Sigma$. Therefore, $\Id(\Sigma,D_0) = 0$ and $\Id(\Sigma,D_1) = \Id(\Sigma,D_2) = 1$}


\cut{
}


\revtwo{The following measures, $\Imi$, $\Ip$ and $\Imc$, apply to systems $\C$ of \e{anti-monotonic}
constraints.} Recall that an integrity constraint $\sigma$ is
anti-monotonic if for all databases $D$ and $D'$, if
$D\subseteq D'$ and $D'\models\sigma$, then $D\models\sigma$.
Examples of anti-monotonic constraints are the DCs~\cite{DBLP:journals/jiis/GaasterlandGM92},
the classic FDs, \e{conditional
  FDs}~\cite{DBLP:conf/icde/BohannonFGJK07}, and EGDs ~\cite{DBLP:conf/icalp/BeeriV81}.

\revtwo{\textbf{Minimal inconsistent subsets measure $\Imi$.~}} For a set $\Sigma\subseteq\C$ of \revtwo{anti-monotonic} constraints and a database
$D$, denote by $\MI_\Sigma(D)$ the set of all \e{minimal
  inconsistent} subsets of $D$; that is, the set of all $E\subseteq D$
such that $E\not\models\Sigma$ and $E'\models\Sigma$ for
all \revtwo{proper subsets} $E'\subsetneq E$. 
\revtwo{Since the constraints are anti-monotonic}, it holds that $D\models\Sigma$ if and
only if $\MI_\Sigma(D)$ is empty.  Drawing from known inconsistency
measures~\cite{DBLP:conf/kr/HunterK08,DBLP:journals/ai/HunterK10, DBLP:journals/ki/Thimm17}, the
measure $\Imi$, also known as \e{MI Shapley Inconsistency},
is the cardinality of this set.
\[\Imi(\Sigma,D) \eqdef |\MI_{\Sigma}(D)|\]
\par
\revtwo{\textbf{Problematic facts measure $\Ip$.}~}  A fact $f$ that belongs to a minimal inconsistent subset $K$ 
(i.e., $f\in K\in\MI_\Sigma(D)$) is called 
\e{problematic}, and the measure $\Ip$ counts the problematic
facts~\cite{DBLP:conf/ecsqaru/GrantH11}.
\[\Ip(\Sigma,D) \eqdef |\cup\MI_\Sigma(D)|\]

\begin{example}\label{eg:MI_P}
The set $\Sigma$ of constraints of Example~\ref{eg:running} consists of FDs, which are violated by pairs of facts; thus, the size of any minimal inconsistent subset w.r.t.~$\Sigma$ is two.
Since the database $D_0$ of Figure~\ref{fig:running} satisfies $\Sigma$, we have that $\Imi(\Sigma,D_0) = \Ip(\Sigma,D_0) = 0$. In $D_1$, all (six) pairs of facts from $\set{f_2,f_3,f_4,f_5}$ as well as the pair $\set{f_1,f_5}$ jointly violate $\Sigma$; hence, there are seven violating pairs in total and $\Imi(\Sigma,D_1)=7$. As every fact of $D_1$ occurs in at least one violating pair, we have that $\Ip(\Sigma,D_1)=5$. \revtwo{Similarly, $\Imi(\Sigma,D_2)=5$ as shown in Table~\ref{tab:running_measures}, and $\Ip(\Sigma,D_2)=4$ as all facts of $D_2$ except $f_1$ are involved in violations.}

\end{example}

\revtwo{\textbf{Maximal consistent subsets measure $\Imc$.~}} For a finite set $\Sigma\subseteq\C$ of \revtwo{anti-monotonic} constraints and a database
$D$, we denote by $\MC_\Sigma(D)$ the set of all \e{maximal
  consistent} subsets of $D$; that is, the set of all $E\subseteq D$
such that $E\models\Sigma$ and, moreover, $E'\not\models\Sigma$
whenever $E\subsetneq E'\subseteq D$.  Observe that if
$D\models\Sigma$, then $\MC_\Sigma(D)$ is simply the singleton
$\set{D}$, 
and for anti-monotonic constraints, the set $\MC_\Sigma(D)$ is never empty (since, e.g.,
the empty set is consistent). The measure $\Imc$ is the cardinality of
$\MC_\Sigma(D)$, minus one.
\[\Imc(\Sigma,D) \eqdef |\MC_{\Sigma}(D)|-1\]

\added{We note that the definition of $\Imc$ in the KR setting is slightly different than ours, as it takes into account the number of self-contradictions in the knowledge base~\cite{DBLP:conf/ecsqaru/GrantH11,DBLP:journals/ijar/GrantH17}. This highlights an important difference between the standard KR setting and ours. As aforementioned, we distinguish between tuples and integrity constraints, and we measure only tuple sets; thus, no self-contradictions exist per se---an individual tuple is always consistent. However, a tuple can be self-inconsistent in the presence of a relevant constraint (e.g., the DC ``$\mbox{height}>0$''). Parisi and Grant~\cite{DBLP:journals/corr/abs-1904-03403} refer to such tuples as \e{contradictory} tuples. Hence, we also consider here the following variant of $\Imc$ that counts \e{self-inconsistencies}.
\[\Imc'(\Sigma,D) \eqdef |\MC_{\Sigma}(D)|+|\SI(D)|-1\]}

\begin{example}\label{eg:MC}
In Figure~\ref{fig:running}, it holds that $\Imc(\Sigma,D_0) = 0$ as $D_0$ is consistent and $\MC_\Sigma(D_0) = \{D_0\}$. The $\Imc$ values for $D_1$ and $D_2$ are shown in Table~\ref{tab:running_measures} along with the maximal consistent subsets. For instance, for $D_1$, the set $\{f_1, f_2\}$ is a maximal consistent subset as adding $f_3$ introduces a violation of the FD $\att{Municipality}\rightarrow \att{Continent}\,\,\att{Country}$, and the addition of $f_4$ or $f_5$ introduces violations of both FD. As we consider FDs, for which there are no self-inconsistencies, for $D_0, D_1, D_2$, the $\Imc'$ and  $\Imc$ values coincide.
\end{example}

\mrev{All the inconsistency measures considered so far are adapted from the KR literature, and only take the database and set of integrity constraints into account. However, in the context of databases, it is also important to consider the \emph{repair system} that aims to fix an inconsistent database via repairing operations (e.g., value updates or tuple deletions).} The next measure assumes an underlying repair system $\R$ and constraint system $\C$ such that $\C$ is realizable by
$\R$.

\revtwo{\textbf{Minimum repair measure $\Imr$.}}   The measure $\Imr$ is the minimal cost of a sequence of
operations that repairs the database. It captures the intuition around
various notions of repairs known as \e{cardinality} repairs and
\e{optimal}
repairs~\cite{DBLP:conf/icdt/KolahiL09,DBLP:conf/pods/LivshitsKR18,DBLP:conf/icdt/AfratiK09}.
\[
\Imr(\Sigma,D) \eqdef \min\set{\kappa\str(o,D)\mid o\in O\str\mbox{ and }
  o(D)\models\Sigma}
\]
This measure is the same as the \e{d-hit} inconsistency measure introduced by Grant and Hunter~\cite{DBLP:conf/ecsqaru/GrantH13}. Moreover, $\Imr$ is the \e{distance from satisfaction} used in property
  testing~\cite{DBLP:journals/jacm/GoldreichGR98}
in the special case where the repair system consists of unit-cost
insertions and deletions.

\mrev{While many of the rule-based approaches to error detection and repairing in the database literature have considered anti-monotonic constraints like DCs, EGDs, and FDs~\cite{DBLP:conf/icde/ChuIP13,DBLP:journals/pvldb/RekatsinasCIR17,DBLP:journals/tods/Wijsen05,DBLP:conf/vldb/CongFGJM07,DBLP:journals/tods/FanGJK08,DBLP:conf/pods/ArenasBC99}, the measure $\Imr$ in general can be used with other types of constraints (like referential integrity constraints or other complex global constraints that are not anti-monotonic). This measure could also naturally incorporate weighted (soft) rules~\cite{DBLP:journals/corr/abs-2009-13821}.}
\par
\begin{example}\label{eg:repair2}
The values of $\Imr$ on the two noisy databases of Figure~\ref{fig:running} are given in Table~\ref{tab:running_measures} for the cases where the repairing operations are tuple deletions or attribute updates. As mentioned in Example~\ref{eg:MI_P}, all pairs of facts from $\set{f_2,f_3,f_4,f_5}$ in $D_1$ jointly violate $\Sigma$; thus, only one of these facts may appear in a repair, and the minimal number of facts that we need to remove from the database to satisfy $\Sigma$ is three (e.g., remove $\{f_2,f_4,f_5\}$). Assuming unit-cost deletions, we have that $\Imr(\Sigma,D_1)=3$ for the repair system $\R_\subseteq$. If we consider attribute updates, we need to update at least every bold (or underlined) value in $D_1$ to satisfy the FD $\att{Municipality}\rightarrow \att{Continent}\,\,\att{Country}$; hence, in this case, $\Imr(\Sigma,D_1)=4$ (again, for the case of unit-cost updates).
\end{example}

%% file: postulate.tex
\def\postul#1#2{
\begin{center}
\textit{\underline{#1}:}\, #2
\end{center}}

\section{Properties of  Measures}
\label{sec:postulates}
We now propose and discuss several properties (\e{postulates}) of
database inconsistency measures. 
We illustrate these properties over
the different measures of the previous
section. \mrev{In our examples throughout the section, we focus on the case where the
  constraints are FDs or the more general DCs, and the repair system is the subset system 
  $\R_{\subseteq}$ where only tuple deletions are allowed. We stress, however, that these inconsistency
  measures can be used under \e{any} repair system and constraint
  system (perhaps without the theoretical guarantees), and we
  illustrate that in our experimental study
  (Section~\ref{sec:experiments})}. The behavior of the measures with
respect to the properties is summarized in
Table~\ref{table:postulates}, which we discuss later on.  
\cut{The measures are all defined in Section~\ref{sec:measure}, except for $\Ilmr$ that
we define later in Section~\ref{sec:rational}.}

\newcolumntype{g}{>{\columncolor{Gray}}c}
\begin{table}[t]
\small
\centering
\caption{\label{table:postulates} \revone{Satisfaction of properties for
    $\Cfd$/$\Cdc$ and $\R_{\subseteq}$ (e.g., $\Imi$ satisfies monotonicity for FDs but not DCs).  Tractability (``PTime'' assuming $\mbox{P}\neq\mbox{NP}$) is discussed in Section~\ref{sec:rational}, where we also define the measure $\Ilmr$ in the last row.}}
\revonecolor
\begin{tabular}{|c|c|c|c|c|g|}\hline
& Pos. & Mono. & B.~Cont. & Prog. & PTime\\\hline
$\Id$ & \checkmark/\checkmark & \checkmark/\checkmark & \xmark/\xmark & \xmark/\xmark & \checkmark/\checkmark  \\\hline
$\Imi$ & \checkmark/\checkmark & \checkmark/\xmark & \xmark/\xmark & \checkmark/\checkmark  & \checkmark/\checkmark \\\hline
$\Ip$ & \checkmark/\checkmark & \checkmark/\xmark & \xmark/\xmark & \checkmark/\checkmark & \checkmark/\checkmark \\\hline
$\Imc$ & \checkmark/\xmark & \xmark/\xmark & \checkmark/\checkmark & \xmark/\xmark & \xmark/\xmark  \\\hline
$\Imc'$ & \checkmark/\checkmark & \xmark/\xmark & \xmark/\xmark & \xmark/\xmark & \xmark/\xmark  \\\hline
$\Imr$ &  \checkmark/\checkmark &  \checkmark/\checkmark &  \checkmark/\checkmark &  \checkmark/\checkmark & \xmark/\xmark\\\hline
\rowcolor{Gray}$\Ilmr$ &  \checkmark/\checkmark &  \checkmark/\checkmark &  \checkmark/\checkmark &  \checkmark/\checkmark & \checkmark/\checkmark \\\hline
\end{tabular}
\end{table}

\textbf{Positivity.}~ A basic property is \e{positivity}, sometimes referred to as
\e{consistency}~\cite{DBLP:journals/ijar/GrantH17,DBLP:journals/corr/abs-1904-03403}. This property is also the first axiom suggested by Martinez et al.~\cite{10.1007/978-3-540-75256-1_12}.
\postul{Positivity}{$\I(\Sigma,D)>0$ whenever $D\not\models\Sigma$.}
Each of $\Id$, $\Imi$, $\Ip$, $\Imc'$, and $\Imr$ satisfies
positivity \revone{(for any set of anti-monotonic constraints)}, but not $\Imc$. For example, let $D$ be a database with of two
facts, $R(a)$ and $R(b)$, and $\Sigma$ consists of the \revone{(denial)} constraint $\neg
R(a)$ (i.e., $R(a)$ is not in the database). Then $\Imc(\Sigma,D)=0$
since $\MC_\Sigma(D)=\set{R(b)}$. Observe that  the fact $R(a)$ is inconsistent by itself; hence, the example does not apply to $\Imc'$, that takes self-inconsistencies into account (so, $\Imc'(\Sigma,D)=1$). Yet, in the case of FDs (i.e., $\C=\Cfd$), every violation involves two facts, and so $|\MC_\Sigma(D)|\ge 2$ and positivity is satisfied.

\textbf{Monotonicity.}~ The next property is \e{monotonicity}---inconsistency cannot decrease
if the constraints get stricter.
\postul{Monotonicity}{$\I(\Sigma,D)\leq I(\Sigma',D)$ whenever
  $\Sigma'\models\Sigma$.}  For example, $\Id$ and $\Imr$ satisfy
monotonicity \revone{for the general class of anti-monotonic constraints}, since every repair w.r.t.~$\Sigma'$ is also a repair w.r.t.~$\Sigma$.
The measures $\Imi$ and $\Ip$ also satisfy monotonicity
in the special case of FDs, since in this case $|\MI_{\Sigma}(D)|$ is the number of fact pairs that jointly violate an FD, which can only increase when adding 
or strengthening FDs.
Yet, they may violate monotonicity when
going beyond FDs to the more general class of DCs.
\begin{proposition}
  In the case of $\Imi$ and $\Ip$, monotonicity can be violated
  already for the class of DCs.
\end{proposition}
\begin{proof}
  We begin with $\Imi$. Consider a schema with a single relation
  symbol, and for a natural number $k>0$, let $\Sigma_k$ consist of a
  single DC stating that \e{there are at most $k-1$ facts in the
    database}. (The reader can easily verify that, indeed, $\Sigma_k$
  can be expressed as a DC.)  Then,
  $\Imi(\Sigma_k,D)=\binom{n}{k}$ whenever $D$ has $n\geq k$ facts. In
  particular, whenever $k'>k$ and $D$ has $n\geq 2{k'}$ facts, it
  holds that $\Imi(\Sigma_{k'},D)>\Imi(\Sigma_{k},D)$ while
  $\Sigma_k\models\Sigma_{k'}$.

  We now consider $\Ip$.  Let $\scs$ be a schema that contains two
  relation symbols $R(A,B)$ and $S(A,B)$. Consider the following two
  EGDs (which are, of course, special cases of DCs):
\begin{align*}
\sigma_1&=\forall x,y,z,w [\big(R(x,y),S(x,z),S(x,w)\big)\Rightarrow z=w]\\
\sigma_2&=\forall x,z,w [\big(S(x,z),S(x,w)\big)\Rightarrow z=w]
\end{align*}
Let $\Sigma_1=\set{\sigma_1}$ and
$\Sigma_2=\set{\sigma_1,\sigma_2}$. Every set in $\MI_{\Sigma_1}(D)$
is of size three, while the size of the sets in $\MI_{\Sigma_2}(D)$ is
two. Hence, in a database where
$|\MI_{\Sigma_1}(D)|=|\MI_{\Sigma_2}(D)|$ (i.e., where
$\sigma_1$ is violated by
$\set{R(\val{a},\val{b}),S(\val{a},\val{c}),S(\val{a},\val{d})}$ if
and only if $\sigma_2$ is violated by
$\set{S(\val{a},\val{c}),S(\val{a},\val{d})}$), we have
$|P_{\Sigma_1}(D)|>|P_{\Sigma_2}(D)|$ while $\Sigma_2\models
\Sigma_1$.
\end{proof}

The measures $\Imc$ and $\Imc'$, on the other hand, can violate monotonicity even
for FDs \revone{(hence, also for DCs)}.
\begin{proposition}\label{prop:imc_monotonicity}
In the case of $\Imc$ and $\Imc'$, monotonicity can be violated already for the
class of FDs.
\end{proposition}
\begin{proof}
  Let $D$ consist of these facts over $R(A,B,C,D)$:
  {\small
\begin{gather*}
  f_1=R(0,0,0,0)\; f_2=R(1,0,0,0)
  f_3=R(1,1,0,1)\; f_4=R(0,1,0,1)
\end{gather*}
}
Let
$\Sigma_1=\set{A\ra B}$ and $\Sigma_2=\set{A\ra B,C\ra D}$.
Then $\Sigma_2\models\Sigma_1$ and the following hold.
{\small
\begin{align*}
\MC(\Sigma_1,D)&=\set{\set{f_1,f_2},\set{f_1,f_3},\set{f_2,f_4},\set{f_3,f_4}}\\
\MC(\Sigma_2,D)&=\set{\set{f_1,f_2},\set{f_3,f_4}}
\end{align*}
}
We conclude that $\Imc(\Sigma_1,D)=\Imc'(\Sigma_1,D)=3$ and $\Imc(\Sigma_2,D)=\Imc'(\Sigma_2,D)=1$, proving that
monotonicity is violated.
\end{proof}

Note that our monotonicity definition is different from that of Parisi and Grant~\cite{DBLP:journals/corr/abs-1904-03403}, as monotonicity in our case is restricted to the integrity constraints rather than the database. Here, we do not consider a property for monotonicity over the database (i.e., if $D\subseteq D'$ then $\I(D)\le \I(D')$), as adding a new tuple may, in fact, reduce inconsistency, for example, under foreign-key constraints.

Positivity and monotonicity serve as sanity conditions that the
measure indeed quantifies inconsistency---it
does not ignore inconsistency, and it does not reward strictness of
constraints. 
Next, we
propose two properties that are aware of the underlying repair system
$\R=(O,\kappa)$ as a model of operations. 
They are inspired by what Luo et al.~\cite{DBLP:conf/sigmod/LuoNEW04}
state informally as ``acceptable pacing''
and ``continuously revised estimates.''
The first, \e{continuity}, limits the ability of an operation to have a drastic
effect, and the second,
\e{progression}, states that
the measure is reactive and not indifferent to operations. 

\textbf{Bounded Continuity.~} \e{Continuity} means that, intuitively speaking, repairing operations
cannot cause disproportional changes to the database. More formally,
it is parameterized by a number $\delta\geq 1$ and it states that, for
every two databases $D_1$ and $D_2$,
and for each operation $o_1$ on $D_1$ we can find an operation $o_2$
on $D_2$ that is (almost) at least as impactful as $o_1$, as it
reduces inconsistency by at least $1/\delta$ of what $o_1$ does in
$D_1$.  This property is important in reliability estimation, for
instance, where one wishes to avoid a situation where the database is
deemed highly inconsistent and, yet, a small change can make it
considerably more consistent. It is also important in
progress indication, where it limits unexpected jumps and changes.
In what follows, we denote by $\Delta_{\I,\Sigma}(o_1,D_1)$ the value
$\I(\Sigma,D_1)-\I(\Sigma,o_1(D_1))$.

\postul{$\delta$-continuity}{For all $\Sigma$, $D_1$, $D_2$ and
  $o_1\in O$, there exists $o_2\in O$ such that 
  $\Delta_{\I,\Sigma}(o_2,D_2) \ge
  \Delta_{\I,\Sigma}(o_1,D_1)/\delta$.}  
  
This definition can be extended to the case where the
measure is aware of the cost of operations in the repair system
$\R$. There, the change is relative to the cost of the operation. That
is, we define the weighted version of $\delta$-continuity in the
following way.
   
\postul{Weighted $\delta$-continuity}{For all $\Sigma$, $D_1$, $D_2$
  and $o_1\in O$, there exists $o_2\in O$ such that 
  $\frac{\Delta_{\I,\Sigma}(o_2,D_2)}{\kappa(o_2,D_2)} \ge
  \frac{\Delta_{\I,\Sigma}(o_1,D_1)}{\delta\cdot\kappa(o_1,D_1)}$.}

%
  
We say that a measure $\I$ has \e{bounded continuity}, if there exists
$\delta>0$ such that $\I$ satisfies $\delta$-continuity. It is an easy observation that $\Imr$ satisfies
bounded continuity, and even bounded \e{weighted} continuity, for any set of anti-monotonic constraints.
\revtwo{We will later prove that
none of the other measures discussed so far satisfies
(unweighted) bounded continuity.}

\textbf{Progression.~} The last property we discuss is \e{progression} that states that, within the underlying
repair system $\R=(O,\kappa)$, there is always a way to progress
towards consistency, as we can find an
operation $o$ of $\R$ such that inconsistency reduces after applying
$o$. This property is particularly important for the task of progress indication in data repairing systems, as the combination of bounded continuity and progressions means that it is always possible to progress without significant slowdowns and long pauses, a behavior that 
users strongly averse towards~\cite{Harrison:2007:RPB:1294211.1294231}. \postul{Progression}{whenever $D\not\models\Sigma$, there is
  $o\in O$ such that $\I(\Sigma,o(D))<I(\Sigma,D)$.}
%
Clearly, the measure $\Id$ violates progression.  The
measure $\Imr$ satisfies progression \revone{for any set of anti-monotonic constraints}, since we can always remove a
fact from the minimum repair. The measure $\Imi$ satisfies
progression, since we can always remove a fact $f$ that participates
in one of the minimal inconsistent subsets and, by doing so, eliminate
all the subsets that include $f$. When we remove a fact $f$ that
appears in a minimal inconsistent subset, the measure $\Ip$ decreases
as well; hence, it satisfies progression. On the other hand, $\Imc$ and $\Imc'$
may violate progression even for functional dependencies, as illustrated in the following example.
\begin{example}\label{example:imc_progression}
  Consider again the database $D$ and the set $\Sigma_2$ from the
  proof of Proposition~\ref{prop:imc_monotonicity}. As explained there,  $\Imc(\Sigma_2,D)=\Imc'(\Sigma_2,D)=1$. The reader can
  easily verify that for every tuple deletion $o$, it is still the
  case that $\Imc(\Sigma_2,o(D))=\Imc'(\Sigma_2,o(D))=1$.
\end{example}

Note that there are some dependencies among the properties, as shown in
the following easy proposition (proof is in the Appendix).

\def\propdependencies{
  Suppose that the class $\C$ is realizable by the repair system $\R$,
  and let $\I$ be an inconsistency measure.
\begin{itemize}
\item If $\I$ satisfies progression, then $\I$ satisfies positivity.
\item If $\I$ satisfies positivity and bounded continuity, then $\I$
  satisfies progression.
  \end{itemize}
}

\begin{proposition}\label{prop:relationsips}
\propdependencies
\end{proposition}

\revtwo{Using Proposition~\ref{prop:relationsips}, we can now prove the following.}

\def\propcontinuity{
  In the case of $\Id$, $\Imi$, $\Ip$, $\Imc$, and $\Imc'$, bounded (unweighted)
  continuity can be violated already for the class $\Cfd$ of
  FDs and the system $\R_{\subseteq}$ of subset
  repairs.
}
\begin{proposition}\label{prop:bounded_continuity}
\propcontinuity
\end{proposition}
{\begin{proof}
Let $\Sigma=\set{A\rightarrow B}$ and let $D$ be a database that contains the following facts over $R(A,B,C)$:
\begin{gather*}
  f_0=R(0,0,0)\quad f_i=R(0,1,i)\quad f_j^k=R(j,k,0)
\end{gather*}
where $i,j\in\set{1,n}$ for some $n$ and $k\in\set{1,2}$. The fact $f_0$ violates the FD with every fact $f_i$, and for each $j$, the facts $f_j^1$ and $f_j^2$ jointly violate the FD. All the facts in the database participate in a violation of the FD; hence, $\Ip(\Sigma,D)=3n+1$. In addition,  $\Imi(\Sigma,D)=2n$. 

Let the operation $o_1$ be the deletion of $f_0$. Applying $o_1$, we significantly reduce inconsistency w.r.t.~these two measures, since none of the facts $f_i$ now participates in a violation; thus, $\Ip(\Sigma,o_1(D))=2n$ and $\Imi(\Sigma,o_1(D))=n$. However, every possible operation $o_2$ on the database $o_1(D)$ only slightly reduces inconsistency (by two in the case of $\Ip$ and by one in the case of $\Imi$). Therefore, $\Delta_{\Imi,\Sigma}(o_1,D)=n$ and $\Delta_{\Imi,\Sigma}(o_2,o_1(D))=1$, and the ratio between these two values depends on $|D|$. Similarly, it holds that $\Delta_{\Ip,\Sigma}(o_1,D)=n+1$ and $\Delta_{\Ip,\Sigma}(o_2,o_1(D))=2$, and again the ratio between these two values depends on $|D|$.


As for $\Id$, $\Imc$, and $\Imc'$, we use Proposition~\ref{prop:relationsips}. In the case of FDs,
each of the three measures satisfies
positivity but not progression  (Example~\ref{example:imc_progression}), and hence, they violate bounded continuity.
\end{proof}}

%

Table~\ref{table:postulates} summarizes the satisfaction of the
properties held by the different inconsistency measures we discussed
here, for the case of a system $\C$ of \revone{FDs or DCs} and
the repair system $\R_\subseteq$. 
The last column refers to computational complexity, discussed in Section~\ref{sec:complexity}, and 
the last row refers to another
measure, $\Ilmr$, introduced in Section~\ref{sec:lmr}.

%% file: update.tex
\section{Rational and Tractable Measures}\label{sec:rational}

In Section~\ref{sec:postulates}, we defined several properties of inconsistency measures, and we have shown that each of the measures we consider satisfies some and violates others, with the exception of $\Imr$ that satisfies all. \revtwo{Unfortunately, as we explain in Section~\ref{sec:complexity}, this measure is often intractable. This, in turn, raises the question whether there is any tractable inconsistency measure that satisfies all of the properties. We answer this question affirmatively, for the case where repairing operations are tuple deletions, by presenting a new measure in Section~\ref{sec:lmr}. In Section~\ref{sec:more-general}, we discuss the challenges in designing such a measure for more general repair systems, particularly for the case where repairing operations are attribute updates.}


\subsection{Computational Complexity}\label{sec:complexity}
We now discuss the
complexity of measuring inconsistency according to the aforementioned measures.
We focus on the class of DCs and the
special case of FDs. Moreover, we focus on \e{data complexity}, which
means that the set $\Sigma$ of constraints is fixed, and only the
database $D$ is given as input for the computation of $\I(\Sigma,D)$.

\revone{\textbf{Complexity of $\Id$, $\Imi$, $\Ip$.}~} The measure $\Id$ boils down to testing consistency, which is doable
in polynomial time (under data complexity). The measures $\Imi$ and $\Ip$ can be
computed by enumerating all the subsets of $D$ of a bounded size,
where this size is determined by $\Sigma$. Hence, $\Imi$ and $\Ip$ can also be
computed in polynomial time. 

\revone{\textbf{Complexity of $\Imc$, $\Imc'$.}~} \revone{The measures $\Imc$ and $\Imc'$ 
can be intractable to compute, already in the case of FDs}. When $\Sigma$ is a set of FDs, $\Imc(\Sigma,D)$ is the number of
maximal independent sets (minus one) of the \e{conflict graph} wherein
the tuples of $D$ are the nodes, and there is an edge between every
two tuples that violate an FD. Counting maximal independent sets is
generally \#P-complete, with several tractable classes of graphs such
as the \e{$P_4$-free} graphs that do not have any
induced subgraph that is a path of length
four. 
Under conventional
complexity assumptions, the finite sets $\Sigma$ of FDs for which
$\Imc(\Sigma,D)$ is computable in polynomial time are \e{precisely}
the sets $\Sigma$ of FDs that entail a $P_4$-free conflict graph for
every database $D$~\cite{DBLP:conf/pods/LivshitsK17}. Note that $\Imc'$ is equivalent to $\Imc$ in the case of FDs; hence, the same applies to $\Imc'$.

\revone{\textbf{Complexity of $\Imr$.}~} \revone{The measure $\Imr$ 
is also intractable even for FDs}. For $\C=\Cfd$ and $\R=\R_\subseteq$, the measure $\Imr(\Sigma,D)$ is
the size of the minimum \e{vertex cover} of the conflict graph. Again,
this is a hard (NP-hard) computational problem on general graphs. 
In a recent work, it has been shown that there is an efficient procedure that takes as input a set $\Sigma$ of FDs and
determines one of two outcomes: \e{(a)} $\Imr(\Sigma,D)$ can be
computed in polynomial time, \e{or
  (b)} $\Imr(\Sigma,D)$ is NP-hard to compute (and even approximate
beyond some
constant)~\cite{DBLP:conf/pods/LivshitsKR18}. 
There, they have also
studied the case where the repair system allows only to update cells
(and not delete or insert tuples). In both repair systems it is the
case that, if $\Sigma$ consists of a single FD per relation (which is
a commonly studied case, e.g., key
constraints~\cite{DBLP:journals/jcss/FuxmanM07,DBLP:journals/tods/KoutrisW17})
then $\Imr(\Sigma,D)$ can be computed in polynomial time.
Unfortunately, this is no longer true (under conventional complexity
assumptions) if we go beyond FDs to simple EGDs.

\begin{example}\label{example:egds}
Consider the following four EGDs.
{\small
\begin{align*}
  \sigma_1:\quad & \forall {x,y,z}[R(x,y), R(x,z) \Rightarrow (y=z)] \\
  \sigma_2:\quad & \forall {x,y,z}[R(x,y), R(y,z) \Rightarrow
  (x=z)] \\
  \sigma_3:\quad & \forall {x,y,z}[R(x,y), R(y,z) \Rightarrow
  (x=y)] \\
    \sigma_4: \quad & \forall {x,y,z}[R(x,y), S(y,z) \Rightarrow (x=z)] 
\end{align*}
}
Observe that $\sigma_1$ is an FD whereas $\sigma_2$, $\sigma_3$ and
$\sigma_4$ are not. The constraint $\sigma_2$ states that there are no
paths of length two except for two-node cycles, and $\sigma_3$ states
that there are no paths of length two except for single-node
cycles. Computing $\Imr(\Sigma,D)$ w.r.t.~$\Sigma=\set{\sigma_1}$ or
$\Sigma=\set{\sigma_4}$ can be done in polynomial time; however, the
problem becomes NP-hard for $\Sigma=\set{\sigma_2}$ and
$\Sigma=\set{\sigma_3}$.\qed
\end{example}

The next theorem fully classifies the complexity of computing $\Imr(\Sigma,D)$ for $\Sigma$ that consists of a single EGD with two binary atoms.

\def\thmegd{
Let $\R=\R_{\subseteq}$, and let $\Sigma$ be a set that consists of a
single EGD $\sigma$ with two binary atoms. If $\sigma$ is of the
following form:
\begin{align*}
  \forall x_1,x_2,x_3 [R(x_1,x_2), R(x_2,x_3) \Rightarrow (x_i=x_j)]
\end{align*}
then computing $\Imr(\Sigma,D)$ is NP-hard. In any other case, $\Imr(\Sigma,D)$ can be computed in
polynomial time.
}

\begin{theorem}\label{THM:EGD}
 \thmegd
\end{theorem}

The proof of the theorem is given in the Appendix. We prove the hardness side by reduction from MaxCut---the problem of finding a cut in a graph (i.e., a partition of the vertices into two disjoint subsets), such that the number of edges crossing the cut is the highest among all possible cuts~\cite{10.5555/574848}. For the tractable cases, we  provide efficient algorithms. Note that the EGDs $\sigma_2$ and $\sigma_3$ from
Example~\ref{example:egds} satisfy the condition of
Theorem~\ref{THM:EGD}; hence, computing $\Imr(\Sigma,D)$ w.r.t.~these
EGDs is indeed NP-hard. The EGDs $\sigma_1$ and $\sigma_4$ do not satisfy the condition of the theorem; thus,
computing $\Imr(\Sigma,D)$ w.r.t.~these EGDs can be done in polynomial
time.

\subsection{The Subset Repair System}\label{sec:lmr}
\revtwo{The discussion in the previous section shows that among the measures considered so far, the rational ones are intractable.}
We now propose a new measure that is both rational and tractable in the
case where $\C$ is the class $\Cdc$ of DCs and
$\R=\R_\subseteq$ (i.e., operations are tuple deletions). Recall that a DC has the form $\forall \vec
x\neg[\varphi(\vec x)\wedge \psi(\vec x)]$
where $\varphi(\vec x)$ is a conjunction of atomic formulas,
and $\psi(\vec x)$ is a conjunction of comparisons over $\vec x$. 
Recall that DCs generalize common classes of constraints such as FDs,
conditional FDs, and EGDs.

\begin{figure}
\small
 \hrule
\begin{align}
        \mbox{Minimize}: & \sum_{i\in\tids(D)}\!\!x_i\cdot\kappa(\del{i}{\cdot},D) \mbox{\, subj.~to:}\notag\\
        \forall E\in\MI_\Sigma(D): \quad & \sum_{i\in\tids(E)}\!\!\! x_i\,\,\geq 1\label{eq:ineq}\\
        \forall i\in\tids(D): \quad & x_i\in\set{0,1}\label{eq:integer}
\end{align} 
\hrule
\vskip-1em
\caption{\label{fig:ilp}ILP for $\Imr(\Sigma,D)$ under $\Cdc$ and $\R_\subseteq$.}
\end{figure}

Let $D$ be a database and $\Sigma$ a finite set of DCs. For
$\R=(O,\kappa)$, the measure $\Imr(\Sigma,D)$ is the result of the
Integer Linear Program (ILP) of Figure~\ref{fig:ilp} wherein each
$x_i$, for $i\in\tids(D)$, determines whether to delete the $i$th
tuple ($x_i=1$) or not ($x_i=0$). Denote by $\Ilmr(\Sigma,D)$ the
solution of the \e{linear relaxation} of this ILP, which is the Linear
Program (LP) obtained from the ILP by replacing the last constraint
(i.e., Equation~\eqref{eq:integer}) with
``$\forall i\in\tids(D): 
0\leq x_i\leq 1$.''

It is easy to see that the relative rankings of the inconsistency measure values of two databases under $\Ilmr$ and $\Imr$ are consistent with each other if they have sufficient separation under the first one.
More formally, 
for two databases $D_1, D_2$ we have that
$\Ilmr(\Sigma,D_1) \geq \mu \cdot \Ilmr(\Sigma,D_2)$ 
implies that $\Imr(\Sigma,D_1) \geq \Imr(\Sigma,D_2)$,
where $\mu$ is the integrality gap of the LP relaxation. 
The maximum number of tuples involved in a violation of a constraint in $\Sigma$ gives an upper bound on this integrality gap. 
In particular, for FDs (as well for the EGDs in Example~\ref{example:egds}), this number is 2; hence, 
$\Ilmr(\Sigma,D_1) \geq 2 \cdot \Ilmr(\Sigma,D_2)$ implies that $\Imr(\Sigma,D_1) \geq \Imr(\Sigma,D_2)$.

\begin{example}\label{eg:IRLIN}
Consider again the databases of Figure~\ref{fig:running} and the FDs of Example~\ref{eg:running}. In the LP of Figure~\ref{fig:ilp}, we define the variables $x_1,\dots,x_5$ corresponding to the facts $f_1,\dots,f_5$. Since $D_0$ is consistent, $MI_\Sigma(D_0)$ is empty, and we obtain a solution to the ILP by assigning $x_i=0$ for all $i\in\set{1,\dots,5}$.
For $D_1$, we have that $\MI_\Sigma(D_1) = \{\{t_2, t_3\}, \{t_2, t_4\},\{t_2, t_5\}, \{t_3, t_4\},\{t_3, t_5\}, \{t_4, t_5\},\{t_1, t_5\}\}$. For every pair $\{t_i, t_j\} \in \MI_\Sigma(D_1)$, the ILP contains a constraint $x_i + x_j \geq 1$. If we assign $0.5$ to all $x_i$, all the constraints are satisfied. Assuming unit cost for deletion, the total cost is $\Ilmr(\Sigma,D_1) = 2.5$. Another possible assignment is $x_1 = 0$, $x_2 = x_3 = x_4 = 0.5$, and $x_5 = 1$. Note that $\Ilmr(\Sigma,D_1) < \Imr(\Sigma,D_1) = 3$. 
However, for $D_2$, the optimal cost is $\Ilmr(\Sigma,D_2) = \Imr(\Sigma,D_2) = 2$, which can be obtained by assigning $0.5$ to $x_2, x_3, x_4, x_5$ or by $x_2 = x_4 = 1, x_1=x_3=x_5 = 0$.
\end{example}

The following theorem (\mrev{proof is in the Appendix}) shows that $\Ilmr$ satisfies all four properties and can be efficiently computed for the class $\Cdc$ of denial constraints and the repair system $\R_\subseteq$.

\def\thmlin{
  The following hold for $\C=\Cdc$ and $\R=\R_\subseteq$.
\begin{enumerate}
\item $\Ilmr$ satisfies positivity, monotonicity, progression and
  constant weighted continuity.
\item $\Ilmr$ can be computed in polynomial time (in data complexity).
\end{enumerate}
}

\begin{theorem}\label{thm:lin}
\thmlin
\end{theorem}

It thus appears from Theorem~\ref{thm:lin} that, for tuple deletions
and DCs, $\Ilmr$ is a desirable inconsistency measure---it satisfies the discussed properties and
avoids the inherent hardness of $\Imr$ (e.g., Theorem~\ref{THM:EGD}).

\added{
\subsection{More General Repair Systems}\label{sec:more-general}

In our analysis of the satisfaction of the properties and computational complexity of different inconsistency measures, we focused mostly on tuple deletions (i.e., where $\R=\R_\subseteq$) under DCs. While we have a good understanding of both aspects of inconsistency measures in this setting, the picture for other types of constraints and repairing operations is quite preliminary. As an example, consider the case of \e{update repairs}~\cite{DBLP:conf/pods/LivshitsKR18,DBLP:conf/icdt/KolahiL09}, where the allowed repairing operations are attribute updates. We again assume that the constraints are DCs. Recall that we denote an attribute update by $\upd{i.A}{c}{\cdot}$, where $i$ is a tuple identifier in $D$, $A$ is an attribute, and $c$ is the new value assigned to the fact $D[i]$ in attribute $A$. Here, we assume a countably infinite domain $\dom$ of possible attribute values.

We observe that none of the measures considered so far is a rational and tractable measure for update repairs. Clearly, the satisfaction of positivity and monotonicity does not depend on the repair system at hand; hence, the measures $\Imi$, $\Ip$, $\Imc$, and $\Imc'$ do not satisfy at least one of these properties also in the case of update repairs (see Table~\ref{table:postulates}). When considering progression, the measures no longer behave the same for both subset and update repairs. Clearly, $\Id$ still violates progression. As discussed in Section~\ref{sec:postulates}, the measures $\Imi$ and $\Ip$ satisfy progression in the case of tuple deletions. However, both measures violate progression when considering update repairs, even for the more restricted case of FDs, as updating a single value does not necessarily completely resolve a conflict between two facts in the database.

\begin{example}\label{example:updates1}
Consider a database $D$ over $R(A,B,C,D)$ consisting of the facts $R(0,0,0,0)$ and $R(0,1,0,1)$. Let $\Sigma=\set{A\rightarrow B, C\rightarrow D}$. The facts violate both FDs, and it is impossible to resolve both conflicts by updating a single value. Hence, $\Imi(\Sigma,D)=\Imi(\Sigma,o(D))=2$ and $\Ip(\Sigma,D)=\Ip(\Sigma,o(D))=2$ for any attribute update $o$.\qed
\end{example}

The above example shows that for update repairs, considering violations at the fact level is insufficient, as a deletion of a single fact resolves every conflict the fact is involved in, but this is not the case when updating a single value in the fact. Hence, one may suggest to consider a measure that counts the number of minimal violations instead of counting minimal inconsistent subsets of facts. Here, a minimal violation is a tuple $(F,\sigma)$, where $F\subseteq D$ and $\sigma$ is a constraint in $\Sigma$, such that $F\not\models\sigma$, but $F'\models\sigma$ for every $F'\subseteq F$. Note that Example~\ref{example:updates1} is not a counterexample for progression in this case, as we can always decrease the number of violations by updating a single attribute value. Nonetheless, the new measure still does not satisfy progression, as illustrated next.

\begin{example}\label{example:updates2}\em
Let $\Sigma=\set{A\rightarrow B, B\rightarrow C, D\rightarrow A}$. Let $D$ be a database that contains the following four facts over $R(A,B,C,D,E)$:
\begin{gather*}
  f_0=R(0,0,0,0,1)\quad f_1=R(0,0,0,0,2)\\ f_2=R(0,1,1,0,3)\quad f_3=R(0,1,1,0,4)
\end{gather*}
The minimal violations in $D$ are: $(\set{f_0,f_2},A\rightarrow B)$, $(\set{f_0,f_3},A\rightarrow B)$, $(\set{f_1,f_2},A\rightarrow B)$, and $(\set{f_1,f_3},A\rightarrow B)$. Clearly, updating values in the attributes $C,D$ or $E$ cannot decrease the number of violations, as these attributes are not involved in violations. Moreover, updating a value in attribute $B$ to a value that does not already belong to the domain of $B$ can only increase the number of violations. Hence, there are only two operations that might decrease the number of violations: \e{(1)} updating some value in attribute $A$ to a new value, or \e{(2)} updating a value in attribute $B$ to another value from the domain of $B$ (i.e., either $0$ or $1$). We now show that both operations actually increase the number of violations.

Suppose that we change the value of attribute $A$ in $f_0$ to $1$. This operation resolves the violations $(\set{f_0,f_2},A\rightarrow B)$ and $(\set{f_0,f_3},A\rightarrow B)$, but introduces new violations: $(\set{f_0,f_1},D\rightarrow A)$, $(\set{f_0,f_2},D\rightarrow A)$, and $(\set{f_0,f_3},D\rightarrow A)$. Hence, the total number of violations increases. Clearly, updating the value of attribute $A$ in one of $f_1$, $f_2$, or $f_3$ will similarly increase the number of violations.

Next, suppose that we change the value of attribute $B$ in $f_0$ to $1$. This operation again resolves the violations $(\set{f_0,f_2},A\rightarrow B)$ and $(\set{f_0,f_3},A\rightarrow B)$. However, it introduces the violations
$(\set{f_0,f_1},A\rightarrow B)$, $(\set{f_0,f_2},B\rightarrow C)$ and $(\set{f_0,f_3},B\rightarrow C)$, and the total number of violations increases. A similar argument applies to the case when the value of attribute $B$ is changed in $f_1$, $f_2$, or $f_3$. Therefore, while there exists a sequence of operations that decreases the number of violations, no individual operation does.\qed
\end{example}

Example~\ref{example:updates2} can be used to show that $\Imc$ and
$\Imc'$ violate progression as well. The measure $\Imr$ satisfies
progression, as we can always update an attribute value from the minimum repair.

Proposition~\ref{prop:relationsips} implies that none of the measures considered so far, except for $\Imr$, satisfies bounded continuity. The measure $\Imr$ satisfies bounded (weighted) continuity, since for any operation $o$ we have that $\frac{\Delta_{\I,\Sigma}(o,D_1)}{\delta\cdot\kappa(o,D_1)}$ is either $1$ (if $o$ belongs to a minimum repair) or $0$ (if no minimum repair contains $o$). We conclude that $\Imr$ again stands out among the measures as it satisfies every property. \mrev{Hence, in cases where $\Imr$ is tractable (e.g., when no two constraints share an attribute), this measure provides both rationality and tractability; however, computing $\Imr$ is often hard, even for simple FD sets~\cite{DBLP:conf/pods/LivshitsKR18}.}

  Recall that in the case of $\R_\subseteq$, the linear relaxation of $\Imr$ can be used to obtain a desired measure that is also efficient.
  This, however, is no longer the case for updates. We do not have any natural linear relaxation for attribute updates, and finding such a measure remains a
  challenging open problem for future research.

\mrev{
\noindent
\textbf{Remark.~}  We again stress that a chosen measure of inconsistency does not restrict the repair operations that a system allows or applies. Any repair system, supporting any type of operations whatsoever, could adopt measures such as $\Imr$ (i.e., the maximum size of a consistent subset) and $\Ilmr$. It is just that continuity and progression are guaranteed only for deletion operations, and could be violated otherwise. Indeed, our empirical experience, which we report in the next section, draws an optimistic picture: these measures work well in practical scenarios even for attribute updates. }}

%% file: experiments.tex
\section{Experimental Evaluation}\label{sec:experiments}

In this section, we present an experimental evaluation of the considered measures \mrev{under different error (and repair) models.}

\subsection{Setup}\label{sec:expt-setup}

\begin{figure*}[t]
\small
\centering
\mrevcolor
\begin{tabular}[b]{|c||c|c|c|c|}\hline
Dataset & \#Tuples & \#Atts. & \#DCs & Example constraint \\\hline\hline
\textbf{Stock} & 123K & $7$ & $6$ & \scriptsize{$\forall t \neg(t[\att{High}]<t[\att{Low}])$} \\\hline
\textbf{Hospital} & 115K & $15$ & $7$ & 
\scriptsize $\forall t,t' \neg(t[\att{State}]=t'[\att{State}],t[\att{Measure}]=t'[\att{Measure}],t[\att{StateAvg}]\neq t'[\att{StateAvg}])$\\\hline
\textbf{Food} & 200K & $17$ & $6$ & \scriptsize$\forall t,t' \neg(t[\att{Location}]=t'[\att{Location}],t[\att{City}]\neq t'[\att{City}])$\\\hline
\textbf{Airport} & 55K & $9$ & $6$ & \scriptsize$\forall t,t' \neg(t[\att{Country}]=t'[\att{Country}],t[\att{Continent}]\neq t'[\att{Continent}])$\\\hline
\textbf{Adult} & 32K & $15$ & $3$ & \scriptsize$\forall t,t' \neg(t[\att{Gain}]<t'[\att{Gain}], t[\att{Loss}]<t'[\att{Loss}])$\\\hline
\textbf{Flight} & 500K & $20$ & $13$ & \scriptsize$\forall t,t' \neg(t[\att{Origin}]=t'[\att{Origin}],t[\att{Dest}]=t'[\att{Dest}],t[\att{Distance}]\neq t'[\att{Distance}])$\\\hline
\textbf{Voter} & 950K &  $22$ & $5$ & \scriptsize$\forall t,t' \neg(t[\att{BirthYear}]<t'[\att{BirthYear}],t[\att{Age}]>t'[\att{Age}])$ \\\hline
\textbf{Tax} & 1M & $15$ & $9$ & \scriptsize$\forall t,t' \neg(t[\att{State}]=t'[\att{State}],t[\att{Salary}]>t'[\att{Salary}],t[\att{Rate}]<t'[\att{Rate}])$ \\\hline
\end{tabular}
\quad
\input{charts_overlap}
\vskip-0.5em
\caption{\label{fig:datasets} \mrev{The datasets used in our experiments. On the right:
the level of attribute overlapping of the constraints.}}
\end{figure*}

\paragraph*{Datasets} We evaluate the measures on datasets that were previously used for the problem of mining constraints~\cite{DBLP:journals/pvldb/ChuIP13,DBLP:journals/pvldb/BleifussKN17,DBLP:journals/pvldb/PenaAN19,DBLP:journals/pvldb/LivshitsHIK20}:
the real-world datasets
\textbf{Stock}, \textbf{Hospital}, \textbf{Food}, \textbf{Airport}, \textbf{Adult}, \textbf{Flight}, and \textbf{Voter}, and 
the synthetic \textbf{Tax} dataset.  
We use a DC mining algorithm~\cite{DBLP:journals/pvldb/LivshitsHIK20} to obtain a set of DCs for each dataset. All DCs are of the form $\forall t,t' \neg(P_1,\dots,P_m)$, where $t,t'$ are database tuples, each $P_i$ is a predicate $t[A]\rhorel t'[B]$, such that $A$ and $B$ are attributes of the schema, and $\rho$ is a comparison operator from $\set{=,\neq,>,<,\ge,\le}$. (Note that it may be the case the $t=t'$.) \mrev{More details about the datasets and constraints are given in Figure~\ref{fig:datasets}.}

\paragraph*{\revone{Measure implementations}}
We implemented all measures in Python 3 using Pandas. \revone{Using SQL, we materialize all conflicting pairs of tuples.} \revone{For example, the following DC over the Tax dataset:
{\small
\begin{align*}
    \forall t,t' \neg&(t[\att{St}]=t'[\att{St}],t[\att{Salary}]> t'[\att{Salary}],t[\att{Tax}]<t'[\att{Tax}])
\end{align*}}
(where $\att{St}$ stands for ``State'') will give rise to the query
{\small
\begin{align*}
    &\mbox{\textsf{SELECT DISTINCT }} R_1.\att{ID}, R_2.\att{ID}\quad\\
    &\mbox{\textsf{FROM }} R\mbox{ \textsf{AS} } R_1, R \mbox{ \textsf{AS} }R_2\\
    &\mbox{\textsf{WHERE }}R_1.\att{St}=R_2.\att{St},\mbox{ }R_1.\att{Salary}>R_2.\att{Salary},\mbox{ }R_1.\att{Tax}<R_2.\att{Tax}\,.
\end{align*}}}
For the measure $\Id$, we simply return $1$ if the query result is nonempty and $0$ otherwise. The measure $\Imi$ counts the tuples in the query result, and $\Ip$ counts the database facts that occur in these tuples. 
To compute  $\Imc$, we use a C++ implementation~\cite{parallelenum} of an algorithm for enumerating the maximal cliques~\cite{DBLP:conf/icalp/ConteGMV16}
over the complement of the conflict graph.
\revone{The conflict graph is also built from the result of the above SQL query: we add a vertex for each fact and an edge for each fact pair in the result.}
We use the Gurobi Optimizer~\cite{gurobi} to compute $\Imr$ and $\Ilmr$ using the LP of Figure~\ref{fig:ilp}. \revone{We dynamically construct the LP from the result of the SQL query (i.e., we add a corresponding constraint for each fact pair in the result)}. 

\mrev{\textbf{Noise Generation}.~}
Initially, all datasets are consistent w.r.t.~the given set of DCs. \mrev{We use two algorithms to add noise to these datasets. In the first, which we refer to as {\bf \algname{CONoise} (for Constraint-Oriented Noise)}, we introduce random violations of the constraints, by running several iterations of the following procedure:}
\begin{enumerate}
    \item Randomly select a constraint $\varphi$ from the set of constraints.
    \item Randomly select two tuples $t$ and $t'$ from the database.
    \item For every predicate $P=(t[A]\rhorel t'[B])$ of $\varphi$:
    \begin{itemize}
        \item If $t$ and $t'$ jointly satisfy $P$, continue to the next predicate.
        \item If $\rho\in\set{=,\le,\ge}$, change either $t[A]$ to $t[B]$ or vice versa (the choice is random).
        \item If $\rho\in\set{<,>,\neq}$, change either $t[A]$ or $t[B]$ (the choice is again random) to another value from the active domain of the attribute such that $P$ is satisfied, if such a value exists, or a random value in the appropriate range otherwise.
    \end{itemize}
\end{enumerate}

\mrev{The second algorithm, {\bf \algname{RNoise} (for Random Noise)}, has two parameters: $\alpha$ is used to control the level of noise (we modify $\alpha$ of the values in the dataset), and $\beta$, controls the data skewness, as we now explain. At each iteration of \algname{RNoise}, we randomly select a database cell corresponding to an attribute that occurs in at least one constraint. Then, we either change its value to another value from the active domain of the corresponding attribute (with probability $0.5$) or to a typo. For the first case, we use the Zipfian distribution, where the probability of selecting the $i$th value in the active domain is proportional to $i^{-\beta}$; hence, larger $\beta$ means larger skew.
}

\paragraph*{General setup.} All experiments were executed on a server with two Intel(R) Xeon(R) Gold 6130 CPUs (2.10GHz, 16 cores) with 512GB of RAM running Ubuntu 20.04. Each experiment was repeated five times and the average times are reported. The graphs of Figures~\ref{fig:results} and~\ref{fig:imc} were obtained in one execution and are representative of all five executions, where we observed a similar behavior.

\subsection{Results}

%

\def\heightp{}

\def\moreparams{%
  width=1.5in,
  height=1.2in,
  axis x line*=bottom,
  axis y line*=left,
  line width=0.2mm,
}

\begin{figure*}[t!]
\subfloat[\small \bf Noise added with
\algname{CONoise}.\label{fig:results_first}]{
  \small
  \includegraphics[width=6in]{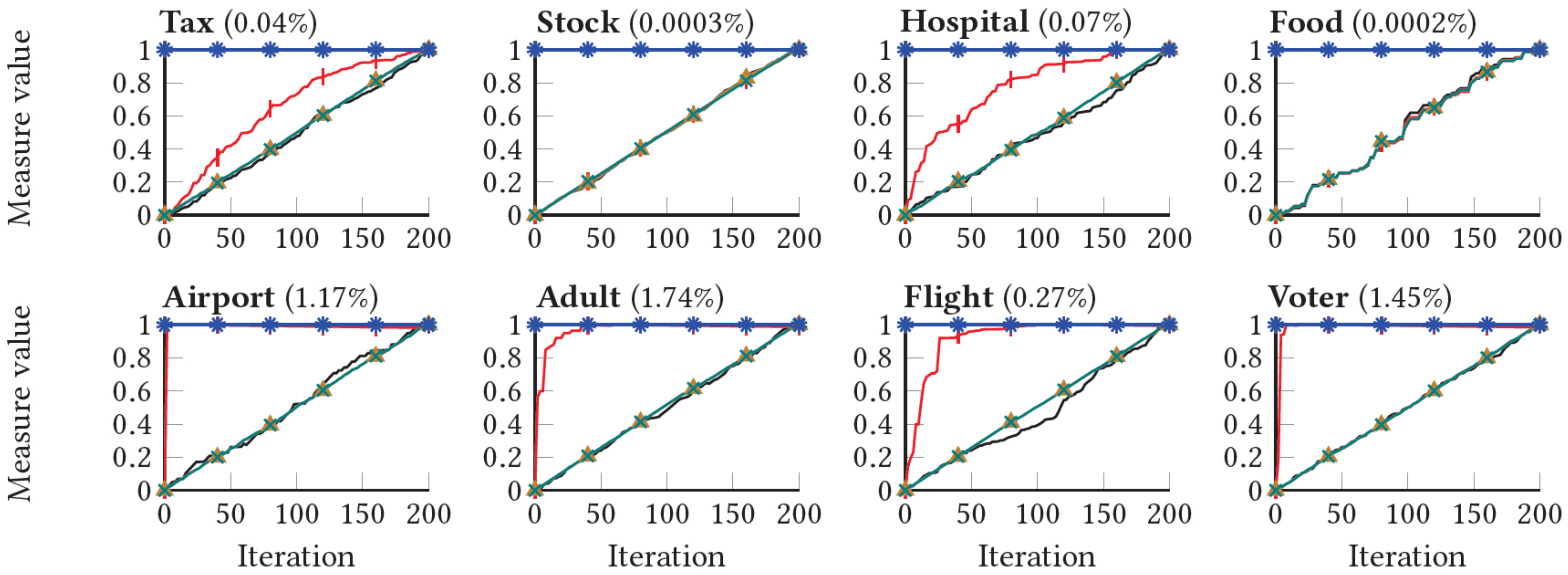}
}
\\
\subfloat[\small \bf Noise added with \algname{RNOise} ($\alpha=0.01$
and $\beta=0$).\label{fig:results_second}]{
  \includegraphics[width=6in]{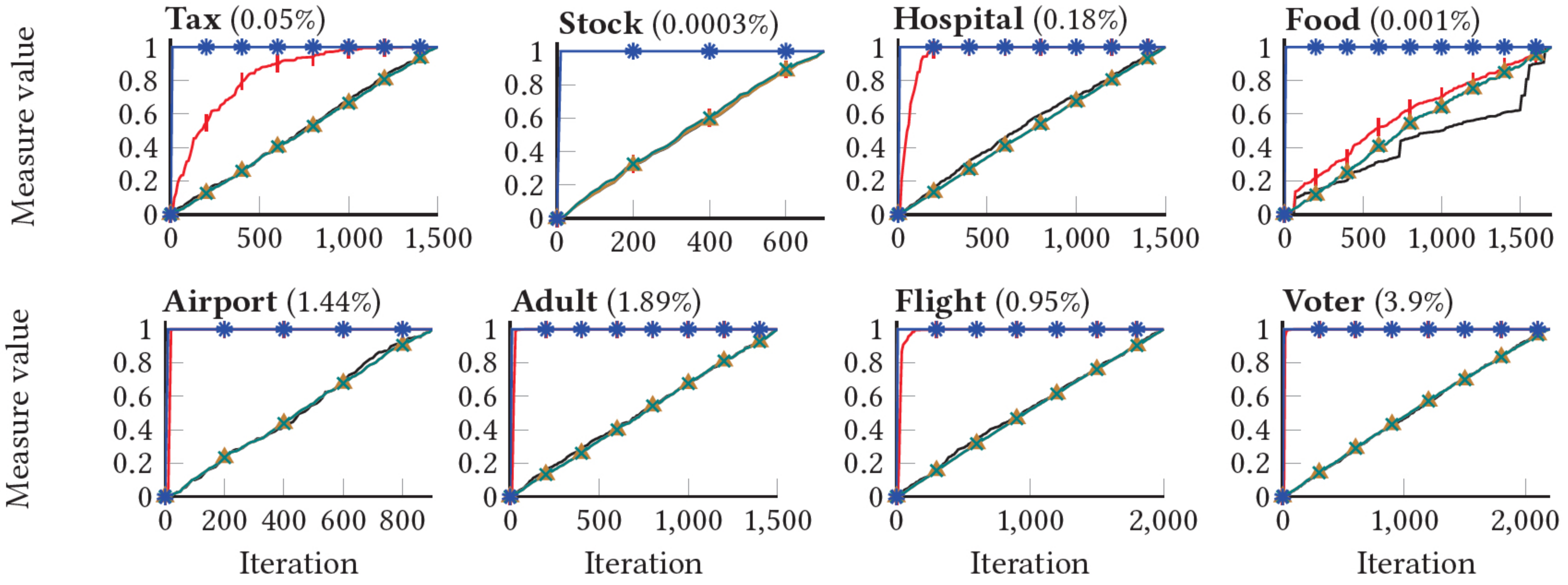}
}
\caption{The normalized values of $\Id$ (\ref{plot_id}), $\Imi$ (\ref{plot_imi}), $\Ip$ (\ref{plot_ip}), $\Imr$ (\ref{plot_ir}), and $\Ilmr$ (\ref{plot_ilinr}). Violation ratio in parentheses.}
    \label{fig:results}
\end{figure*}

\input{charts_imc}

\subsubsection{Measure Behavior}\label{sec:exp:behavior}

\mrev{We evaluate the behavior of the measures on samples of 10K tuples from each dataset. First, we run 200 iterations of \algname{CONoise} on each dataset and compute the measure values after each iteration; the results are in Figure~\ref{fig:results_first}. Then, we run \algname{RNoise} with $\beta=0$, $\beta=1$, and $\beta=2$ until we modify $1\%$ of the values in the dataset (hence, $\alpha=0.01$). As the number of iterations may be high, we compute the measure values every ten iterations. The results for $\beta=0$ are depicted in Figure~\ref{fig:results_second}, and the (similar) results for $\beta=1,2$ are in the Appendix, where we also test different probabilities for typos. We report the final violation ratio (i.e., percentage of violating tuple pairs out of all pairs) obtained in the experiments above each diagram (in parentheses).} \revtwo{Note that when we modify values in an iteration of \algname{CONoise} or \algname{RNoise}, we may introduce several violations at once, and resolve other violations at the same time. In our experiments, we observed that the number of newly introduced violations is usually significantly higher than the number of resolved ones, as evidenced by the behavior of $\Imi$ (that counts violations) in the charts of Figure~\ref{fig:results}: its value generally increases with the number of iterations.}

\mrev{\textbf{Variations with noise.~}} In general, we see that the measures may behave very differently on the same dataset. As expected, the drastic measure $\Id$ jumps from zero to one when we introduce the first violation, and stays at that point until the end of the execution. We see that the measure $\Ip$ often behaves in a similar way. For example, on the Airport dataset, it jumps from zero to its maximal value already in the first iteration. This behavior can be explained by observing the constraint set used for this dataset. For example, one of the DCs in this set is $\forall t,t' \neg(t[\att{Country}]=t'[\att{Country}],t[\att{Continent}]\neq t'[\att{Continent}])$, and all the tuples in the dataset initially agree on the value of the $\att{country}$ and $\att{continent}$ attributes. Hence, whenever we change the value of the $\att{continent}$ attribute for a single tuple, all the other tuples are immediately involved in a violation with it, and the value of $\Ip$ jumps to \#tuples.

Contrarily, the measure $\Imr$ is able to recognize, in this example, that the dataset contains a single erroneous tuple, and react to this small change in a more proportional way. In general, we see that the measures $\Imr$, $\Ilmr$, and $\Imi$ behave similarly in most cases and the corresponding graphs are generally monotonically increasing and close to being linear. \mrev{While this behavior seems to be very consistent for $\Imr$ and $\Ilmr$, the measure $\Imi$ is slightly less stable, as can be seen, for example, on the Food dataset in Figure~\ref{fig:results_second}.}



\mrev{Due to the high computational cost of $\Imc$ (as we report later on), we evaluate its behavior on a small sample of 100 tuples from each dataset. The results are in Figure~\ref{fig:imc}; missing graphs are due to timeout. The left chart is for \algname{CONoise} and the right is for \algname{RNoise} with $\beta=0$. We run both algorithms for 100 iterations.} Observe that this measure is the least stable of all, as we get very different graphs on the different datasets. In particular, on the Stock dataset it resembles a step function \mrev{and fails to indicate progress for long periods of time}. On the Airport dataset, we see a lot of jumps and jitters on these graphs. \mrev{This behavior may be affected, to some extent, by the small size of the datasets. However, the charts that we obtain for the other measures in this case (given in the Appendix), while also less stable, follow a similar trend as the ones in Figure~\ref{fig:results}.}

\mrev{\textbf{Error rate, data skew, and overlapping constraints.~} The experiments indicate that the behavior of the measures is largely stable across several properties of the data and constraints. 
We can see that the error rate, which increases with the number of iterations, affects the value of the measures, but has no evident impact on their behavior (i.e., the trend of the graph). 
Moreover, we obtain very similar charts in the experiments with $\beta=1,2$ and distinct typo probabilities (see the Appendix); hence, {\em data skew} and different {\em distributions of error types} also do not seem to affect the results. Finally, we examine how {\em overlap of dependencies} affects the results. For each dataset, and for each DC in its constraint set, we compute the ratio of DCs that overlap with it (i.e., the DCs share at least one attribute).
  Figure~\ref{fig:datasets} (right) shows the minimum, maximum, and average values for each dataset. We again see no clear correlation between the behavior of the measures and level of overlap.
}

\mrev{
\subsubsection{Case study: the HoloClean repair system} \label{sec:exp:holoclean}

Up to now, we described experiments with our synthetic noise generation models.
We have shown that the behavior of the measures is robust to the operations and is not sensitive to various parameters of the input such as data skew, error rate, and overlap of constraints. We now further strengthen this finding by showing similar results on a  cleaning system that we treat as a black box, namely the HoloClean system~\cite{DBLP:journals/pvldb/RekatsinasCIR17}, that
    uses soft rules and a statistical approach for automatic data
    cleaning. To accurately analyze the behavior of our measures, we need a dataset where the behavior of HoloClean is predictable; hence, we run the system on the (dirty) Hospital dataset provided in the HoloClean repository (\url{https://github.com/HoloClean}) with a set of 15 DCs. It has been shown that the accuracy of HoloClean on this dataset is very high~\cite{DBLP:journals/pvldb/RekatsinasCIR17}; thus, the system should significantly decrease the level of inconsistency in this dataset. 

Since HoloClean features one-shot automatic cleaning, we simulate a cleaning pipeline by providing it with \e{a single DC at a time}. That is, we first run HoloClean on the original dataset with a single DC; then on the resulting dataset after adding one more DC to the constraint set, and so on. Note that HoloClean uses soft constraints; hence, it does not necessarily eliminate all violations. 
We compute the measures after every step; the (normalized) values are in Figure~\ref{fig:holoclean}. We tried several random permutations of the DCs and obtained similar results. Again, we see that $\Id$ and $\Ip$ fall short of effectively indicating progress. Contrarily, the other measures, particularly $\Imr$ and $\Ilmr$, are able to capture the reduction in the inconsistency level, and show an almost linear decay as desired.}

\subsubsection{Running Times}\label{sec:exp:runningtime}

{
  \begin{table}[t]
  \mrevcolor
  \small
    \centering
    \caption{\mrev{Running Times in sec. (The $\Imc$ measure and the Voter dataset are excluded due to timeout.)}\label{table:runnningtime}}
    \vskip-1em
  \scalebox{0.95}{
\begin{tabular}{|c||c|c|c|c|c|c|}
\hline
 & $\Id$ & $\Imr$ & $\Imi$ & $\Ip$ & $\Ilmr$\\
\hline\hline
Tax & 8092.894 & 10102.15 & 8092.894 & 8275.692 & 8804.30 \\\hline
Stock & 1.16 & 1.16 & 1.16 & 1.28 & 1.16 \\\hline
Hospital & 199.08 & 212.59 & 199.08 & 200.19 & 207.91 \\\hline
Food & 89.38 & 89.92 & 89.38 & 91.25 & 89.81 \\\hline
Airport & 61.64 & 78.96 & 61.64 & 63.33 & 73.77 \\\hline
Adult & 119.19 & 240.96 & 119.19 & 132.30 & 179.31 \\\hline
Flight & 8084.05 & 8222.35 & 8084.05 & 8138.40 & 8157.30 \\\hline
\end{tabular}}
\vspace{-2em}
\end{table}
}

\begin{figure}[b]
\subfloat[\label{fig:runtime_tax}]{\input{charts_runtime}}\quad
\subfloat[\label{fig:runtime_voter}]{\input{charts_runtime2}}
\vskip-1em
\caption{\mrev{Scalability in $|D|$ on Tax~(a) and error rate on Voter (b): $\Id$ (\ref{plot_id_run}), $\Imi$ (\ref{plot_imi_run}), $\Ip$ (\ref{plot_ip_run}), $\Imr$ (\ref{plot_imr_run}), and $\Ilmr$ (\ref{plot_ilmr_run}).}}
\end{figure}

We now study the running times of the measures we discussed in the paper.  \mrev{We stress that our focus is not on optimizing these measures, but rather on understanding the execution cost obtained in reasonable implementations. } Table~\ref{table:runnningtime} shows the average running times of the measures on all datasets \mrev{after running $\#\mbox{tuples}/1000$ iterations of \algname{CONoise}} \revone{(with the number of tuples reported in Figure~\ref{fig:datasets}).} The computation of $\Imc$ exceeded our 24-hour limit on all datasets (and, in some cases, even on datasets with only one hundred tuples).  An immediate conclusion is that $\Imc$ is not only behaving oddly, but is prohibitively infeasible. \mrev{The Voter dataset (that we discuss next) is also excluded due to timeout.}

We can see that the running times of the measures are usually close to each other. \mrev{This is because the dominant part of the computation of $\Imr$ and $\Ilmr$ for large datasets is the evaluation of the SQL query that finds all violations of the constraints. In the case of the Voter dataset, the SQL engine reached the 24-hour limit. The domination of the SQL computation can also be seen in Figure~\ref{fig:runtime_tax} that depicts the running times of all the measures, except for $\Imc$ (again, due to timeout), on samples of the Tax dataset, consisting of 100K to 1M tuples. The figure shows a quadratic trend, which is consistent with the complexity of the dominating SQL part.

This is not the case, however, for smaller datasets, as can be seen in Figure~\ref{fig:runtime_voter}. This chart depicts the running times of the measures on a sample of 10K tuples from the Voter dataset. There, we add noise using \algname{RNoise} with $\alpha=0.01$ and $\beta=0$, and compute the running times every ten iterations. The evaluation of the SQL query is quite fast on these smaller datasets, and the computation of $\Imr$ and $\Ilmr$ is now dominated by the LP solver. We also see a more significant difference in running times between these measures. While the computation of $\Id$, $\Imi$, and $\Ip$ is only slightly affected by the change in error rate (that increases with the number of iterations), the computations of $\Imr$ significantly increases with the increased error rate. We provide similar charts for the other datasets in the Appendix.}

  
%

\begin{figure}[t]
\input{charts_holoclean}
\vskip-1em
      \caption{\mrev{HoloClean case study---normalized measures on
    Hospital: $\Id$ (\ref{plot_id}), $\Imi$ (\ref{plot_imi}), $\Ip$ (\ref{plot_ip}), $\Imr$ (\ref{plot_ir}), $\Ilmr$ (\ref{plot_ilinr}).}}\label{fig:holoclean}
\end{figure}

%% file: charts_overlap.tex
\pgfplotstableread{
x y y-max y-min
{Stock} 0.53 0.6 0
{Hospital} 0.38  0.5 0.167
{Food}  0.47  0.6 0.4
{Airport}  0.4 0.6 0
{Adult}  0.33 0.5 0
{Flight}  0.39 1 0
{Voter}  0.1 0.25 0
{Tax} 0.64 0.875 0.125
}{\differanser}
\begin{tikzpicture}[scale=0.9]
  \small
\begin{axis} [
width  = 2.1in,
axis x line*=bottom,
axis y line*=left,
height = 1.3in,
symbolic x coords={{Stock},{Hospital},{Food},{Airport},{Adult},{Flight},{Voter},{Tax}},
ymax=1,
xtick=data,
x tick label style={rotate=45,anchor=east, yshift=-5pt, font=\bfseries},
]
\addplot+[blue, very thick, forget plot,only marks] 
plot[very thick, error bars/.cd, y dir=plus, y explicit]
table[x=x,y=y,y error expr=\thisrow{y-max}-\thisrow{y}] {\differanser};
\addplot+[red, very thick, only marks,xticklabels=\empty] 
plot[very thick, error bars/.cd, y dir=minus, y explicit]
table[x=x,y=y,y error expr=\thisrow{y}-\thisrow{y-min}] {\differanser};
\end{axis} 
\end{tikzpicture}

%% file: charts_imc.tex
\begin{figure}[b]
\scalebox{1}{
\begin{tikzpicture}
\begin{axis}
[
ytick={0,0.2,0.4,0.6,0.8,1},
ymin=0,
ymax=1.05,
xmin=0, xmax=100,
line width=0.2mm,
ylabel = Measure value,
xlabel = Iteration,
width=0.53\linewidth,
ytick pos=left,
xtick pos=left,
xlabel style={yshift=+5pt},
ylabel style={yshift=-5pt},
axis x line*=bottom,
axis y line*=left,
height = 1.3in,
]
\addplot[color=black,mark repeat=7] coordinates{
(0, 0.000108506944444)
(1, 0.000108506944444)
(2, 0.000108506944444)
(3, 0.000217013888889)
(4, 0.000434027777778)
(5, 0.000434027777778)
(6, 0.000434027777778)
(7, 0.000868055555556)
(8, 0.000868055555556)
(9, 0.000868055555556)
(10, 0.000868055555556)
(11, 0.000868055555556)
(12, 0.00173611111111)
(13, 0.00173611111111)
(14, 0.00173611111111)
(15, 0.00173611111111)
(16, 0.00173611111111)
(17, 0.00173611111111)
(18, 0.00173611111111)
(19, 0.00173611111111)
(20, 0.00173611111111)
(21, 0.00347222222222)
(22, 0.00694444444444)
(23, 0.00694444444444)
(24, 0.00694444444444)
(25, 0.00694444444444)
(26, 0.00694444444444)
(27, 0.00694444444444)
(28, 0.00694444444444)
(29, 0.00694444444444)
(30, 0.00694444444444)
(31, 0.00694444444444)
(32, 0.00694444444444)
(33, 0.00694444444444)
(34, 0.00694444444444)
(35, 0.0138888888889)
(36, 0.0277777777778)
(37, 0.0277777777778)
(38, 0.0277777777778)
(39, 0.0277777777778)
(40, 0.0277777777778)
(41, 0.0416666666667)
(42, 0.0416666666667)
(43, 0.0416666666667)
(44, 0.0416666666667)
(45, 0.0416666666667)
(46, 0.0416666666667)
(47, 0.0416666666667)
(48, 0.0416666666667)
(49, 0.0416666666667)
(50, 0.0416666666667)
(51, 0.0416666666667)
(52, 0.0416666666667)
(53, 0.0416666666667)
(54, 0.0416666666667)
(55, 0.0416666666667)
(56, 0.0416666666667)
(57, 0.0416666666667)
(58, 0.0625)
(59, 0.0625)
(60, 0.125)
(61, 0.125)
(62, 0.125)
(63, 0.125)
(64, 0.25)
(65, 0.25)
(66, 0.25)
(67, 0.25)
(68, 0.25)
(69, 0.25)
(70, 0.25)
(71, 0.25)
(72, 0.25)
(73, 0.25)
(74, 0.25)
(75, 0.25)
(76, 0.25)
(77, 0.25)
(78, 0.25)
(79, 0.25)
(80, 0.25)
(81, 0.25)
(82, 0.25)
(83, 0.25)
(84, 0.25)
(85, 0.25)
(86, 0.5)
(87, 0.5)
(88, 0.5)
(89, 0.5)
(90, 0.5)
(91, 0.5)
(92, 0.5)
(93, 0.5)
(94, 0.5)
(95, 0.5)
(96, 1.0)
(97, 1.0)
(98, 1.0)
(99, 1.0)
(100, 1.0)
};\label{plot_stock}
\addplot[color=red,mark=|,mark repeat=7] coordinates{
(0, 1.21564289274e-05)
(1, 2.43128578549e-05)
(2, 7.29385735646e-05)
(3, 0.000145877147129)
(4, 0.000218815720694)
(5, 0.000328223581041)
(6, 0.00047410072817)
(7, 0.000717229306719)
(8, 0.00103329645883)
(9, 0.00139798932666)
(10, 0.00149524075807)
(11, 0.00162896147628)
(12, 0.00267441436404)
(13, 0.00285676079795)
(14, 0.00465591227921)
(15, 0.00577430374053)
(16, 0.00718444949612)
(17, 0.00867969025419)
(18, 0.0131775689573)
(19, 0.013906954693)
(20, 0.0143081168476)
(21, 0.0155480725982)
(22, 0.0156696368875)
(23, 0.0143202732765)
(24, 0.0108070653165)
(25, 0.0118160489175)
(26, 0.014247334703)
(27, 0.014077144698)
(28, 0.0154386647378)
(29, 0.0216627563487)
(30, 0.0266225793511)
(31, 0.0281907586827)
(32, 0.0394476118695)
(33, 0.0474343856749)
(34, 0.0561870145026)
(35, 0.0602837310512)
(36, 0.0764274686668)
(37, 0.0774729215546)
(38, 0.0774729215546)
(39, 0.0898481662027)
(40, 0.0684650077193)
(41, 0.0738503057342)
(42, 0.0876843218536)
(43, 0.0876843218536)
(44, 0.0881341097239)
(45, 0.088960746891)
(46, 0.092315921275)
(47, 0.092315921275)
(48, 0.0885474283075)
(49, 0.115097069085)
(50, 0.114622968357)
(51, 0.115182164087)
(52, 0.128116604466)
(53, 0.125369251529)
(54, 0.128651487339)
(55, 0.171490742879)
(56, 0.169813155687)
(57, 0.148417840775)
(58, 0.148429997204)
(59, 0.149852299389)
(60, 0.170968016435)
(61, 0.209576834709)
(62, 0.257011220384)
(63, 0.253291353132)
(64, 0.245863775057)
(65, 0.246143372923)
(66, 0.250082055895)
(67, 0.251334168075)
(68, 0.254579934598)
(69, 0.251929833092)
(70, 0.322558685161)
(71, 0.335067650527)
(72, 0.37980330898)
(73, 0.380131532561)
(74, 0.306342008971)
(75, 0.305478902518)
(76, 0.428453337548)
(77, 0.48123655195)
(78, 0.436744022076)
(79, 0.434689585587)
(80, 0.436707552789)
(81, 0.433546881268)
(82, 0.378295911793)
(83, 0.366710835025)
(84, 0.355709266846)
(85, 0.506983868419)
(86, 0.513135021456)
(87, 0.413586024969)
(88, 0.424295838854)
(89, 0.5102539478)
(90, 0.584699918552)
(91, 0.60205929906)
(92, 0.594935631709)
(93, 0.596382246751)
(94, 0.698350372595)
(95, 0.697864115437)
(96, 0.719928033941)
(97, 0.70995976222)
(98, 0.745334970399)
(99, 0.821640874777)
(100, 1.0)
};\label{plot_hospital}
\addplot[color=brown,mark=triangle*,mark repeat=7] coordinates{
(0, 2.67881060809e-06)
(1, 5.35762121618e-06)
(2, 8.03643182427e-06)
(3, 8.03643182427e-06)
(4, 1.33940530405e-05)
(5, 1.33940530405e-05)
(6, 2.41092954728e-05)
(7, 2.41092954728e-05)
(8, 2.41092954728e-05)
(9, 2.41092954728e-05)
(10, 2.41092954728e-05)
(11, 4.55397803375e-05)
(12, 8.8400750067e-05)
(13, 8.8400750067e-05)
(14, 0.000174122689526)
(15, 0.000174122689526)
(16, 0.000174122689526)
(17, 0.000174122689526)
(18, 0.000174122689526)
(19, 0.000174122689526)
(20, 0.000345566568444)
(21, 0.000688454326279)
(22, 0.00137690865256)
(23, 0.00137690865256)
(24, 0.0027484596839)
(25, 0.0027484596839)
(26, 0.0027484596839)
(27, 0.00549156174658)
(28, 0.00550495579962)
(29, 0.00551031342084)
(30, 0.0110072327886)
(31, 0.0110072327886)
(32, 0.0219796410394)
(33, 0.0439244575409)
(34, 0.0658692740423)
(35, 0.0658692740423)
(36, 0.0878248057862)
(37, 0.175604071792)
(38, 0.175604071792)
(39, 0.21951513528)
(40, 0.438963300295)
(41, 0.21951513528)
(42, 0.21951513528)
(43, 0.21960353603)
(44, 0.21960353603)
(45, 0.439051701045)
(46, 0.439051701045)
(47, 0.219624966515)
(48, 0.219624966515)
(49, 0.43907313153)
(50, 0.329284757568)
(51, 0.329327618537)
(52, 0.384253951246)
(53, 0.576356817573)
(54, 0.576356817573)
(55, 0.576356817573)
(56, 0.576528261452)
(57, 0.576533619073)
(58, 0.749413340477)
(59, 0.749413340477)
(60, 0.749413340477)
(61, 0.998971336726)
(62, 0.998971336726)
(63, 0.999314224484)
(64, 0.999314224484)
(65, 0.999314224484)
(66, 0.999314224484)
(67, 0.999314224484)
(68, 1.0)
(69, 1.0)
(70, 1.0)
(71, 0.596763996785)
(72, 0.596763996785)
(73, 0.596763996785)
(74, 0.596763996785)
(75, 0.596763996785)
(76, 0.511733190463)
(77, 0.511733190463)
(78, 0.512547548888)
(79, 0.512418965979)
(80, 0.512418965979)
(81, 0.514476292526)
(82, 0.514476292526)
(83, 0.257310474149)
(84, 0.384864720064)
(85, 0.384864720064)
(86, 0.384864720064)
(87, 0.384864720064)
(88, 0.384864720064)
(89, 0.384864720064)
(90, 0.384864720064)
(91, 0.38392177873)
(92, 0.384178944549)
(93, 0.32065898741)
(94, 0.575981784088)
(95, 0.575981784088)
(96, 0.575981784088)
(97, 0.576753281543)
(98, 0.385079024913)
(99, 0.385079024913)
(100, 0.385079024913)
};\label{plot_airport}
\addplot[color=teal,mark=x,mark repeat=7] coordinates{
(0, 7.4068587512e-05)
(1, 0.000296274350048)
(2, 0.000444411525072)
(3, 0.00074068587512)
(4, 0.000888823050144)
(5, 0.00103696022517)
(6, 0.0012591659877)
(7, 0.00251833197541)
(8, 0.00503666395082)
(9, 0.0079994074513)
(10, 0.0136286201022)
(11, 0.0159988149026)
(12, 0.0159988149026)
(13, 0.00955484778905)
(14, 0.00955484778905)
(15, 0.0101473964891)
(16, 0.0142952373898)
(17, 0.0142952373898)
(18, 0.0142952373898)
(19, 0.0142952373898)
(20, 0.0168135693652)
(21, 0.0113324938893)
(22, 0.0113324938893)
(23, 0.0113324938893)
(24, 0.0147396489149)
(25, 0.0168135693652)
(26, 0.0227390563662)
(27, 0.0249611139916)
(28, 0.0379971853937)
(29, 0.0510332567958)
(30, 0.0530331086586)
(31, 0.0586623213095)
(32, 0.0633286423228)
(33, 0.0636249166728)
(34, 0.0594030071847)
(35, 0.0662173172358)
(36, 0.0776238797126)
(37, 0.0776238797126)
(38, 0.0925116658025)
(39, 0.0941411747278)
(40, 0.0941411747278)
(41, 0.098363084216)
(42, 0.116139545219)
(43, 0.116509888156)
(44, 0.135323309384)
(45, 0.135323309384)
(46, 0.139989630398)
(47, 0.128434930746)
(48, 0.140656247685)
(49, 0.140656247685)
(50, 0.140656247685)
(51, 0.158803051626)
(52, 0.194726316569)
(53, 0.194726316569)
(54, 0.195392933857)
(55, 0.195392933857)
(56, 0.203392341308)
(57, 0.203392341308)
(58, 0.202207243908)
(59, 0.204947781646)
(60, 0.204947781646)
(61, 0.170950299978)
(62, 0.177172061329)
(63, 0.177172061329)
(64, 0.162506481001)
(65, 0.152359084512)
(66, 0.189763721206)
(67, 0.189763721206)
(68, 0.189763721206)
(69, 0.189985926968)
(70, 0.176505444041)
(71, 0.175838826754)
(72, 0.175838826754)
(73, 0.165543293089)
(74, 0.178505295904)
(75, 0.232649433375)
(76, 0.240500703652)
(77, 0.250648100141)
(78, 0.25961039923)
(79, 0.345826235094)
(80, 0.347826086957)
(81, 0.347826086957)
(82, 0.347752018369)
(83, 0.426486926894)
(84, 0.422487223169)
(85, 0.408710465891)
(86, 0.408710465891)
(87, 0.408784534479)
(88, 0.351381379157)
(89, 0.36975038886)
(90, 0.50581438412)
(91, 0.619361528776)
(92, 0.614991482112)
(93, 0.618398637138)
(94, 0.629731131027)
(95, 0.710243685653)
(96, 0.706095844752)
(97, 0.706095844752)
(98, 0.706095844752)
(99, 0.748833419747)
(100, 1.0)
};\label{plot_adult}
\addplot[color=violet,mark=diamond*,mark repeat=7] coordinates{
(0, 2.0133282329e-05)
(1, 4.0266564658e-05)
(2, 4.0266564658e-05)
(3, 6.03998469871e-05)
(4, 0.000120799693974)
(5, 0.000161066258632)
(6, 0.000322132517264)
(7, 0.000483198775896)
(8, 0.000805331293161)
(9, 0.00140932976303)
(10, 0.00144959632769)
(11, 0.00169119571564)
(12, 0.00245626044414)
(13, 0.00265759326743)
(14, 0.0027783929614)
(15, 0.00318105860798)
(16, 0.00352332440758)
(17, 0.00414745615978)
(18, 0.00414745615978)
(19, 0.00392599005416)
(20, 0.00603998469871)
(21, 0.00795264651996)
(22, 0.00962370895327)
(23, 0.0114155710806)
(24, 0.0119793029858)
(25, 0.0122611689384)
(26, 0.0128853006906)
(27, 0.0128853006906)
(28, 0.0125229016086)
(29, 0.0184622198957)
(30, 0.0193279510359)
(31, 0.0218043447623)
(32, 0.0237170065836)
(33, 0.0247438039824)
(34, 0.0243814049004)
(35, 0.0285892609072)
(36, 0.0372264390264)
(37, 0.0502325394109)
(38, 0.052930399243)
(39, 0.0576013207433)
(40, 0.0490245424712)
(41, 0.0633997060541)
(42, 0.073184481266)
(43, 0.0857879160039)
(44, 0.105578932533)
(45, 0.112343715396)
(46, 0.112343715396)
(47, 0.106001731462)
(48, 0.124443818076)
(49, 0.145966296885)
(50, 0.186353661237)
(51, 0.199460428034)
(52, 0.200849624514)
(53, 0.153657210735)
(54, 0.146389095814)
(55, 0.178642614105)
(56, 0.191427248384)
(57, 0.194970706074)
(58, 0.179447945399)
(59, 0.191809780749)
(60, 0.218425979988)
(61, 0.192876844712)
(62, 0.207171475166)
(63, 0.207171475166)
(64, 0.201956955042)
(65, 0.210070667821)
(66, 0.218144114035)
(67, 0.25138416316)
(68, 0.322494916346)
(69, 0.30610642453)
(70, 0.319615856973)
(71, 0.310636413054)
(72, 0.301938835088)
(73, 0.326722905635)
(74, 0.314602669673)
(75, 0.356419497071)
(76, 0.367492802352)
(77, 0.404497775272)
(78, 0.424550524472)
(79, 0.473715999919)
(80, 0.452576053474)
(81, 0.469669210171)
(82, 0.427087318045)
(83, 0.423865992873)
(84, 0.465481487447)
(85, 0.546799814774)
(86, 0.571100686545)
(87, 0.553020999013)
(88, 0.555396726328)
(89, 0.663653385411)
(90, 0.682759870342)
(91, 0.760736072802)
(92, 0.770561114579)
(93, 0.770561114579)
(94, 0.840443737543)
(95, 0.776701765689)
(96, 0.759910608226)
(97, 0.827518170287)
(98, 0.895266665324)
(99, 1.0)
(100, 0.994443214077)
};\label{plot_flight}
\addplot[color=blue,mark=10-pointed star,mark repeat=10] coordinates{
(0, 0.00131578947368)
(1, 0.00263157894737)
(2, 0.00394736842105)
(3, 0.00921052631579)
(4, 0.00921052631579)
(5, 0.0184210526316)
(6, 0.0197368421053)
(7, 0.025)
(8, 0.025)
(9, 0.0263157894737)
(10, 0.0368421052632)
(11, 0.0394736842105)
(12, 0.0526315789474)
(13, 0.0552631578947)
(14, 0.0552631578947)
(15, 0.0815789473684)
(16, 0.0815789473684)
(17, 0.126315789474)
(18, 0.142105263158)
(19, 0.142105263158)
(20, 0.142105263158)
(21, 0.143421052632)
(22, 0.151315789474)
(23, 0.119736842105)
(24, 0.15)
(25, 0.15)
(26, 0.153947368421)
(27, 0.153947368421)
(28, 0.181578947368)
(29, 0.181578947368)
(30, 0.181578947368)
(31, 0.185526315789)
(32, 0.185526315789)
(33, 0.294736842105)
(34, 0.294736842105)
(35, 0.294736842105)
(36, 0.355263157895)
(37, 0.377631578947)
(38, 0.377631578947)
(39, 0.486842105263)
(40, 0.505263157895)
(41, 0.547368421053)
(42, 0.557894736842)
(43, 0.564473684211)
(44, 0.502631578947)
(45, 0.502631578947)
(46, 0.502631578947)
(47, 0.502631578947)
(48, 0.502631578947)
(49, 0.502631578947)
(50, 0.502631578947)
(51, 0.525)
(52, 0.476315789474)
(53, 0.514473684211)
(54, 0.5)
(55, 0.476315789474)
(56, 0.476315789474)
(57, 0.476315789474)
(58, 0.476315789474)
(59, 0.476315789474)
(60, 0.522368421053)
(61, 0.598684210526)
(62, 0.584210526316)
(63, 0.584210526316)
(64, 0.627631578947)
(65, 0.669736842105)
(66, 0.731578947368)
(67, 0.601315789474)
(68, 0.678947368421)
(69, 0.681578947368)
(70, 0.681578947368)
(71, 0.709210526316)
(72, 0.725)
(73, 0.717105263158)
(74, 0.717105263158)
(75, 0.717105263158)
(76, 0.702631578947)
(77, 0.706578947368)
(78, 0.806578947368)
(79, 0.809210526316)
(80, 0.798684210526)
(81, 0.832894736842)
(82, 0.796052631579)
(83, 0.796052631579)
(84, 0.809210526316)
(85, 0.817105263158)
(86, 0.818421052632)
(87, 0.815789473684)
(88, 0.805263157895)
(89, 0.805263157895)
(90, 0.977631578947)
(91, 0.986842105263)
(92, 0.986842105263)
(93, 0.986842105263)
(94, 0.972368421053)
(95, 0.953947368421)
(96, 0.953947368421)
(97, 1.0)
(98, 0.936842105263)
(99, 0.936842105263)
(100, 0.936842105263)
};\label{plot_voters}
\end{axis}
\end{tikzpicture}\quad
\begin{tikzpicture}
\begin{axis}
[
ytick={0,0.2,0.4,0.6,0.8,1},
ymin=0,
ymax=1.05,
xmin=0, xmax=100,
line width=0.2mm,
xlabel = Iteration,
width=0.53\linewidth,
ytick pos=left,
xtick pos=left,
xlabel style={yshift=+5pt},
axis x line*=bottom,
axis y line*=left,
height = 1.3in
]
\addplot[color=black,mark repeat=7] coordinates{
(0 , 0.3333333333333333)
(1 , 0.3333333333333333)
(2 , 0.3333333333333333)
(3 , 0.3333333333333333)
(4 , 0.3333333333333333)
(5 , 0.3333333333333333)
(6 , 0.3333333333333333)
(7 , 0.3333333333333333)
(8 , 0.3333333333333333)
(9 , 0.3333333333333333)
(10 , 0.3333333333333333)
(11 , 0.3333333333333333)
(12 , 0.3333333333333333)
(13 , 0.3333333333333333)
(14 , 0.3333333333333333)
(15 , 0.3333333333333333)
(16 , 0.3333333333333333)
(17 , 0.3333333333333333)
(18 , 0.3333333333333333)
(19 , 0.3333333333333333)
(20 , 0.3333333333333333)
(21 , 0.3333333333333333)
(22 , 0.3333333333333333)
(23 , 0.3333333333333333)
(24 , 0.3333333333333333)
(25 , 0.3333333333333333)
(26 , 0.3333333333333333)
(27 , 0.3333333333333333)
(28 , 0.3333333333333333)
(29 , 0.3333333333333333)
(30 , 0.3333333333333333)
(31 , 0.3333333333333333)
(32 , 0.3333333333333333)
(33 , 0.3333333333333333)
(34 , 0.3333333333333333)
(35 , 0.3333333333333333)
(36 , 0.3333333333333333)
(37 , 0.3333333333333333)
(38 , 0.3333333333333333)
(39 , 0.3333333333333333)
(40 , 0.3333333333333333)
(41 , 0.3333333333333333)
(42 , 0.3333333333333333)
(43 , 0.3333333333333333)
(44 , 0.3333333333333333)
(45 , 0.3333333333333333)
(46 , 0.3333333333333333)
(47 , 0.3333333333333333)
(48 , 0.3333333333333333)
(49 , 0.3333333333333333)
(50 , 0.3333333333333333)
(51 , 0.3333333333333333)
(52 , 0.3333333333333333)
(53 , 0.3333333333333333)
(54 , 0.3333333333333333)
(55 , 0.3333333333333333)
(56 , 0.3333333333333333)
(57 , 0.3333333333333333)
(58 , 0.3333333333333333)
(59 , 0.3333333333333333)
(60 , 0.3333333333333333)
(61 , 0.3333333333333333)
(62 , 0.3333333333333333)
(63 , 0.3333333333333333)
(64 , 0.3333333333333333)
(65 , 0.3333333333333333)
(66 , 0.3333333333333333)
(67 , 0.3333333333333333)
(68 , 0.3333333333333333)
(69 , 0.3333333333333333)
(70 , 0.3333333333333333)
(71 , 0.3333333333333333)
(72 , 0.3333333333333333)
(73 , 1.0)
(74 , 1.0)
(75 , 1.0)
(76 , 1.0)
(77 , 1.0)
(78 , 1.0)
(79 , 1.0)
(80 , 1.0)
(81 , 1.0)
(82 , 1.0)
(83 , 1.0)
(84 , 1.0)
(85 , 1.0)
(86 , 1.0)
(87 , 1.0)
(88 , 1.0)
(89 , 1.0)
(90 , 1.0)
(91 , 1.0)
(92 , 1.0)
(93 , 1.0)
(94 , 1.0)
(95 , 1.0)
(96 , 1.0)
(97 , 1.0)
(98 , 1.0)
(99 , 1.0)
(100 , 1.0)
(101 , 1.0)
};
\addplot[color=red,mark=|,mark repeat=7] coordinates{
(0 , 2.3722340874386365e-09)
(1 , 2.3722340874386365e-09)
(2 , 2.3722340874386365e-09)
(3 , 2.3722340874386365e-09)
(4 , 1.4233404524631818e-08)
(5 , 1.4233404524631818e-08)
(6 , 1.4233404524631818e-08)
(7 , 1.4233404524631818e-08)
(8 , 7.116702262315909e-08)
(9 , 7.8283724885475e-07)
(10 , 7.8283724885475e-07)
(11 , 1.810963502350655e-05)
(12 , 3.470578469922725e-05)
(13 , 3.493351917162136e-05)
(14 , 5.170046970163764e-05)
(15 , 5.170046970163764e-05)
(16 , 6.931430780086952e-05)
(17 , 6.931430780086952e-05)
(18 , 0.00012474630172204813)
(19 , 0.00018359668496322582)
(20 , 0.00026172858686710474)
(21 , 0.00026422654936117764)
(22 , 0.0003822333340408126)
(23 , 0.0003822333340408126)
(24 , 0.00038032843006859937)
(25 , 0.0018056449535266176)
(26 , 0.0030533167273467934)
(27 , 0.0028530005365353)
(28 , 0.00512719730584236)
(29 , 0.0052974051016160826)
(30 , 0.007209487454177896)
(31 , 0.007357555189213554)
(32 , 0.007353809431589488)
(33 , 0.01059383284278814)
(34 , 0.010532942338231766)
(35 , 0.053344854926368354)
(36 , 0.05548403124812067)
(37 , 0.03916275470834557)
(38 , 0.04078295737878894)
(39 , 0.03594515839820957)
(40 , 0.035648446475255)
(41 , 0.045904527745375893)
(42 , 0.045904527745375893)
(43 , 0.04517989563209054)
(44 , 0.04517989563209054)
(45 , 0.0426882883531379)
(46 , 0.03448821050542964)
(47 , 0.03448821050542964)
(48 , 0.04229324021055675)
(49 , 0.04319627856062201)
(50 , 0.04319627856062201)
(51 , 0.05257370330262833)
(52 , 0.0840539186129517)
(53 , 0.08395902213275189)
(54 , 0.08486953302019258)
(55 , 0.08897478848680498)
(56 , 0.11419068082594046)
(57 , 0.10914672046975814)
(58 , 0.10914672046975814)
(59 , 0.18533822977947575)
(60 , 0.245700950207141)
(61 , 0.2728188737430098)
(62 , 0.2425982838940732)
(63 , 0.3211547113213038)
(64 , 0.36903890162392444)
(65 , 0.3630341271564773)
(66 , 0.4175673859602253)
(67 , 0.5021871749201605)
(68 , 0.5385078223909003)
(69 , 0.5385078223909003)
(70 , 0.5974970999378976)
(71 , 0.6238844133305806)
(72 , 0.6203105881164703)
(73 , 0.6203105881164703)
(74 , 0.6169408770399455)
(75 , 0.749763765219812)
(76 , 0.9122571008136314)
(77 , 0.9649740467271616)
(78 , 0.9798187071813378)
(79 , 0.982526885199069)
(80 , 0.9228966205127096)
(81 , 0.9228966205127096)
(82 , 0.9222975887054178)
(83 , 0.7633531129341725)
(84 , 0.7633531129341725)
(85 , 0.7597751339381751)
(86 , 0.7897178474028114)
(87 , 0.7897178474028114)
(88 , 0.7634878487136367)
(89 , 0.7564041062717317)
(90 , 0.7564041062717317)
(91 , 0.7564041062717317)
(92 , 0.6599175626355614)
(93 , 0.7516474068578995)
(94 , 0.8697956312505299)
(95 , 0.8445672633323286)
(96 , 0.8445672633323286)
(97 , 0.9245316804852426)
(98 , 0.9245316804852426)
(99 , 1.0)
(100 , 1.0)
};
\addplot[color=brown,mark=triangle*,mark repeat=7] coordinates{
(0 , 0.00025581990278843696)
(1 , 0.00025581990278843696)
(2 , 0.00025581990278843696)
(3 , 0.02532617037605526)
(4 , 0.02532617037605526)
(5 , 0.02532617037605526)
(6 , 0.02532617037605526)
(7 , 0.02532617037605526)
(8 , 0.026093630084420567)
(9 , 0.026093630084420567)
(10 , 0.0012790995139421847)
(11 , 0.0012790995139421847)
(12 , 0.0012790995139421847)
(13 , 0.0012790995139421847)
(14 , 0.0012790995139421847)
(15 , 0.0012790995139421847)
(16 , 0.0012790995139421847)
(17 , 0.0012790995139421847)
(18 , 0.0012790995139421847)
(19 , 0.0012790995139421847)
(20 , 0.0012790995139421847)
(21 , 0.0012790995139421847)
(22 , 0.0012790995139421847)
(23 , 0.0012790995139421847)
(24 , 0.0012790995139421847)
(25 , 0.0012790995139421847)
(26 , 0.0012790995139421847)
(27 , 0.0012790995139421847)
(28 , 0.0012790995139421847)
(29 , 0.0012790995139421847)
(30 , 0.0012790995139421847)
(31 , 0.0012790995139421847)
(32 , 0.0012790995139421847)
(33 , 0.0012790995139421847)
(34 , 0.0012790995139421847)
(35 , 0.0012790995139421847)
(36 , 0.0012790995139421847)
(37 , 0.0012790995139421847)
(38 , 0.0012790995139421847)
(39 , 0.0012790995139421847)
(40 , 0.0012790995139421847)
(41 , 0.0010232796111537478)
(42 , 0.0010232796111537478)
(43 , 0.0010232796111537478)
(44 , 0.0010232796111537478)
(45 , 0.02507035047326682)
(46 , 0.02507035047326682)
(47 , 0.02507035047326682)
(48 , 0.07495523151701203)
(49 , 0.07495523151701203)
(50 , 0.07495523151701203)
(51 , 0.07495523151701203)
(52 , 0.0844205679201842)
(53 , 0.0844205679201842)
(54 , 0.0844205679201842)
(55 , 0.0844205679201842)
(56 , 0.0844205679201842)
(57 , 0.0844205679201842)
(58 , 0.0844205679201842)
(59 , 0.0844205679201842)
(60 , 0.0844205679201842)
(61 , 0.0844205679201842)
(62 , 0.08365310821181889)
(63 , 0.08365310821181889)
(64 , 0.08365310821181889)
(65 , 0.08365310821181889)
(66 , 0.08876950626758762)
(67 , 0.1903300076745971)
(68 , 0.1903300076745971)
(69 , 0.1903300076745971)
(70 , 0.1903300076745971)
(71 , 0.1903300076745971)
(72 , 0.1903300076745971)
(73 , 0.9263238679969301)
(74 , 0.9263238679969301)
(75 , 0.9263238679969301)
(76 , 0.9263238679969301)
(77 , 0.9923254029163469)
(78 , 0.9923254029163469)
(79 , 0.9923254029163469)
(80 , 0.9923254029163469)
(81 , 0.9923254029163469)
(82 , 0.9923254029163469)
(83 , 0.9923254029163469)
(84 , 0.9923254029163469)
(85 , 0.9923254029163469)
(86 , 0.9923254029163469)
(87 , 0.9355333844973139)
(88 , 1.0)
(89 , 1.0)
(90 , 1.0)
(91 , 1.0)
(92 , 1.0)
(93 , 0.8158096699923254)
(94 , 0.3553338449731389)
(95 , 0.3553338449731389)
(96 , 0.3553338449731389)
(97 , 0.3553338449731389)
(98 , 0.3553338449731389)
(99 , 0.3553338449731389)
(100 , 0.3553338449731389)
};
\addplot[color=teal,mark=x,mark repeat=7] coordinates{
(0 , 5.40087656021339e-10)
(1 , 5.40087656021339e-10)
(2 , 5.40087656021339e-10)
(3 , 5.40087656021339e-10)
(4 , 5.40087656021339e-10)
(5 , 5.40087656021339e-10)
(6 , 1.8362980304725523e-08)
(7 , 5.508894091417657e-08)
(8 , 9.916009364551783e-07)
(9 , 9.916009364551783e-07)
(10 , 1.0391286501850563e-06)
(11 , 2.0307295866402344e-06)
(12 , 2.0307295866402344e-06)
(13 , 2.030189498984213e-06)
(14 , 2.030189498984213e-06)
(15 , 4.9627034448832775e-05)
(16 , 9.629060792907645e-05)
(17 , 9.629006784142043e-05)
(18 , 0.00014181297619214706)
(19 , 0.000141811355929179)
(20 , 0.0005668889658637418)
(21 , 0.0008159233840551813)
(22 , 0.0015990936922986045)
(23 , 0.0015990936922986045)
(24 , 0.0015990920720356364)
(25 , 0.0015992017098298088)
(26 , 0.002082500228995816)
(27 , 0.0023288774159196304)
(28 , 0.0023288774159196304)
(29 , 0.0023288612132899497)
(30 , 0.002854638706777347)
(31 , 0.003757872661304938)
(32 , 0.003728553462810164)
(33 , 0.003730086231577952)
(34 , 0.004377984025143647)
(35 , 0.004174587553973667)
(36 , 0.004737265996471066)
(37 , 0.004737265996471066)
(38 , 0.004737265996471066)
(39 , 0.0047378946585026754)
(40 , 0.0030145727040043224)
(41 , 0.003974952253240219)
(42 , 0.004592192511099518)
(43 , 0.006669882719430809)
(44 , 0.006670131159752578)
(45 , 0.006670062028532607)
(46 , 0.007524636245603299)
(47 , 0.007524636245603299)
(48 , 0.007524636245603299)
(49 , 0.013763318381343231)
(50 , 0.014054225795506005)
(51 , 0.014054225795506005)
(52 , 0.01521439673080794)
(53 , 0.02043552347769493)
(54 , 0.013893839204823283)
(55 , 0.011816668560817086)
(56 , 0.011817208108385453)
(57 , 0.011838937455050158)
(58 , 0.011839370605350287)
(59 , 0.011002135362388504)
(60 , 0.016501365125289314)
(61 , 0.016501365125289314)
(62 , 0.018471615696208282)
(63 , 0.018473070152265948)
(64 , 0.024737624647079924)
(65 , 0.02473849256794315)
(66 , 0.02473980282059666)
(67 , 0.02474097157028429)
(68 , 0.02474066047979442)
(69 , 0.01959267078302287)
(70 , 0.024697652759657785)
(71 , 0.12345154215838072)
(72 , 0.12345289075725781)
(73 , 0.1418478010442073)
(74 , 0.13734024134898165)
(75 , 0.13199085458554272)
(76 , 0.13186692607199205)
(77 , 0.1318674753411382)
(78 , 0.1407871931079423)
(79 , 0.1407871931079423)
(80 , 0.12152007150142706)
(81 , 0.12152007150142706)
(82 , 0.1722802638613582)
(83 , 0.19382234323265418)
(84 , 0.24467635962018933)
(85 , 0.28265380324485817)
(86 , 0.3313004252131179)
(87 , 0.3719021693407612)
(88 , 0.41397152984198415)
(89 , 0.43609720633087057)
(90 , 0.43609733973252157)
(91 , 0.43609733973252157)
(92 , 0.43609730138629804)
(93 , 0.43609730138629804)
(94 , 0.31516913679113695)
(95 , 0.9678238982109065)
(96 , 1.0)
(97 , 1.0)
(98 , 0.9999999708352666)
(99 , 0.9745654116759802)
(100 , 0.9931661067017146)
};
\addplot[color=blue,mark=10-pointed star,mark repeat=10] coordinates{
(0 , 6.690131634361025e-08)
(1 , 6.690131634361025e-08)
(2 , 5.887315838237702e-06)
(3 , 5.820414521894091e-06)
(4 , 6.890835583391855e-06)
(5 , 1.2577447472598726e-05)
(6 , 2.167602649532972e-05)
(7 , 2.167602649532972e-05)
(8 , 2.167602649532972e-05)
(9 , 2.167602649532972e-05)
(10 , 2.749644101722381e-05)
(11 , 2.7362638384536592e-05)
(12 , 4.7232329338588835e-05)
(13 , 4.7232329338588835e-05)
(14 , 6.964427031369826e-05)
(15 , 7.31900400799096e-05)
(16 , 0.0002520841599827234)
(17 , 0.0006139533800853112)
(18 , 0.0006549638870039443)
(19 , 0.0006549638870039443)
(20 , 0.0006549638870039443)
(21 , 0.0006549638870039443)
(22 , 0.0006549638870039443)
(23 , 0.0006549638870039443)
(24 , 0.0018269411467113086)
(25 , 0.0018269411467113086)
(26 , 0.0018269411467113086)
(27 , 0.0018579164561784002)
(28 , 0.0018565115285351843)
(29 , 0.0018565115285351843)
(30 , 0.0018571136403822768)
(31 , 0.0018571136403822768)
(32 , 0.001882469239276505)
(33 , 0.001882469239276505)
(34 , 0.001834233390192762)
(35 , 0.006305181460119891)
(36 , 0.006342913802537688)
(37 , 0.010017200997445106)
(38 , 0.010017200997445106)
(39 , 0.010006764392095503)
(40 , 0.02691827984138233)
(41 , 0.02859676696712717)
(42 , 0.028310897642390923)
(43 , 0.028310897642390923)
(44 , 0.05146256647231165)
(45 , 0.05059539160986578)
(46 , 0.05059565921513115)
(47 , 0.05059164513615053)
(48 , 0.04758001547962658)
(49 , 0.04427790030753866)
(50 , 0.041378464158522935)
(51 , 0.04147078797507712)
(52 , 0.03874235159063566)
(53 , 0.03843226398938303)
(54 , 0.035598324229067696)
(55 , 0.035598324229067696)
(56 , 0.08072627266206828)
(57 , 0.08072627266206828)
(58 , 0.09187497562283285)
(59 , 0.09187497562283285)
(60 , 0.09187497562283285)
(61 , 0.20805205511143116)
(62 , 0.20185378195482837)
(63 , 0.20185378195482837)
(64 , 0.20464731332007216)
(65 , 0.20464731332007216)
(66 , 0.32443880332545033)
(67 , 0.32443880332545033)
(68 , 0.3249667885140341)
(69 , 0.43910344475546864)
(70 , 0.43910344475546864)
(71 , 0.44516610894384295)
(72 , 0.36353526358349275)
(73 , 0.36236188139614217)
(74 , 0.4330216467906201)
(75 , 0.43292818565168806)
(76 , 0.42055431888472294)
(77 , 0.6870977265661582)
(78 , 0.6618216064357461)
(79 , 0.6692600961947267)
(80 , 0.8093117801376548)
(81 , 0.7360154338660752)
(82 , 0.8262788905808232)
(83 , 0.8194315408530547)
(84 , 0.787174870474034)
(85 , 0.787174870474034)
(86 , 0.9952419783816424)
(87 , 0.9950296336035678)
(88 , 1.0)
(89 , 0.9988676952208844)
(90 , 0.996461321773321)
(91 , 0.9896825458948048)
(92 , 0.9923042077783618)
(93 , 0.9923042077783618)
(94 , 0.9492370874940098)
(95 , 0.9488882640305942)
(96 , 0.9285492603423916)
(97 , 0.9285492603423916)
(98 , 0.9741669933123437)
(99 , 0.987904576511657)
(100 , 0.8957027879584856)
};
\addplot[color=gray,mark=square,mark repeat=20] coordinates{
(0 , 0.027777777777777776)
(1 , 0.027777777777777776)
(2 , 0.027777777777777776)
(3 , 0.027777777777777776)
(4 , 0.027777777777777776)
(5 , 0.027777777777777776)
(6 , 0.027777777777777776)
(7 , 0.027777777777777776)
(8 , 0.027777777777777776)
(9 , 0.027777777777777776)
(10 , 0.027777777777777776)
(11 , 0.08333333333333333)
(12 , 0.08333333333333333)
(13 , 0.08333333333333333)
(14 , 0.08333333333333333)
(15 , 0.08333333333333333)
(16 , 0.08333333333333333)
(17 , 0.08333333333333333)
(18 , 0.08333333333333333)
(19 , 0.08333333333333333)
(20 , 0.08333333333333333)
(21 , 0.08333333333333333)
(22 , 0.08333333333333333)
(23 , 0.08333333333333333)
(24 , 0.08333333333333333)
(25 , 0.08333333333333333)
(26 , 0.08333333333333333)
(27 , 0.08333333333333333)
(28 , 0.08333333333333333)
(29 , 0.08333333333333333)
(30 , 0.08333333333333333)
(31 , 0.08333333333333333)
(32 , 0.08333333333333333)
(33 , 0.08333333333333333)
(34 , 0.08333333333333333)
(35 , 0.08333333333333333)
(36 , 0.08333333333333333)
(37 , 0.08333333333333333)
(38 , 0.08333333333333333)
(39 , 0.25)
(40 , 0.25)
(41 , 0.25)
(42 , 0.25)
(43 , 0.25)
(44 , 0.25)
(45 , 0.25)
(46 , 0.25)
(47 , 0.25)
(48 , 0.3333333333333333)
(49 , 0.3333333333333333)
(50 , 0.3333333333333333)
(51 , 0.3333333333333333)
(52 , 0.3333333333333333)
(53 , 0.3333333333333333)
(54 , 0.3333333333333333)
(55 , 0.3333333333333333)
(56 , 0.3333333333333333)
(57 , 0.3333333333333333)
(58 , 0.3333333333333333)
(59 , 0.3333333333333333)
(60 , 0.3333333333333333)
(61 , 0.3333333333333333)
(62 , 0.3333333333333333)
(63 , 0.3333333333333333)
(64 , 0.3333333333333333)
(65 , 0.3333333333333333)
(66 , 0.3333333333333333)
(67 , 0.3333333333333333)
(68 , 0.3333333333333333)
(69 , 0.3333333333333333)
(70 , 0.3333333333333333)
(71 , 0.3333333333333333)
(72 , 0.3333333333333333)
(73 , 0.3333333333333333)
(74 , 0.3333333333333333)
(75 , 0.3333333333333333)
(76 , 0.3333333333333333)
(77 , 0.3333333333333333)
(78 , 0.3333333333333333)
(79 , 0.3333333333333333)
(80 , 1.0)
(81 , 1.0)
(82 , 1.0)
(83 , 1.0)
(84 , 1.0)
(85 , 1.0)
(86 , 1.0)
(87 , 1.0)
(88 , 1.0)
(89 , 1.0)
(90 , 1.0)
(91 , 1.0)
(92 , 1.0)
(93 , 1.0)
(94 , 1.0)
(95 , 1.0)
(96 , 1.0)
(97 , 1.0)
(98 , 1.0)
(99 , 1.0)
(100 , 1.0)
};\label{plot_food}
\end{axis}
\end{tikzpicture}
}
  \vskip-1em
    \caption{\mrev{$\Imc$ (normal.): Stock (\ref{plot_stock}),
        Hospital (\ref{plot_hospital}), Food (\ref{plot_food}),
        Airport (\ref{plot_airport}), Adult (\ref{plot_adult}), Flight
        (\ref{plot_flight}), Voter (\ref{plot_voters}),  100
        iterations. Left: \algname{CONoise}, Right: \algname{RNoise} ($\alpha = 0.01, \beta=0$).}}
    \label{fig:imc}
\end{figure}

%% file: charts_runtime.tex
\scalebox{0.9}{
\begin{tikzpicture}
\begin{axis}
[
ytick={0,2000,4000,6000,8000,10000},
ymin=0,
ymax=10110,
xmin=0, xmax=1000,
ylabel = Running time (sec),
xlabel = \#tuples (thousands),
width=0.5\linewidth,
line width=0.3mm,
height=40mm,
ytick pos=left,
xtick pos=left,
xticklabel style = {font=\scriptsize},
xlabel style={yshift=5pt},
ylabel style={yshift=-10pt},
axis x line*=bottom,
axis y line*=left,
height = 1.3in
]
\addplot[color=black, mark repeat=3] coordinates{
(0, 0)
(100, 74.967)
(200, 301.087)
(300, 677.822)
(400, 1225.410)
(500, 1916.701)
(600, 2881.352)
(700, 3922.766)
(800, 4933.157)
(900, 6557.634)
(1000, 8092.894)
};\label{plot_imi_run}
\addplot[color=red,mark=|,mark repeat=3] coordinates{
(100, 76.858)
(200, 305.315)
(300, 694.946)
(400, 1251.053)
(500, 1957.758)
(600, 2943.645)
(700, 3994.433)
(800, 5053.575)
(900, 6718.977)
(1000, 8275.692)
};\label{plot_ip_run}
\addplot[color=brown,mark=triangle*,mark repeat=3] coordinates{
(100, 81.611)
(200, 333.964)
(300, 809.508)
(400, 1513.936)
(500, 2367.166)
(600, 3561.700)
(700, 4853.112)
(800, 6194.972)
(900, 8374.540)
(1000, 10102.158)
};\label{plot_imr_run}
\addplot[color=teal,mark=x,mark repeat=3] coordinates{
(100, 80.183)
(200, 320.439)
(300, 731.086)
(400, 1324.255)
(500, 2084.847)
(600, 3121.484)
(700, 4255.330)
(800, 5383.008)
(900, 7238.029)
(1000, 8804.302832126617)
};\label{plot_ilmr_run}
\addplot[color=blue,mark=10-pointed star,mark repeat=3] coordinates{
(100, 74.967)
(200, 301.087)
(300, 677.822)
(400, 1225.410)
(500, 1916.701)
(600, 2881.352)
(700, 3922.766)
(800, 4933.157)
(900, 6557.634)
(1000, 8092.894)
};\label{plot_id_run}
\end{axis}
\end{tikzpicture}
}

%% file: charts_runtime2.tex
\scalebox{0.9}{
\begin{tikzpicture}
\begin{axis}
[
ymin=0,
ymax=450,
xmin=0, xmax=2200,
xlabel = Iteration,
line width=0.3mm,
height=40mm,
ytick pos=left,
xtick pos=left,
xticklabel style = {font=\scriptsize},
xlabel style={yshift=5pt},
ylabel style={yshift=-10pt},
axis x line*=bottom,
axis y line*=left,
height = 1.3in,
width=0.5\linewidth
]
\addplot[color=black, mark repeat=20] coordinates{
(0 , 25.074521780014038)
(10 , 25.129196405410767)
(20 , 25.09213376045227)
(30 , 25.157063245773315)
(40 , 25.263591051101685)
(50 , 25.399030923843384)
(60 , 25.48612952232361)
(70 , 25.57797384262085)
(80 , 25.596266269683838)
(90 , 25.70475745201111)
(100 , 25.710497856140137)
(110 , 25.79739737510681)
(120 , 25.811776161193848)
(130 , 25.916924715042114)
(140 , 26.018372774124146)
(150 , 26.376500606536865)
(160 , 26.146461486816406)
(170 , 26.219006538391113)
(180 , 26.29550528526306)
(190 , 26.426640033721924)
(200 , 26.578562021255493)
(210 , 26.81566286087036)
(220 , 26.826810359954834)
(230 , 26.977219343185425)
(240 , 26.96749234199524)
(250 , 27.038452863693237)
(260 , 27.04692578315735)
(270 , 27.16376543045044)
(280 , 27.613882303237915)
(290 , 27.113322973251343)
(300 , 27.578330755233765)
(310 , 27.30938959121704)
(320 , 27.641091346740723)
(330 , 27.566463708877563)
(340 , 27.698620796203613)
(350 , 27.737303018569946)
(360 , 27.88825798034668)
(370 , 28.26515007019043)
(380 , 28.189384937286377)
(390 , 28.29190421104431)
(400 , 28.42594313621521)
(410 , 28.831026554107666)
(420 , 28.987040758132935)
(430 , 29.151796340942383)
(440 , 28.875860691070557)
(450 , 29.09329581260681)
(460 , 29.25063133239746)
(470 , 28.92104434967041)
(480 , 29.220211505889893)
(490 , 29.41259241104126)
(500 , 29.473719358444214)
(510 , 29.718798875808716)
(520 , 30.0163094997406)
(530 , 30.038823127746582)
(540 , 30.22940969467163)
(550 , 30.171936988830566)
(560 , 30.350597620010376)
(570 , 30.452359914779663)
(580 , 30.159749507904053)
(590 , 30.7020845413208)
(600 , 30.480971336364746)
(610 , 30.577945709228516)
(620 , 30.786019325256348)
(630 , 30.767058849334717)
(640 , 30.52671766281128)
(650 , 30.891332149505615)
(660 , 31.280264616012573)
(670 , 31.156057357788086)
(680 , 31.463599681854248)
(690 , 31.127300262451172)
(700 , 31.71012568473816)
(710 , 31.54692268371582)
(720 , 31.56072688102722)
(730 , 31.681392192840576)
(740 , 31.901854038238525)
(750 , 31.975966930389404)
(760 , 32.03885626792908)
(770 , 31.728169202804565)
(780 , 32.5095853805542)
(790 , 32.51958131790161)
(800 , 32.9377326965332)
(810 , 32.722901344299316)
(820 , 32.70341086387634)
(830 , 33.10396695137024)
(840 , 32.92009139060974)
(850 , 33.09859895706177)
(860 , 33.184799671173096)
(870 , 33.12744331359863)
(880 , 33.415096282958984)
(890 , 33.51657581329346)
(900 , 33.672873735427856)
(910 , 33.820064544677734)
(920 , 33.769140005111694)
(930 , 34.11166429519653)
(940 , 34.00509285926819)
(950 , 33.78897166252136)
(960 , 35.53795766830444)
(970 , 35.18163275718689)
(980 , 35.34104585647583)
(990 , 34.20970797538757)
(1000 , 34.707602739334106)
(1010 , 34.63207459449768)
(1020 , 34.316401958465576)
(1030 , 34.272953033447266)
(1040 , 34.64439010620117)
(1050 , 34.50737929344177)
(1060 , 35.00750923156738)
(1070 , 35.02718901634216)
(1080 , 34.6437087059021)
(1090 , 34.60973286628723)
(1100 , 35.11449384689331)
(1110 , 35.17917084693909)
(1120 , 35.71437215805054)
(1130 , 35.50011610984802)
(1140 , 35.61989235877991)
(1150 , 36.288827657699585)
(1160 , 35.472877740859985)
(1170 , 35.89848327636719)
(1180 , 35.63593888282776)
(1190 , 35.742334604263306)
(1200 , 36.19567799568176)
(1210 , 36.11689805984497)
(1220 , 36.37950301170349)
(1230 , 36.74225974082947)
(1240 , 36.60114121437073)
(1250 , 36.46691846847534)
(1260 , 37.21964883804321)
(1270 , 36.917845010757446)
(1280 , 37.18182635307312)
(1290 , 37.56662583351135)
(1300 , 37.21667790412903)
(1310 , 37.931599378585815)
(1320 , 38.46339774131775)
(1330 , 38.05517339706421)
(1340 , 38.89413380622864)
(1350 , 38.49103140830994)
(1360 , 37.60744905471802)
(1370 , 37.95081639289856)
(1380 , 38.92530918121338)
(1390 , 38.28565192222595)
(1400 , 38.23695230484009)
(1410 , 38.76095485687256)
(1420 , 38.752299785614014)
(1430 , 39.03774046897888)
(1440 , 38.601271867752075)
(1450 , 38.75175762176514)
(1460 , 39.609434366226196)
(1470 , 39.28938150405884)
(1480 , 39.470184087753296)
(1490 , 39.305500507354736)
(1500 , 39.12072944641113)
(1510 , 39.114131450653076)
(1520 , 39.49954080581665)
(1530 , 39.50244379043579)
(1540 , 39.81306600570679)
(1550 , 40.766305446624756)
(1560 , 39.87765121459961)
(1570 , 39.99611186981201)
(1580 , 39.73477792739868)
(1590 , 40.223634481430054)
(1600 , 40.005327224731445)
(1610 , 40.037967681884766)
(1620 , 40.30651330947876)
(1630 , 40.39955472946167)
(1640 , 40.94805359840393)
(1650 , 40.75419473648071)
(1660 , 40.70978283882141)
(1670 , 40.94232892990112)
(1680 , 40.358585357666016)
(1690 , 40.626275062561035)
(1700 , 40.39437651634216)
(1710 , 40.32929086685181)
(1720 , 40.97081017494202)
(1730 , 40.698137283325195)
(1740 , 41.59165644645691)
(1750 , 41.43220567703247)
(1760 , 41.598044633865356)
(1770 , 41.52061986923218)
(1780 , 41.43430685997009)
(1790 , 41.71810269355774)
(1800 , 42.41388559341431)
(1810 , 41.605738401412964)
(1820 , 42.15545463562012)
(1830 , 42.19201445579529)
(1840 , 42.78064823150635)
(1850 , 42.78844451904297)
(1860 , 41.95810151100159)
(1870 , 42.87033224105835)
(1880 , 42.701781272888184)
(1890 , 42.82617425918579)
(1900 , 42.6419472694397)
(1910 , 43.574304819107056)
(1920 , 42.87227249145508)
(1930 , 43.6321907043457)
(1940 , 43.45995736122131)
(1950 , 43.200833320617676)
(1960 , 42.96063446998596)
(1970 , 43.321014642715454)
(1980 , 43.114914417266846)
(1990 , 43.4623384475708)
(2000 , 43.247000217437744)
(2010 , 43.31692814826965)
(2020 , 43.52454090118408)
(2030 , 43.72251796722412)
(2040 , 43.99578857421875)
(2050 , 43.741448640823364)
(2060 , 44.391852378845215)
(2070 , 44.17952013015747)
(2080 , 44.57727265357971)
(2090 , 44.168997287750244)
(2100 , 44.33603358268738)
(2110 , 44.182589292526245)
(2120 , 44.89382553100586)
(2130 , 45.16859030723572)
(2140 , 44.671103954315186)
(2150 , 45.309248208999634)
(2160 , 45.372355699539185)
(2170 , 45.68674993515015)
(2180 , 45.553953647613525)
};\label{plot_imi_run3}
\addplot[color=red,mark=|,mark repeat=20] coordinates{
(0 , 25.503491163253784)
(10 , 25.834861278533936)
(20 , 25.953145742416382)
(30 , 26.410066843032837)
(40 , 26.9615638256073)
(50 , 27.358519315719604)
(60 , 27.709665298461914)
(70 , 27.966916799545288)
(80 , 28.149142503738403)
(90 , 28.375141143798828)
(100 , 28.866729259490967)
(110 , 28.816964864730835)
(120 , 29.23628044128418)
(130 , 29.474023818969727)
(140 , 29.71818447113037)
(150 , 30.06463098526001)
(160 , 30.456629991531372)
(170 , 30.658974647521973)
(180 , 31.0151104927063)
(190 , 31.502667903900146)
(200 , 31.487815856933594)
(210 , 31.84354066848755)
(220 , 32.07619857788086)
(230 , 32.49022603034973)
(240 , 32.89822840690613)
(250 , 32.901723861694336)
(260 , 32.958160400390625)
(270 , 33.383583068847656)
(280 , 33.95399498939514)
(290 , 34.032843828201294)
(300 , 34.27652406692505)
(310 , 34.319878816604614)
(320 , 34.50160098075867)
(330 , 34.909183979034424)
(340 , 35.10449409484863)
(350 , 35.388142824172974)
(360 , 36.02386450767517)
(370 , 36.31546688079834)
(380 , 36.401713132858276)
(390 , 36.54459810256958)
(400 , 37.269484758377075)
(410 , 37.88654804229736)
(420 , 37.70620036125183)
(430 , 38.39025521278381)
(440 , 38.79856204986572)
(450 , 39.066534996032715)
(460 , 38.83051109313965)
(470 , 39.159772872924805)
(480 , 39.56906294822693)
(490 , 39.927804708480835)
(500 , 40.366004943847656)
(510 , 40.688173055648804)
(520 , 41.061092138290405)
(530 , 41.17707061767578)
(540 , 41.621384382247925)
(550 , 41.89777636528015)
(560 , 42.092734813690186)
(570 , 42.215362787246704)
(580 , 42.35313439369202)
(590 , 42.75678062438965)
(600 , 42.697203397750854)
(610 , 43.163039207458496)
(620 , 43.47535419464111)
(630 , 43.65563917160034)
(640 , 44.16851305961609)
(650 , 44.3415949344635)
(660 , 44.990423917770386)
(670 , 45.003966331481934)
(680 , 45.64560413360596)
(690 , 45.55917739868164)
(700 , 45.897977113723755)
(710 , 46.02626323699951)
(720 , 46.37941241264343)
(730 , 46.588146686553955)
(740 , 46.75814485549927)
(750 , 47.53279447555542)
(760 , 47.39136505126953)
(770 , 47.72250962257385)
(780 , 48.45032262802124)
(790 , 48.70291233062744)
(800 , 49.38605499267578)
(810 , 49.780301094055176)
(820 , 49.57896947860718)
(830 , 50.575445890426636)
(840 , 50.44328045845032)
(850 , 50.49584126472473)
(860 , 51.08758592605591)
(870 , 50.72922134399414)
(880 , 51.37474179267883)
(890 , 51.737873792648315)
(900 , 52.128353118896484)
(910 , 51.89935636520386)
(920 , 52.286508321762085)
(930 , 52.49301886558533)
(940 , 52.4048855304718)
(950 , 53.078338861465454)
(960 , 56.56641125679016)
(970 , 54.02080845832825)
(980 , 55.017613887786865)
(990 , 54.503323554992676)
(1000 , 55.18786334991455)
(1010 , 54.74099016189575)
(1020 , 54.6911404132843)
(1030 , 54.6622748374939)
(1040 , 55.53467297554016)
(1050 , 55.42452526092529)
(1060 , 55.84244704246521)
(1070 , 55.95669746398926)
(1080 , 55.92869973182678)
(1090 , 56.10347390174866)
(1100 , 56.86227560043335)
(1110 , 57.38290286064148)
(1120 , 57.94810080528259)
(1130 , 57.97410535812378)
(1140 , 58.37896537780762)
(1150 , 58.84461832046509)
(1160 , 58.607590436935425)
(1170 , 58.966381788253784)
(1180 , 59.04979848861694)
(1190 , 59.473793745040894)
(1200 , 59.913851261138916)
(1210 , 59.9903666973114)
(1220 , 60.3142364025116)
(1230 , 60.70616960525513)
(1240 , 60.803593158721924)
(1250 , 61.55029225349426)
(1260 , 62.08362436294556)
(1270 , 62.141124963760376)
(1280 , 62.87101912498474)
(1290 , 63.03670573234558)
(1300 , 63.24995255470276)
(1310 , 63.84394407272339)
(1320 , 64.40563416481018)
(1330 , 64.38403606414795)
(1340 , 65.44647455215454)
(1350 , 64.84865522384644)
(1360 , 65.15091705322266)
(1370 , 64.96838021278381)
(1380 , 65.35294604301453)
(1390 , 65.937912940979)
(1400 , 66.25771641731262)
(1410 , 66.7510883808136)
(1420 , 66.68630480766296)
(1430 , 66.88041353225708)
(1440 , 67.02987933158875)
(1450 , 67.03957986831665)
(1460 , 67.90333557128906)
(1470 , 67.87660956382751)
(1480 , 68.37080717086792)
(1490 , 68.63466262817383)
(1500 , 68.41839599609375)
(1510 , 68.71791458129883)
(1520 , 69.08529591560364)
(1530 , 69.26828145980835)
(1540 , 70.11073923110962)
(1550 , 70.50349164009094)
(1560 , 70.54450249671936)
(1570 , 70.5578920841217)
(1580 , 70.16935348510742)
(1590 , 71.02635979652405)
(1600 , 70.95986199378967)
(1610 , 71.43714046478271)
(1620 , 71.91552543640137)
(1630 , 72.16801452636719)
(1640 , 72.39211368560791)
(1650 , 72.19345140457153)
(1660 , 72.39243531227112)
(1670 , 73.06309294700623)
(1680 , 72.50606203079224)
(1690 , 72.6999671459198)
(1700 , 72.85542416572571)
(1710 , 72.96260166168213)
(1720 , 73.37909293174744)
(1730 , 73.82872247695923)
(1740 , 74.81832098960876)
(1750 , 74.87033987045288)
(1760 , 75.52430868148804)
(1770 , 75.35496068000793)
(1780 , 75.47237968444824)
(1790 , 76.09341502189636)
(1800 , 76.2355375289917)
(1810 , 76.23707580566406)
(1820 , 76.75555992126465)
(1830 , 77.49354195594788)
(1840 , 78.08543419837952)
(1850 , 78.12549376487732)
(1860 , 77.7292468547821)
(1870 , 78.36880230903625)
(1880 , 78.95546269416809)
(1890 , 78.97350835800171)
(1900 , 79.08761548995972)
(1910 , 79.23193907737732)
(1920 , 79.64977169036865)
(1930 , 79.82951164245605)
(1940 , 80.15406703948975)
(1950 , 80.49646544456482)
(1960 , 80.3672845363617)
(1970 , 80.4475736618042)
(1980 , 81.08033275604248)
(1990 , 80.57315516471863)
(2000 , 80.99608302116394)
(2010 , 81.26245164871216)
(2020 , 81.51374459266663)
(2030 , 82.70509839057922)
(2040 , 82.7456202507019)
(2050 , 82.6517698764801)
(2060 , 82.98520755767822)
(2070 , 83.32477235794067)
(2080 , 83.50733065605164)
(2090 , 83.45534062385559)
(2100 , 84.00576663017273)
(2110 , 83.83671021461487)
(2120 , 84.63727807998657)
(2130 , 85.31890630722046)
(2140 , 85.1618423461914)
(2150 , 85.81177496910095)
(2160 , 85.74441075325012)
(2170 , 86.0435380935669)
(2180 , 86.45502376556396)
};\label{plot_ip_run3}
\addplot[color=brown,mark=triangle*,mark repeat=20] coordinates{
(0 , 25.38556718826294)
(10 , 25.84475827217102)
(20 , 26.431734561920166)
(30 , 27.761497497558594)
(40 , 28.95021152496338)
(50 , 30.107333183288574)
(60 , 31.10263991355896)
(70 , 31.7861270904541)
(80 , 32.707194566726685)
(90 , 33.01503348350525)
(100 , 34.46321654319763)
(110 , 34.79875898361206)
(120 , 35.45231533050537)
(130 , 36.201844453811646)
(140 , 37.14863586425781)
(150 , 38.53230309486389)
(160 , 39.551995038986206)
(170 , 40.289177656173706)
(180 , 41.34355330467224)
(190 , 44.076597929000854)
(200 , 44.36135697364807)
(210 , 45.07152199745178)
(220 , 45.707468032836914)
(230 , 47.64521336555481)
(240 , 48.181320905685425)
(250 , 48.47364902496338)
(260 , 49.89038896560669)
(270 , 51.94196081161499)
(280 , 54.8047194480896)
(290 , 54.55363345146179)
(300 , 56.08442544937134)
(310 , 56.323312520980835)
(320 , 57.219889640808105)
(330 , 59.42997097969055)
(340 , 59.780943155288696)
(350 , 65.74087905883789)
(360 , 67.36569452285767)
(370 , 70.44060564041138)
(380 , 71.24686932563782)
(390 , 74.25474238395691)
(400 , 82.82651019096375)
(410 , 80.90109705924988)
(420 , 81.46128582954407)
(430 , 83.95561671257019)
(440 , 89.64281630516052)
(450 , 91.5124020576477)
(460 , 92.3571605682373)
(470 , 96.93716979026794)
(480 , 92.1370370388031)
(490 , 94.49719285964966)
(500 , 98.38131427764893)
(510 , 99.63814043998718)
(520 , 97.70626831054688)
(530 , 103.5829439163208)
(540 , 115.11657047271729)
(550 , 119.04458212852478)
(560 , 117.09003019332886)
(570 , 118.84987354278564)
(580 , 120.71456623077393)
(590 , 123.0201027393341)
(600 , 123.05013990402222)
(610 , 114.0535831451416)
(620 , 119.80580496788025)
(630 , 121.13330578804016)
(640 , 121.51136112213135)
(650 , 121.42000889778137)
(660 , 125.12496733665466)
(670 , 139.20516896247864)
(680 , 130.76436281204224)
(690 , 135.52395915985107)
(700 , 140.46577978134155)
(710 , 140.93461966514587)
(720 , 139.91933941841125)
(730 , 140.31727600097656)
(740 , 140.63884377479553)
(750 , 144.6907877922058)
(760 , 144.05916476249695)
(770 , 147.1460211277008)
(780 , 149.7040731906891)
(790 , 153.44728088378906)
(800 , 154.968918800354)
(810 , 163.5828778743744)
(820 , 160.2998788356781)
(830 , 161.18198990821838)
(840 , 163.83286380767822)
(850 , 169.01466369628906)
(860 , 174.3666651248932)
(870 , 173.52756762504578)
(880 , 177.01867270469666)
(890 , 177.16317915916443)
(900 , 178.932683467865)
(910 , 178.80101251602173)
(920 , 184.0389859676361)
(930 , 182.3471188545227)
(940 , 186.31870460510254)
(950 , 188.06659150123596)
(960 , 208.00855493545532)
(970 , 196.213054895401)
(980 , 197.00502610206604)
(990 , 202.73430752754211)
(1000 , 203.28201413154602)
(1010 , 198.72260522842407)
(1020 , 195.5516219139099)
(1030 , 198.28184390068054)
(1040 , 205.53014492988586)
(1050 , 202.79473185539246)
(1060 , 210.86107468605042)
(1070 , 205.31938648223877)
(1080 , 204.46801567077637)
(1090 , 208.92884016036987)
(1100 , 214.97750544548035)
(1110 , 221.91615796089172)
(1120 , 226.09549474716187)
(1130 , 218.68016958236694)
(1140 , 226.9167377948761)
(1150 , 229.18047213554382)
(1160 , 221.59851241111755)
(1170 , 226.54292678833008)
(1180 , 223.57164669036865)
(1190 , 230.3176634311676)
(1200 , 233.90245938301086)
(1210 , 231.2934126853943)
(1220 , 233.46538972854614)
(1230 , 242.64671897888184)
(1240 , 241.79494857788086)
(1250 , 241.55906581878662)
(1260 , 252.34345746040344)
(1270 , 253.0208592414856)
(1280 , 258.2556366920471)
(1290 , 255.63359236717224)
(1300 , 263.6917531490326)
(1310 , 266.0434114933014)
(1320 , 270.7087595462799)
(1330 , 297.2402620315552)
(1340 , 272.98899030685425)
(1350 , 270.0177803039551)
(1360 , 259.74318861961365)
(1370 , 278.82690596580505)
(1380 , 274.6273202896118)
(1390 , 283.61635422706604)
(1400 , 271.8168318271637)
(1410 , 288.4082763195038)
(1420 , 279.5849013328552)
(1430 , 328.48179149627686)
(1440 , 288.82888102531433)
(1450 , 298.7547855377197)
(1460 , 292.82370352745056)
(1470 , 299.4849011898041)
(1480 , 305.24357414245605)
(1490 , 292.41618824005127)
(1500 , 285.5348846912384)
(1510 , 286.19056034088135)
(1520 , 312.3257532119751)
(1530 , 302.73621463775635)
(1540 , 322.6502323150635)
(1550 , 316.0224826335907)
(1560 , 322.24173974990845)
(1570 , 348.682089805603)
(1580 , 310.9704821109772)
(1590 , 355.45769000053406)
(1600 , 349.8512589931488)
(1610 , 311.50594997406006)
(1620 , 320.35478615760803)
(1630 , 319.3910005092621)
(1640 , 326.7304255962372)
(1650 , 330.19148802757263)
(1660 , 321.51243567466736)
(1670 , 328.52015137672424)
(1680 , 326.17190647125244)
(1690 , 317.1120834350586)
(1700 , 326.0961527824402)
(1710 , 320.41651153564453)
(1720 , 365.27201890945435)
(1730 , 346.50506258010864)
(1740 , 343.3980362415314)
(1750 , 338.261070728302)
(1760 , 358.4277939796448)
(1770 , 399.3969805240631)
(1780 , 365.3536455631256)
(1790 , 389.7907793521881)
(1800 , 398.877405166626)
(1810 , 355.51145482063293)
(1820 , 350.66352462768555)
(1830 , 403.6628084182739)
(1840 , 356.4792249202728)
(1850 , 384.8842628002167)
(1860 , 403.0037353038788)
(1870 , 371.0213830471039)
(1880 , 389.2722659111023)
(1890 , 428.7170226573944)
(1900 , 382.92942786216736)
(1910 , 429.0789759159088)
(1920 , 431.9883725643158)
(1930 , 435.92234110832214)
(1940 , 386.92686009407043)
(1950 , 384.80518770217896)
(1960 , 375.4782965183258)
(1970 , 384.38692474365234)
(1980 , 374.5700013637543)
(1990 , 383.59420919418335)
(2000 , 384.31414461135864)
(2010 , 386.52217745780945)
(2020 , 446.3257358074188)
(2030 , 400.8078279495239)
(2040 , 401.41518235206604)
(2050 , 436.9190444946289)
(2060 , 426.33636355400085)
(2070 , 445.7528839111328)
(2080 , 457.65113711357117)
(2090 , 391.0124258995056)
(2100 , 464.69100308418274)
(2110 , 420.2045431137085)
(2120 , 467.04617619514465)
(2130 , 464.03860330581665)
(2140 , 464.1860029697418)
(2150 , 486.83293652534485)
(2160 , 433.84249329566956)
(2170 , 441.54139709472656)
(2180 , 442.9468548297882)
};\label{plot_imr_run3}
\addplot[color=teal,mark=x,mark repeat=20] coordinates{
(0 , 25.442628860473633)
(10 , 25.773725032806396)
(20 , 26.02989888191223)
(30 , 26.61217451095581)
(40 , 27.448691606521606)
(50 , 28.047232627868652)
(60 , 28.75437021255493)
(70 , 29.111451864242554)
(80 , 29.143840312957764)
(90 , 29.461080312728882)
(100 , 30.2906277179718)
(110 , 30.464904069900513)
(120 , 30.734069108963013)
(130 , 31.180429935455322)
(140 , 31.665070295333862)
(150 , 32.605746030807495)
(160 , 32.82266354560852)
(170 , 33.45480489730835)
(180 , 33.884772539138794)
(190 , 34.83810567855835)
(200 , 34.93401265144348)
(210 , 35.417319774627686)
(220 , 35.88921093940735)
(230 , 36.69970226287842)
(240 , 36.937812089920044)
(250 , 37.08249592781067)
(260 , 39.76689624786377)
(270 , 38.28941082954407)
(280 , 39.51887059211731)
(290 , 39.52102518081665)
(300 , 40.54122447967529)
(310 , 40.14650249481201)
(320 , 41.4234561920166)
(330 , 41.032787799835205)
(340 , 41.65507793426514)
(350 , 42.06362724304199)
(360 , 42.10163497924805)
(370 , 45.19161534309387)
(380 , 44.25824975967407)
(390 , 44.28877663612366)
(400 , 45.20856189727783)
(410 , 48.47554540634155)
(420 , 46.97636008262634)
(430 , 47.94418907165527)
(440 , 47.87141418457031)
(450 , 48.60627722740173)
(460 , 50.37498211860657)
(470 , 48.78057527542114)
(480 , 50.82713770866394)
(490 , 51.89747881889343)
(500 , 53.88199806213379)
(510 , 52.53499507904053)
(520 , 54.931588888168335)
(530 , 53.93359160423279)
(540 , 54.653443574905396)
(550 , 55.98694968223572)
(560 , 55.49626803398132)
(570 , 55.436182498931885)
(580 , 56.760375022888184)
(590 , 57.24195456504822)
(600 , 56.66060280799866)
(610 , 57.64599919319153)
(620 , 58.45728349685669)
(630 , 60.64168930053711)
(640 , 59.704699754714966)
(650 , 59.90460467338562)
(660 , 61.78936147689819)
(670 , 62.09872579574585)
(680 , 65.27701592445374)
(690 , 64.8772177696228)
(700 , 64.89518666267395)
(710 , 64.09345173835754)
(720 , 66.73690748214722)
(730 , 65.04027462005615)
(740 , 66.23583197593689)
(750 , 65.50822615623474)
(760 , 65.43336582183838)
(770 , 68.2142403125763)
(780 , 70.31436276435852)
(790 , 70.38785171508789)
(800 , 71.87969136238098)
(810 , 71.71931719779968)
(820 , 69.62266278266907)
(830 , 74.45288848876953)
(840 , 72.0482280254364)
(850 , 72.99803900718689)
(860 , 74.7149007320404)
(870 , 77.90635371208191)
(880 , 74.98851299285889)
(890 , 73.02760124206543)
(900 , 76.13036894798279)
(910 , 75.51126670837402)
(920 , 76.4782977104187)
(930 , 78.37859296798706)
(940 , 78.39118242263794)
(950 , 80.79150581359863)
(960 , 84.98897433280945)
(970 , 87.65429830551147)
(980 , 80.5598738193512)
(990 , 82.41020226478577)
(1000 , 92.11413550376892)
(1010 , 82.47640204429626)
(1020 , 82.5706946849823)
(1030 , 89.45185613632202)
(1040 , 81.16060590744019)
(1050 , 78.81584477424622)
(1060 , 90.9998300075531)
(1070 , 84.84638524055481)
(1080 , 84.23021221160889)
(1090 , 85.06242394447327)
(1100 , 95.36606025695801)
(1110 , 96.75079298019409)
(1120 , 84.90560054779053)
(1130 , 90.30372381210327)
(1140 , 85.50417232513428)
(1150 , 89.67244386672974)
(1160 , 86.03400683403015)
(1170 , 90.98407745361328)
(1180 , 91.70696449279785)
(1190 , 94.19771647453308)
(1200 , 93.18387699127197)
(1210 , 92.99597835540771)
(1220 , 103.25173044204712)
(1230 , 98.1987235546112)
(1240 , 95.65326499938965)
(1250 , 93.28183960914612)
(1260 , 98.85529780387878)
(1270 , 101.87752223014832)
(1280 , 100.73755145072937)
(1290 , 95.83915710449219)
(1300 , 102.20499229431152)
(1310 , 112.27620577812195)
(1320 , 103.16180062294006)
(1330 , 103.80511617660522)
(1340 , 100.72421073913574)
(1350 , 105.92977619171143)
(1360 , 108.81554770469666)
(1370 , 114.59197235107422)
(1380 , 107.72586488723755)
(1390 , 109.92265057563782)
(1400 , 102.33181095123291)
(1410 , 121.04299783706665)
(1420 , 120.04905223846436)
(1430 , 109.79274439811707)
(1440 , 109.4028651714325)
(1450 , 117.5864725112915)
(1460 , 107.8478627204895)
(1470 , 105.77000856399536)
(1480 , 121.71756172180176)
(1490 , 109.18219590187073)
(1500 , 108.14434099197388)
(1510 , 110.64759826660156)
(1520 , 124.21567273139954)
(1530 , 112.0795886516571)
(1540 , 115.41571140289307)
(1550 , 115.46174383163452)
(1560 , 114.50400233268738)
(1570 , 111.29958295822144)
(1580 , 112.67297840118408)
(1590 , 126.77256035804749)
(1600 , 113.80431938171387)
(1610 , 125.18620800971985)
(1620 , 115.92171382904053)
(1630 , 115.36429381370544)
(1640 , 123.90708541870117)
(1650 , 122.73959279060364)
(1660 , 118.80906057357788)
(1670 , 124.67685079574585)
(1680 , 127.71010327339172)
(1690 , 128.71279501914978)
(1700 , 131.57976055145264)
(1710 , 127.51275706291199)
(1720 , 120.18976283073425)
(1730 , 140.93870425224304)
(1740 , 137.6375458240509)
(1750 , 117.3117983341217)
(1760 , 131.05701780319214)
(1770 , 138.30977368354797)
(1780 , 135.9681739807129)
(1790 , 119.64297366142273)
(1800 , 126.08526372909546)
(1810 , 134.9546880722046)
(1820 , 143.44678354263306)
(1830 , 144.86282539367676)
(1840 , 142.1945571899414)
(1850 , 130.77789902687073)
(1860 , 126.61487913131714)
(1870 , 129.47033309936523)
(1880 , 146.2201910018921)
(1890 , 132.75348544120789)
(1900 , 133.67890334129333)
(1910 , 135.04660487174988)
(1920 , 147.0606336593628)
(1930 , 149.92789316177368)
(1940 , 138.08070135116577)
(1950 , 136.70394277572632)
(1960 , 130.0709261894226)
(1970 , 138.63151335716248)
(1980 , 132.1836416721344)
(1990 , 138.55967044830322)
(2000 , 139.8173110485077)
(2010 , 149.52530884742737)
(2020 , 139.9520127773285)
(2030 , 141.37263941764832)
(2040 , 146.27126336097717)
(2050 , 158.72298216819763)
(2060 , 148.58999729156494)
(2070 , 157.6023952960968)
(2080 , 134.9789478778839)
(2090 , 146.07008957862854)
(2100 , 144.76493167877197)
(2110 , 144.90907549858093)
(2120 , 139.18031454086304)
(2130 , 145.66331434249878)
(2140 , 140.10511422157288)
(2150 , 155.87658309936523)
(2160 , 139.27021527290344)
(2170 , 145.29450225830078)
(2180 , 157.20511150360107)
};\label{plot_ilmr_run3}
\addplot[color=blue,mark=10-pointed star,mark repeat=20] coordinates{
(0 , 25.074521780014038)
(10 , 25.129196405410767)
(20 , 25.09213376045227)
(30 , 25.157063245773315)
(40 , 25.263591051101685)
(50 , 25.399030923843384)
(60 , 25.48612952232361)
(70 , 25.57797384262085)
(80 , 25.596266269683838)
(90 , 25.70475745201111)
(100 , 25.710497856140137)
(110 , 25.79739737510681)
(120 , 25.811776161193848)
(130 , 25.916924715042114)
(140 , 26.018372774124146)
(150 , 26.376500606536865)
(160 , 26.146461486816406)
(170 , 26.219006538391113)
(180 , 26.29550528526306)
(190 , 26.426640033721924)
(200 , 26.578562021255493)
(210 , 26.81566286087036)
(220 , 26.826810359954834)
(230 , 26.977219343185425)
(240 , 26.96749234199524)
(250 , 27.038452863693237)
(260 , 27.04692578315735)
(270 , 27.16376543045044)
(280 , 27.613882303237915)
(290 , 27.113322973251343)
(300 , 27.578330755233765)
(310 , 27.30938959121704)
(320 , 27.641091346740723)
(330 , 27.566463708877563)
(340 , 27.698620796203613)
(350 , 27.737303018569946)
(360 , 27.88825798034668)
(370 , 28.26515007019043)
(380 , 28.189384937286377)
(390 , 28.29190421104431)
(400 , 28.42594313621521)
(410 , 28.831026554107666)
(420 , 28.987040758132935)
(430 , 29.151796340942383)
(440 , 28.875860691070557)
(450 , 29.09329581260681)
(460 , 29.25063133239746)
(470 , 28.92104434967041)
(480 , 29.220211505889893)
(490 , 29.41259241104126)
(500 , 29.473719358444214)
(510 , 29.718798875808716)
(520 , 30.0163094997406)
(530 , 30.038823127746582)
(540 , 30.22940969467163)
(550 , 30.171936988830566)
(560 , 30.350597620010376)
(570 , 30.452359914779663)
(580 , 30.159749507904053)
(590 , 30.7020845413208)
(600 , 30.480971336364746)
(610 , 30.577945709228516)
(620 , 30.786019325256348)
(630 , 30.767058849334717)
(640 , 30.52671766281128)
(650 , 30.891332149505615)
(660 , 31.280264616012573)
(670 , 31.156057357788086)
(680 , 31.463599681854248)
(690 , 31.127300262451172)
(700 , 31.71012568473816)
(710 , 31.54692268371582)
(720 , 31.56072688102722)
(730 , 31.681392192840576)
(740 , 31.901854038238525)
(750 , 31.975966930389404)
(760 , 32.03885626792908)
(770 , 31.728169202804565)
(780 , 32.5095853805542)
(790 , 32.51958131790161)
(800 , 32.9377326965332)
(810 , 32.722901344299316)
(820 , 32.70341086387634)
(830 , 33.10396695137024)
(840 , 32.92009139060974)
(850 , 33.09859895706177)
(860 , 33.184799671173096)
(870 , 33.12744331359863)
(880 , 33.415096282958984)
(890 , 33.51657581329346)
(900 , 33.672873735427856)
(910 , 33.820064544677734)
(920 , 33.769140005111694)
(930 , 34.11166429519653)
(940 , 34.00509285926819)
(950 , 33.78897166252136)
(960 , 35.53795766830444)
(970 , 35.18163275718689)
(980 , 35.34104585647583)
(990 , 34.20970797538757)
(1000 , 34.707602739334106)
(1010 , 34.63207459449768)
(1020 , 34.316401958465576)
(1030 , 34.272953033447266)
(1040 , 34.64439010620117)
(1050 , 34.50737929344177)
(1060 , 35.00750923156738)
(1070 , 35.02718901634216)
(1080 , 34.6437087059021)
(1090 , 34.60973286628723)
(1100 , 35.11449384689331)
(1110 , 35.17917084693909)
(1120 , 35.71437215805054)
(1130 , 35.50011610984802)
(1140 , 35.61989235877991)
(1150 , 36.288827657699585)
(1160 , 35.472877740859985)
(1170 , 35.89848327636719)
(1180 , 35.63593888282776)
(1190 , 35.742334604263306)
(1200 , 36.19567799568176)
(1210 , 36.11689805984497)
(1220 , 36.37950301170349)
(1230 , 36.74225974082947)
(1240 , 36.60114121437073)
(1250 , 36.46691846847534)
(1260 , 37.21964883804321)
(1270 , 36.917845010757446)
(1280 , 37.18182635307312)
(1290 , 37.56662583351135)
(1300 , 37.21667790412903)
(1310 , 37.931599378585815)
(1320 , 38.46339774131775)
(1330 , 38.05517339706421)
(1340 , 38.89413380622864)
(1350 , 38.49103140830994)
(1360 , 37.60744905471802)
(1370 , 37.95081639289856)
(1380 , 38.92530918121338)
(1390 , 38.28565192222595)
(1400 , 38.23695230484009)
(1410 , 38.76095485687256)
(1420 , 38.752299785614014)
(1430 , 39.03774046897888)
(1440 , 38.601271867752075)
(1450 , 38.75175762176514)
(1460 , 39.609434366226196)
(1470 , 39.28938150405884)
(1480 , 39.470184087753296)
(1490 , 39.305500507354736)
(1500 , 39.12072944641113)
(1510 , 39.114131450653076)
(1520 , 39.49954080581665)
(1530 , 39.50244379043579)
(1540 , 39.81306600570679)
(1550 , 40.766305446624756)
(1560 , 39.87765121459961)
(1570 , 39.99611186981201)
(1580 , 39.73477792739868)
(1590 , 40.223634481430054)
(1600 , 40.005327224731445)
(1610 , 40.037967681884766)
(1620 , 40.30651330947876)
(1630 , 40.39955472946167)
(1640 , 40.94805359840393)
(1650 , 40.75419473648071)
(1660 , 40.70978283882141)
(1670 , 40.94232892990112)
(1680 , 40.358585357666016)
(1690 , 40.626275062561035)
(1700 , 40.39437651634216)
(1710 , 40.32929086685181)
(1720 , 40.97081017494202)
(1730 , 40.698137283325195)
(1740 , 41.59165644645691)
(1750 , 41.43220567703247)
(1760 , 41.598044633865356)
(1770 , 41.52061986923218)
(1780 , 41.43430685997009)
(1790 , 41.71810269355774)
(1800 , 42.41388559341431)
(1810 , 41.605738401412964)
(1820 , 42.15545463562012)
(1830 , 42.19201445579529)
(1840 , 42.78064823150635)
(1850 , 42.78844451904297)
(1860 , 41.95810151100159)
(1870 , 42.87033224105835)
(1880 , 42.701781272888184)
(1890 , 42.82617425918579)
(1900 , 42.6419472694397)
(1910 , 43.574304819107056)
(1920 , 42.87227249145508)
(1930 , 43.6321907043457)
(1940 , 43.45995736122131)
(1950 , 43.200833320617676)
(1960 , 42.96063446998596)
(1970 , 43.321014642715454)
(1980 , 43.114914417266846)
(1990 , 43.4623384475708)
(2000 , 43.247000217437744)
(2010 , 43.31692814826965)
(2020 , 43.52454090118408)
(2030 , 43.72251796722412)
(2040 , 43.99578857421875)
(2050 , 43.741448640823364)
(2060 , 44.391852378845215)
(2070 , 44.17952013015747)
(2080 , 44.57727265357971)
(2090 , 44.168997287750244)
(2100 , 44.33603358268738)
(2110 , 44.182589292526245)
(2120 , 44.89382553100586)
(2130 , 45.16859030723572)
(2140 , 44.671103954315186)
(2150 , 45.309248208999634)
(2160 , 45.372355699539185)
(2170 , 45.68674993515015)
(2180 , 45.553953647613525)
};\label{plot_id_run3}
\end{axis}
\end{tikzpicture}
}

%% file: charts_holoclean.tex
\scalebox{0.9}{
\begin{tikzpicture}
\begin{axis}
[
ytick={0,0.2,0.4,0.6,0.8,1},
ymin=0,
ymax=1,
xmin=0, xmax=15,
line width=0.3mm,
ylabel = Measure value,
xlabel = \#DCs,
width=0.5\linewidth,
ytick pos=left,
xtick pos=left,
xlabel style={yshift=5pt},
ylabel style={yshift=-10pt},
axis x line*=bottom,
axis y line*=left,
height=1.3in,
]
\addplot[color=black,mark repeat=3] coordinates{
(0, 1)
(1, 0.93)
(2, 0.89)
(3, 0.85)
(4, 0.77)
(5, 0.63)
(6, 0.6)
(7, 0.51)
(8, 0.48)
(9, 0.43)
(10, 0.39)
(11, 0.31)
(12, 0.25)
(13, 0.2)
(14, 0.11)
(15, 0.05)
};\label{plot_imi}
\addplot[color=red,mark=|,mark repeat=3] coordinates{
(0, 1)
(1, 1)
(2, 1)
(3, 1)
(4, 1)
(5, 1)
(6, 1)
(7, 0.998)
(8, 0.998)
(9, 0.998)
(10, 0.998)
(11, 0.995)
(12, 0.981)
(13, 0.955)
(14, 0.761)
(15, 0.441)
};\label{plot_ip}
\addplot[color=brown,mark=triangle*,mark repeat=3] coordinates{
(0, 1)
(1, 0.95)
(2, 0.89)
(3, 0.84)
(4, 0.78)
(5, 0.7)
(6, 0.65)
(7, 0.59)
(8, 0.55)
(9, 0.49)
(10, 0.43)
(11, 0.34)
(12, 0.28)
(13, 0.2)
(14, 0.13)
(15, 0.05)
};\label{plot_ir}
\addplot[color=teal,mark=x,mark repeat=3] coordinates{
(0, 1)
(1, 0.95)
(2, 0.89)
(3, 0.84)
(4, 0.78)
(5, 0.7)
(6, 0.65)
(7, 0.59)
(8, 0.55)
(9, 0.49)
(10, 0.43)
(11, 0.34)
(12, 0.28)
(13, 0.2)
(14, 0.13)
(15, 0.05)
};\label{plot_ilinr}
\addplot[color=blue,mark=10-pointed star,mark repeat=3] coordinates{
(0, 1)
(1, 1)
(2, 1)
(3, 1)
(4, 1)
(5, 1)
(6, 1)
(7, 1)
(8, 1)
(9, 1)
(10, 1)
(11, 1)
(12, 1)
(13, 1)
(14, 1)
(15, 1)
};\label{plot_id}
\end{axis}
\end{tikzpicture}
}

%% file: conclusions.tex
\section{Concluding Remarks}\label{sec:conclusions}
We explored inconsistency measures over databases, and investigated the properties
that should be accounted for in the choice of a measure in a specific use case. We discussed four properties where two, continuity and  progression, are defined in the context of the underlying repair system.
We also used the properties to reason about various specific instances of inconsistency measures.  The combination of the properties and the computational complexity shed a positive light on the linear relaxation of minimal repairing when considering DCs and tuple deletion.  In fact, the design of this measure is driven by that combination, and is not as interpretable as the others.  \mrev{Our experimental study shows that the measures that well behave theoretically for tuple deletions also exhibit a good empirical
      behavior, even when our error models capture attribute updates
      rather than tuple insertions (to be remedied by tuple
      deletions). These measures were also able to effectively capture progress in our HoloClean case study.}

\added{This work opens the way to an important angle of inconsistency measurement that has not been treated before, and many fundamental problems remain open.
We plan to explore other properties as well as \e{completeness} criteria for sets of properties to determine sufficiency for
certain use cases. Another important direction is to explore more general repair systems (allowing different types of constraints and repairing operations).
It is also interesting to investigate the adaption of Grant and Hunter's concept of information loss~\cite{DBLP:conf/ecsqaru/GrantH11}, and explore the trade-off between inconsistency reduction and information loss, in the context of database repairing.
The final goal is to devise actual measures that are practically useful, efficient to compute, and justified by a
clear theoretical ground.

\eat{A problem that is complementary to the problem considered in this paper is that of measuring the responsibility of individual facts to inconsistency. Such a measure can assist in directing the data scientist to the problematic facts in the database (e.g., as a next-step suggestion in interactive data repairing systems). One possible way to quantify this responsibility is to use the \e{Shapley value}~\cite{shapley:book1952}, a well-known measure from cooperative game theory for determining the contribution of a
  player to a coalition game. In the KR community, the Shapley value has been used to quantify the contribution of formulas to the inconsistency of knowledge bases~\cite{DBLP:journals/ai/HunterK10}. In the context of databases, it has recently been proposed to use the Shapley value to measure the contribution of tuples to query results~\cite{shapley-icdt2020,DBLP:conf/pods/ReshefKL20}. 
Note that the Shapley value can be used to determine the responsibility of facts to database inconsistency w.r.t.~any given inconsistency measure, which again highlights the importance of designing good database inconsistency measures.}

}

\balance

%% file: appendix.tex
\section{Proof of Proposition~\ref{prop:relationsips}}

\begin{repproposition}{\ref{prop:relationsips}}
\propdependencies
\end{repproposition}
\begin{proof}
The first part of the proposition holds since if $\I$ satisfies progression, then for every inconsistent database $D$ there exists an operation $o$ such that $\I(\Sigma,o(D))$ is strictly lower than $\I(\Sigma,D)$, which implies that $\I(\Sigma,D)>0$.

For the second part, let us assume, by way of contradiction, that $\I$ does not satisfy progression. Then, there exists a database $D$ and a set of constraints $\Sigma$ such that $D\not\models\Sigma$, and for every operation $o\in O$ on $D$ it holds that $\I(\Sigma,o(D))\ge\I(\Sigma,D)$. Moreover, $\I$ satisfies positivity; thus, it holds that $\I(\Sigma,D)>0$. Since $\C$ is realizable by $\R$, there exists a sequence of operations $o_1,\dots,o_n$ from $\R$, such that $o_n\circ\dots\circ
o_1(D)$ is consistent (hence, $\I(\Sigma,o_n\circ\dots\circ
o_1(D))=0$). We conclude that there exists an operation $o_j$ such that $\Delta_{\I,\Sigma}(o_j,o_{j-1}\circ\dots\circ o_1(D))>0$, but there is no such operation on $D$, which is a contradiction to the fact that $\I$ satisfies bounded continuity, and that concludes our proof.
\end{proof}

\section{Proof of Theorem~\ref{THM:EGD}}

\begin{reptheorem}{\ref{THM:EGD}}
\thmegd
\end{reptheorem}

We start by proving the negative side of the theorem. That is, we prove the following.
\begin{lemma}
Let $\R=\R_{\subseteq}$ and let $\Sigma=\{\sigma\}$ where $\sigma$ is an EGD of the form
$\forall x_1,x_2,x_3\left[R(x_1,x_2),R(x_2,x_3)\Rightarrow (x_i=x_j)\right]$.
Then, computing $\Imr(\Sigma,D)$ is NP-hard.
\end{lemma}

\begin{proof}
  We build a
  reduction from the MaxCut problem to the problem of computing
  $\Imr(\Sigma,D)$ for $\Sigma=\set{\sigma}$. The MaxCut problem is
  the problem of finding a cut in a graph (i.e., a partition of the
  vertices into two disjoint subsets), such that the number of edges
  crossing the cut is the highest among all possible cuts. This
  problem is known to be NP-hard~\cite{10.5555/574848}.

Given a graph $g$, with $n$ vertices and $m$ edges, we construct an input to our problem (that is, a database $D$) as follows. For each vertex $v_i$ we add the following two facts to the database:
\begin{align*}
    R(\val{1},v_i), R(v_i,\val{2})
\end{align*}
Moreover, for each edge $(v_i, v_j)$, we add the following two facts to the database:
\begin{align*}
    R(v_j,v_i), R(v_i,v_j)
\end{align*}

Note that for each vertex $v_i$, the facts $R(\val{1},v_i)$ and
$R(v_i,\val{2})$ violate the EGD together. Moreover, two facts of the form $R(\val{1},v_i)$ and $R(v_i,v_j)$ jointly
violate the EGD, and two facts of the form $R(v_i,\val{2})$ and $R(v_j,v_i)$
jointly violate the EGD. Finally, two facts of the form $R(v_i,v_j)$ and $R(v_j,v_k)$ violate the EGD with each other (and in the case where $x_i=x_1$ and $x_j=x_2$ or the case where $x_i=x_2$ and $x_j=x_3$ this also holds when $v_i=v_k$). These are the only violations of the EGD in
the database.

We set the cost $\kappa(o,D)$ to be $1$ when the operation $o$ is a deletion of a fact of the form $R(v_i,v_j)$, and we set $\kappa(o,D)$ to be $m+1$ when the operation $o$ is a deletion of a fact of the form $R(\val{1},v_i)$ or $R(v_i,\val{2})$.
We now prove that there is a cut of size at least $k$, if and only if \[\Imr(\Sigma,D)\le(m+1)\cdot n+2(m-k)+k\]

First, assume that there exists a cut of size $k$ in the graph, that partitions the vertices into two groups - $S_1$ and $S_2$. In this case, we can remove the following facts from $D$ to obtain a consistent subset $D'$.
\begin{itemize}
    \item $R(\val{1},v_i)$ if $v_i\in S_2$,
    \item $R(v_i,\val{2})$ if $v_i\in S_1$,
    \item $R(v_j,v_i)$ if either $R(\val{1},v_j)$ or $R(v_i,\val{2})$ have not been removed.
\end{itemize}
Each vertex $v_i$ belongs to either $S_1$ or $S_2$; hence, we remove exactly one of the facts $R(\val{1},v_i)$ and $R(v_i,\val{2})$ for each $v_i$, and resolve the conflict between these two facts. The cost of removing these $n$ facts is $(m+1)\cdot n$. Next, for each edge $(v_i,v_j)$ such that both $v_i$ and $v_j$ belong to the same subset $S_k$, we remove both $R(v_j,v_i)$ and $R(v_i,v_j)$, since the first violates the EGD with $R(\val{1},v_j)$ if $v_i,v_j\in S_1$ or with $R(v_i,\val{2})$ if $v_i,v_j\in S_2$, and the second violates the EGD with $R(\val{1},v_i)$ if $v_i,v_j\in S_1$ or with $R(v_j,\val{2})$ if $v_i,v_j\in S_2$. The cost of removing these $2(m-k)$ facts is $2(m-k)$. 

Finally, for each edge $(v_i,v_j)$ that crosses the cut, we remove one of $R(v_j,v_i)$ or $R(v_i,v_j)$ from the database. If $v_i\in S_1$ and $v_j\in S_2$, then we have already removed the facts $R(\val{1},v_j)$ and $R(v_i,\val{2})$ from the database; thus, the fact $R(v_j,v_i)$ does not violate the EGD with any other fact, and we only have to remove the fact $R(v_i,v_j)$ that violates the EGD with both $R(\val{1},v_i)$ and $R(v_j,\val{2})$. Similarly, if $v_i\in S_2$ and $v_j\in S_1$, we only remove the fact $R(v_j,v_i)$. The cost of removing these $k$ facts is $k$. Hence, the total cost of removing all these facts is $(m+1)\cdot n+2(m-k)+k$.

Clearly, the result is a consistent subset $D'$ of $D$. As previously explained, we have resolved the conflict between the facts $R(\val{1},v_i)$ and $R(v_i,\val{2})$ for each $v_i$, and we have resolved the conflict between every pair $\set{R(v_i,v_j),R(\val{1},v_i)}$ of conflicting facts, and every pair $\set{R(v_i,v_j),R(v_j,\val{2})}$ of conflicting facts. Finally, there are no conflicts among facts $R(v_i,v_j)$ and $R(v_j,v_k)$ in $D'$ since $v_j$ either belongs to $S_1$, in which case the fact $R(\val{1},v_j)$ appears in $D'$ and $R(v_j,v_k)$ has been removed from $D'$ as a result, or it belongs to $S_2$, in which case the fact $R(v_j, \val{2})$ appears in $D'$ and $R(v_i,v_j)$ has been removed from $D'$ as a result. Thus, the minimal cost of obtaining a consistent subset of $D$ (i.e., a repair) is at most $(m+1)\cdot n+2(m-k)+k$.

Next, we assume that $\Imr(\Sigma,D)\le(m+1)\cdot n+2(m-k)+k$, and we prove that there exists a cut of size at least $k$ in the graph. First, note that if there exists a consistent subset $D'$ of $D$ that can be obtained with cost at most $(m+1)\cdot n+2(m-k)+k$, such that both $R(\val{1},v_i)$ and $R(v_i,\val{2})$ have been deleted, then we can obtain another consistent subset of $D$ with a lower cost by removing only $R(\val{1},v_i)$ and removing all the facts of the form $R(v_j,v_i)$ instead of removing $R(v_i,\val{2})$. There are at most $m$ such facts and the cost of removing them is at most $m$, while the cost of removing $R(v_i,\val{2})$ is $m+1$. Hence, from now on we assume that the subset $D'$ contains exactly one fact from $\set{R(\val{1},v_i), R(v_i,\val{2})}$ for each $v_i$.

Now, we construct a cut in the graph from $D'$ in the following way. For each $v_i$, if the fact $R(\val{1},v_i)$ belongs to $D'$, then we put $v_i$ in $S_1$, and if the fact $R(v_i,\val{2})$ belongs to $D'$, then we put $v_i$ in $S_2$. As mentioned above, exactly one of these two cases holds for each $v_i$. It is only left to show that the size of the cut is at least $k$. Since the cost of removing the facts of the form $R(\val{1},v_i)$ and $R(v_i,\val{2})$ is $(m+1)\cdot n$, and the cost of removing each fact of the form $R(v_i,v_j)$ is one, at most $2(m-k)+k$ facts of the form $R(v_i,v_j)$ have been removed from $D$ to obtain $D'$. There are $2m$ facts of the form $R(v_i,v_j)$ in $D$; thus, at least $k$ of them belong to $D'$. 

For each fact $R(v_i,v_j)$ in $D'$ it holds that both $R(\val{1},v_i)$ and ${R(v_j,\val{2})}$ do not belong to $D'$ (otherwise, $D'$ is inconsistent). Hence, it holds that $v_i\in S_2$ and $v_j\in S_1$, and the corresponding edge $(v_i,v_j)$ crosses the cut. We conclude that there are at least $k$ edges that cross the cut.
\end{proof}

We now move on to the proof of the positive side of the theorem. We first consider the case where the EGD contains two different relations and show that in this case $\Imr(\Sigma,D)$ can always be computed in polynomial time.

\begin{lemma}
Let $\R=\R_{\subseteq}$ and let $\Sigma=\{\sigma\}$ where $\sigma$ is an EGD of the form
$$\forall x_1,x_2,x_3,x_4 \big[R(x_{i_1},x_{i_2}),S(x_{j_1},x_{j_2})\Rightarrow (x_i=x_j)\big]$$
Then, computing $\Imr(\Sigma,D)$ can be done in polynomial time.
\end{lemma}
\begin{proof}
First, if $x_{i_1}=x_{i_2}$, then facts of the form $R(\val{a},\val{b})$ for $\val{a}\neq\val{b}$ will never participate in a violation; hence we can ignore these facts when computing a minimum repair and add them to the repair later. Similarly, if $x_{i_3}=x_{i_4}$, then we can ignore all the facts of the form $R(\val{a},\val{b})$ for $\val{a}\neq\val{b}$. Thus, from now on we assume that the database only contains facts that may participate in a violation.

If $R(x_{i_1},x_{i_2})$ and $S(x_{j_1},x_{j_2})$ share variables (e.g., $x_{i_2}=x_{j_1}$), then two facts will violate the EGD only if they agree on the values of the shared variables. Hence, we first split the database into blocks of facts that agree on the values of the shared variables, and we solve the problem separately for each one of these blocks, since there are no violations among facts in different blocks. Then, a minimum repair for the entire database will be the disjoint union of the minimum repairs for each one of the blocks. For example, if the only shared variable is $x_{i_2}=x_{j_1}$, then each block will contain facts of the form $R(\cdot,\val{a})$, $S(\val{a},\cdot)$ for some value $\val{a}$. Note that if the two atoms do not share any variables, then we will have one block that contains all the facts in the database.

We start by considering the case where $x_i$ and $x_j$ both appear in either $R(x_{i_1},x_{i_2})$ or in $S(x_{j_1},x_{j_2})$. We assume, without the loss of generality, that they both appear in $R(x_{i_1},x_{i_2})$ (that is, $x_i=x_{i_1}$ and $x_j=x_{i_2}$). This means, that as long as there is at least one fact in $S$ that belongs to the current block, all the facts in $R$ that belong to the block should be of the form $R(\val{a},\val{a})$ for some value $\val{a}$. Hence, we have two possible ways to obtain a consistent subset of a block: \e{(a)} remove all the facts from $S$, or \e{(b)} remove from $R$ all the facts $R(\val{a},\val{b})$ such that $\val{a}\neq\val{b}$. We can compute the cost of each one of these options and choose the one with the lower cost.

In the case where $x_i$ appears only in $R(x_{i_1},x_{i_2})$ and $x_j$ appears only in $S(x_{j_1},x_{j_2})$, we have three possible ways to obtain a consistent subset of a block: \e{(a)} remove all the facts from $R$, \e{(b)} remove all the facts from $S$, or \e{(c)} choose a single value that will appear in all the attributes corresponding to the variables $x_i$ and $x_j$, and remove from $R$ and $S$ all the facts that use different values in these attributes. We can again compute the cost of each one of these options and choose the one with the lowest cost.
\end{proof}

Next, we consider EGDs that use a single relation, such that the two atoms in the EGD do not share any variables.

\begin{lemma}
Let $\R=\R_{\subseteq}$ and let $\Sigma=\{\sigma\}$ where $\sigma$ is an EGD of the form
$$\forall x_1,x_2,x_3,x_4 \big[R(x_1,x_2),
R(x_3,x_4)\Rightarrow 
(x_i=x_j)\big]$$
Then, computing $\Imr(\Sigma,D)$ can be done in polynomial time.
\end{lemma}
\begin{proof}
If $x_i$ and $x_j$ are both from either $\set{x_1,x_2}$ or $\set{x_3,x_4}$, then each fact of the form $R(\val{a},\val{b})$ such that $\val{a}\neq\val{b}$ violates the EGD by itself. Thus, we have to remove from the database all such facts, and we can compute the cost of removing these facts in polynomial time. Otherwise, $x_i$ is one of $x_1,x_2$ and $x_j$ is one of $x_3,x_4$, in which case we have to choose some value $\val{a}$ and remove from the database all the facts that do not use this value in both the attributes corresponding to $x_i, x_j$. 

If $x_i=x_1$ and $x_j=x_3$, then all the facts in the database should agree on the value of the first attribute, and we have to choose the value that entails the lowest cost (that is, the cost of removing all the facts that use a different value in the first attribute is the lowest among all possible values). The case where $x_i=x_2$ and $x_j=x_4$ is symmetric to this case. 

If $x_i=x_2$ and $x_j=x_3$, then again each fact will violate the EGD by itself unless it is of the form $R(\val{a},\val{a})$. In addition, two facts $R(\val{a},\val{a})$ and $R(\val{b},\val{b})$ for $\val{a}\neq\val{b}$ will jointly violate the EGD. Thus, only one fact of the form $R(\val{a},\val{a})$ will not be removed from the database, and we again have to choose the one that entails the lowest cost. The case where $x_i=x_1$ and $x_j=x_3$ is symmetric to this case. Hence, computing $\Imr(\Sigma,D)$ can again be done in polynomial time. 
\end{proof}

So far, we have covered the cases where the EGD either considers two different relations or considers only one relation, but there are no shared variables among the atoms. In both cases, computing $\Imr(\Sigma,D)$ can always be done in polynomial time. Next, we consider the cases where the EGD uses a single relation and the two atoms do share variables, but the EGD does not satisfy the condition of the theorem.

\begin{lemma}
Let $\R=\R_{\subseteq}$ and let $\Sigma=\{\sigma\}$ where $\sigma$ is an EGD of one of the following forms:
\begin{enumerate}
\item $\forall x_1,x_2 \big[R(x_1,x_2),
R(x_1,x_2)\Rightarrow 
(x_1=x_2)\big]$
\item $\forall x_1,x_2,x_3 \big[R(x_1,x_2),
R(x_1,x_3)\Rightarrow 
(x_i=x_j)\big]$
\item $\forall x_1,x_2 \big[R(x_1,x_2),
R(x_2,x_1)\Rightarrow 
(x_1=x_2)\big]$
\end{enumerate}
Then, computing $\Imr(\Sigma,D)$ can be done in polynomial time.
\end{lemma}
\begin{proof}
It is straightforward that computing $\Imr(\Sigma,D)$ w.r.t.~an EGD of
the first form can be done in polynomial time, as each fact in $R$ violates
the EGD by itself, unless it uses the same value in both
attributes. Thus, we have to remove from the database all the facts of
the form $R(\val{a},\val{b})$ for $\val{a}\neq\val{b}$, and we can
compute the cost of removing these facts in polynomial time.

For the second case, if $x_i=x_2$ and $x_j=x_3$, then an EGD of this form is an FD, and $\Imr(\Sigma,D)$ can be computed in polynomial time~\cite{DBLP:conf/pods/LivshitsKR18}. If $x_i,x_j\in\set{x_1,x_2}$ or $x_i,x_j\in\set{x_1,x_3}$, then we are again in the case where only facts of the form $R(\val{a},\val{a})$ are allowed, and we have to remove the rest of the facts from the database.

Finally, for EGDs of the third form, it holds that only pairs of facts $\set{R(\val{a},\val{b}), R(\val{b},\val{a})}$ violate the EGD; hence, for each such pair we have to remove one of facts in the pair, and we will remove the one that entails the lower cost.
\end{proof}

\section{Proof of Theorem~\ref{thm:lin}}

\begin{reptheorem}{\ref{thm:lin}}
\thmlin
\end{reptheorem}
\begin{proof}
  The tractability of $\Ilmr$ is simply due to the fact that we can
  enumerate $\MI_\Sigma(D)$ in polynomial time, hence construct the
  LP, and then use any polynomial-time LP
  solver~\cite{khachiyan1979polynomial}. So, in the remainder of the
  proof, we prove that the four properties are satisfied. By ``the LP
  program'' we refer to the linear relaxation of the ILP program of
  Figure~\ref{fig:ilp}, where the bottom condition is replaced with
  $\forall i\in\tids(D): 0\leq x_i\leq 1$.

  Proving \e{positivity} is straightforward. In particular, if $D$
  violates $\Sigma$, then $\MI_\Sigma(D)$ is nonempty; hence, the zero
  assignment (i.e., assigning zero to every $x_i$) is infeasible due
  to the first constraint, and the resulting value of the objective
  function is strictly positive.

  For \e{monotonicity}, let $D$ be a database, and let $\Sigma$ and
  $\Sigma'$ be two sets of DCs such that $\Sigma'\models\Sigma$.  To
  prove that $\Ilmr(\Sigma,D)\leq \Ilmr(\Sigma',D)$ it suffices to
  prove that every feasible solution w.r.t.~$\Sigma'$ is also a
  feasible solution w.r.t.~$\Sigma$. More specifically, if an
  assignment to the $x_i$'s satisfies the set~\eqref{eq:ineq} of
  inequalities for $\MI_{\Sigma'}(D)$, then it also
  satisfies~\eqref{eq:ineq} for $\MI_{\Sigma}(D)$. For that, it
  suffices to show that for every $E\in\MI_\Sigma(D)$ there is
  $E'\in\MI_{\Sigma'}(D)$ such that $E'\subseteq E$. This is
  straightforward from our assumption that $\Sigma'\models\Sigma$: if
  $E\in\MI_\Sigma(D)$, then $E\not\models\Sigma$; hence,
  $E\not\models\Sigma'$, and $E$ contains a minimal inconsistent
  subset $E'\in\MI_{\Sigma'}(D)$. 

  As for \e{progression}, let $D$ be a database, and let $\Sigma$ be a
  set of DCs such that $D$ violates $\Sigma$.  Let $\mu$ be an
  assignment to the $x_i$'s of the LP such that $\mu$ realizes
  $\Ilmr(\Sigma,D)$. Take any $j\in\tids(D)$ such that $x_j>0$. Let
  $D'$ be the database obtained from $D$ by removing $D[j]$, and let $\mu'$ be the assignment obtained by restricting $\mu$ to $\tids(D)\setminus\set{j}$. Then,
  $\mu'$ is a feasible solution, and its objective value is smaller
  than that of $\mu$. We conclude that $\Ilmr(\Sigma,D')<
  \Ilmr(\Sigma,D)$.

  Finally, we prove constant weighted \e{continuity}. Let $\Sigma$ be
  a set of DCs, let $D_1$ and $D_2$ be two inconsistent databases, and
  let $o_1=\del{i}{\cdot}$ be an operation on $D_1$. We need to find
  an operation $o_2=\del{j}{\cdot}$ on $D_2$ such that
  $\Delta_{\Ilmr,\Sigma}(o_1,D_1) \le \delta\cdot
  \Delta_{\Ilmr,\Sigma}(o_2,D_2)$ for some constant $\delta$ that depends
  only on $\Sigma$. Let $\mu_1$ be an assignment for the LP that
  realizes $\Ilmr(\Sigma,D_1)$, and let $\mu_1'$ be an assignment for
  the LP that realizes $\Ilmr(\Sigma,o_1(D_1))$. Let $\mu_1''$ be the
  assignment that is obtained by extending $\mu_1'$ with
  $\mu_1''(x_i)=1$; that is, $\mu_1''(x_j)=\mu_1'(x_j)$ for
  $j\in\tids(D_1)\setminus\set{i}$, and $\mu_1''(x_j)=1$ for
  $j=i$. Then the objective value for $\mu''_1$ is greater than the
  objective value for $\mu'_1$ by at most $\kappa(o_1,D_1)$. Moreover,
  $\mu_1''$ is a feasible solution w.r.t.~$D_1$. We conclude that
\begin{equation}\label{eq:d1}  
\Delta_{\Ilmr,\Sigma}(o_1,D_1) \leq \kappa(o_1,D_1)\,.
\end{equation}
  Now, let $\mu_2$ be an assignment that realizes $\Ilmr(\Sigma,D_2)$,
  and let $E$ be any set in $\MI_\Sigma(D_2)$. Then, from the
  definition of the LP it follows that there exists a tuple identifier
  $j$ such that $D_2[j]\in E$ and $\mu_2(x_j)\geq 1/|E|$. Therefore, by
  removing $D_2[j]$ we get a feasible solution w.r.t.~$D_2$ such that
  the reduction in inconsistency is at least $\kappa(o_2,D_2)/|E|$ for
  $o_2=\del{j}{\cdot}$. We conclude that
  \begin{equation}\label{eq:d2}
\Delta_{\Ilmr,\Sigma}(o_2,D_2) \geq \frac{\kappa(o_2,D_2)}{|E|}\,.
\end{equation}
Combining~\eqref{eq:d1} and~\eqref{eq:d2}, we conclude that
\[\Delta_{\Ilmr,\Sigma}(o_1,D_1) \leq |E|\cdot
\frac{\kappa(o_1,D_1)}{\kappa(o_2,D_2)}\cdot \Delta_{\Ilmr,\Sigma}(o_2,D_2)\,.
\]
Finally, we observe that the cardinality $|E|$ is bounded by the
maximal number of atoms in any DC of $\Sigma$. Denoting this maximal
number by $d_\Sigma$, we conclude constant weighted continuity by
taking $\delta=d_\Sigma$.
  \end{proof}
  
\section{Additional Experimental Results}

\subsection{Measure Behavior} For completeness, we have conducted the measure behavior experiments on a small sample of one hundred tuples from each dataset, where we are able to compare the behavior of all the measures with the behavior of $\Imc$. The results are depicted in Figure~\ref{fig:results-100}. Some of the graphs for $\Imc$ are missing, as the computation again exceeded our 24 hour time limit, even on these small datasets. While we see that the behavior of most measures is less stable on the small datasets and most graphs are more jittery than the ones we obtained on the larger dataset, the general trend is very similar to the one we observed on the larger datasets. The behavior of $\Imc$ seems to be the most unpredictable. For example, on the Stock dataset, we see that $\Imc$ fails to indicate progress for long periods of time. As another example, on the Airport dataset, its value jump up and down throughout the execution.

\begin{figure*}[t!]
  \subfloat[\small \bf Noise added with \algname{CONoise}.\label{fig:results_first_100}]
  {\includegraphics[width=6in]{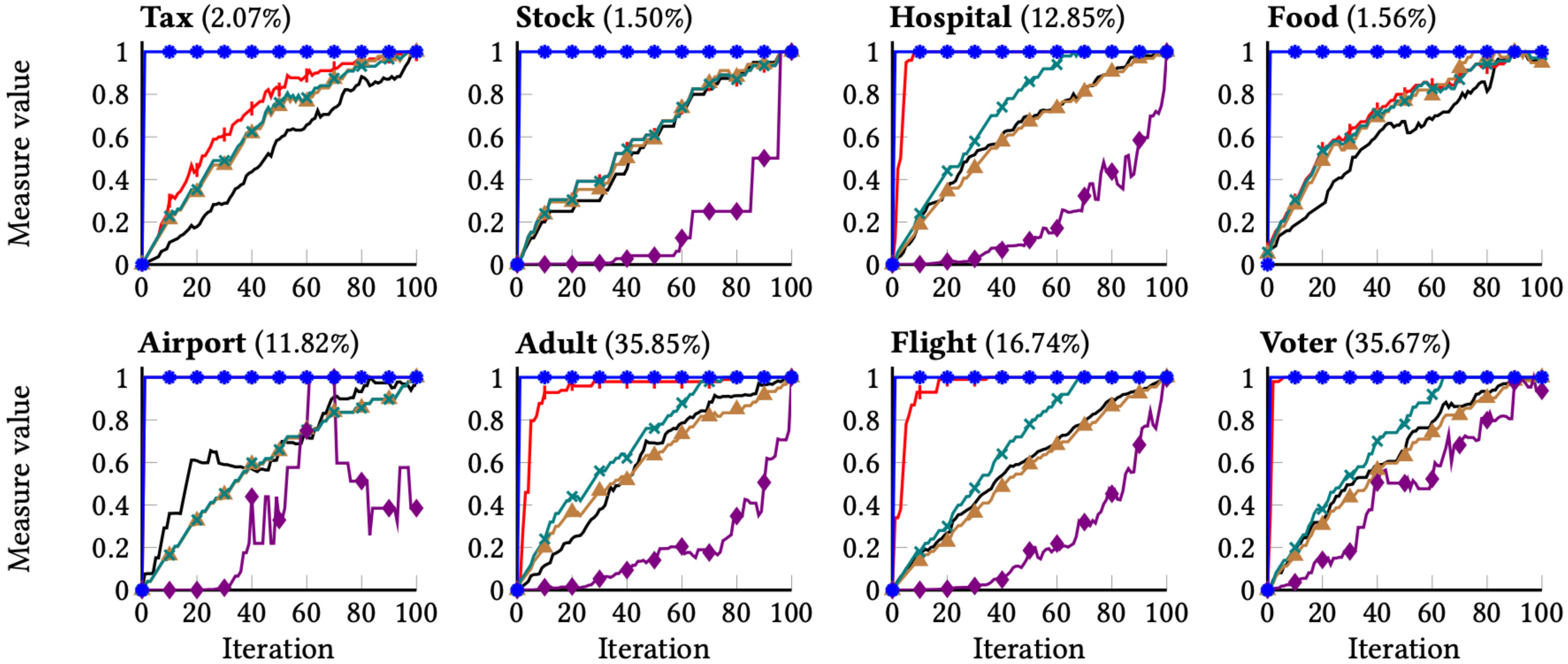}}
  \\
  \subfloat[\small \bf Noise added with \algname{RNOise} ($\alpha=0.01$ and $\beta=0$).\label{fig:results_second_100}]{\includegraphics[width=6in]{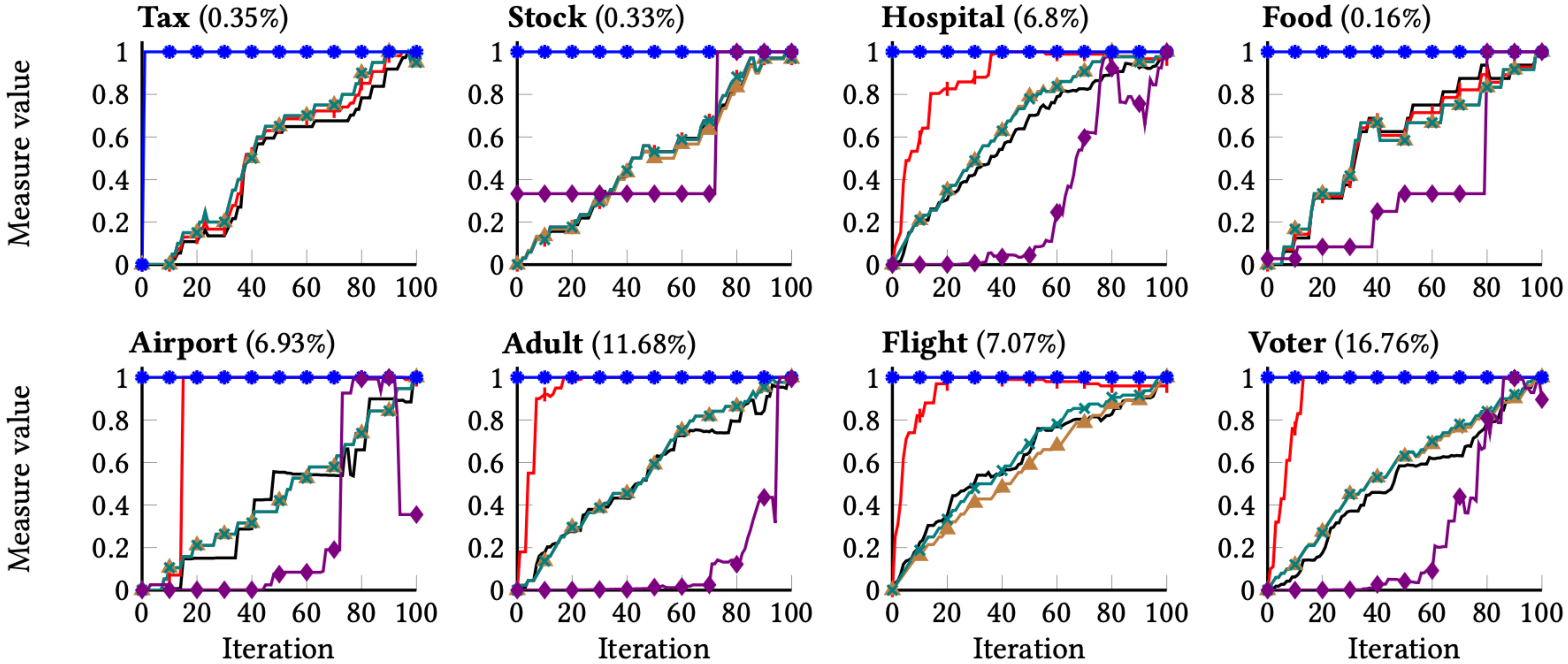}}
   \caption{The normalized values of $\Id$ (\ref{plot_id}), $\Imi$ (\ref{plot_imi}), $\Ip$ (\ref{plot_ip}), $\Imc$ (\ref{plot_imc}), $\Imr$ (\ref{plot_ir}), and $\Ilmr$ (\ref{plot_ilinr}) on datasets with one hundred tuples. Missing $\Imc$ graphs are due to timeout.}
    \label{fig:results-100}
  \end{figure*}
  
\paragraph*{Data skew and error types}

As aforementioned, we have also conducted the measure behavior experiments using \algname{RNoise} with $\beta=1$ and $\beta=2$, to test different data skew levels. The results are given in Figure~\ref{fig:results_skew}. Moreover, we repeated the experiment with $\beta=1$ under different probabilities for introducing typos. Figure~\ref{fig:prob1} depicts the results for the case where we change a cell value to a typo with probability $0.2$ (and to another random value from the active domain with probability $0.8$), and Figure~\ref{fig:prob2} depicts the results for the case where the probability of a typo is $0.8$.
All the charts are very similar to the one we obtained for $\beta=0$ (Figure~\ref{fig:results_second}); hence, the behavior of the measures does not seem to be sensitive to these parameters.

\begin{figure*}[t!]
\subfloat[\small \bf Noise added with \algname{RNoise} ($\alpha=0.01$ and $\beta=1$).\label{fig:skew1}]{
  \includegraphics[width=6in]{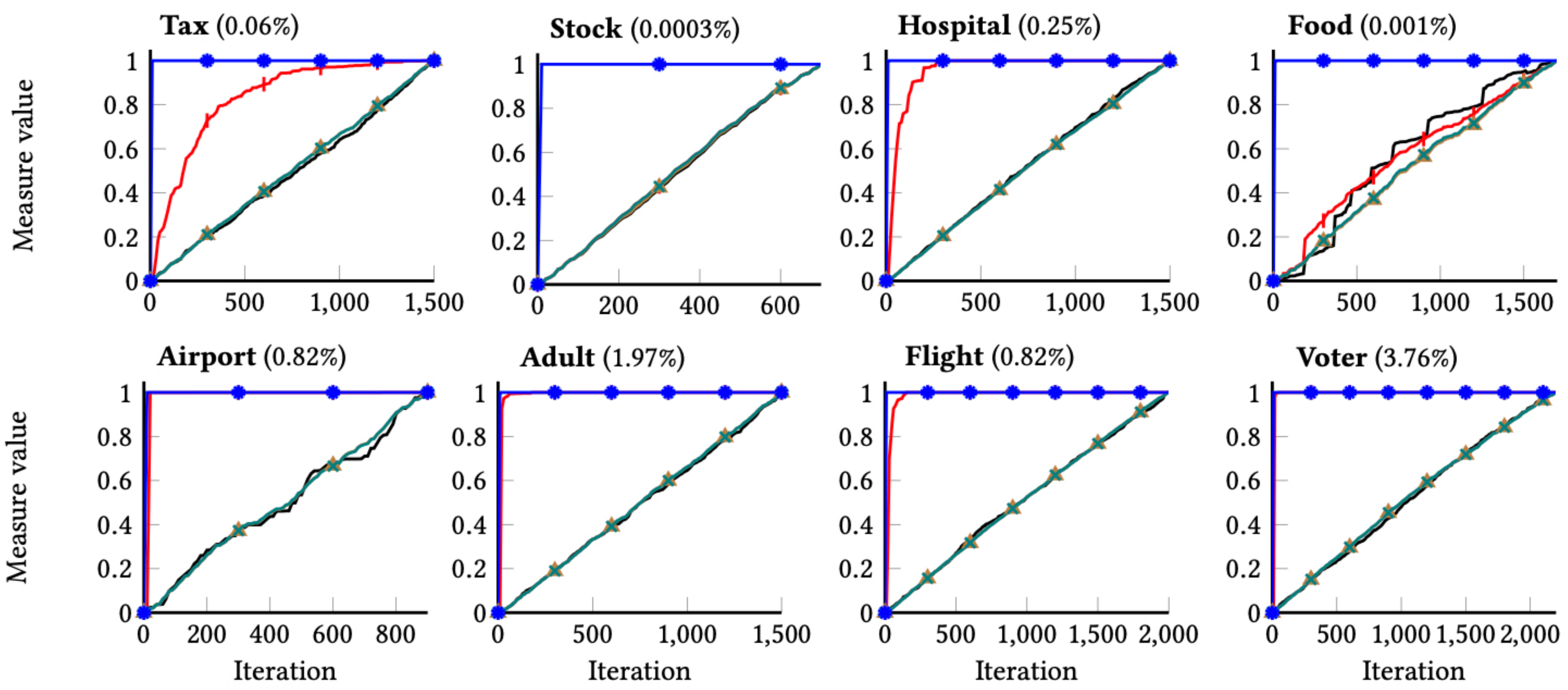}}\\
  \subfloat[\small \bf Noise added with \algname{RNoise} ($\alpha=0.01$ and $\beta=2$).\label{fig:skew2}]{
    \includegraphics[width=6in]{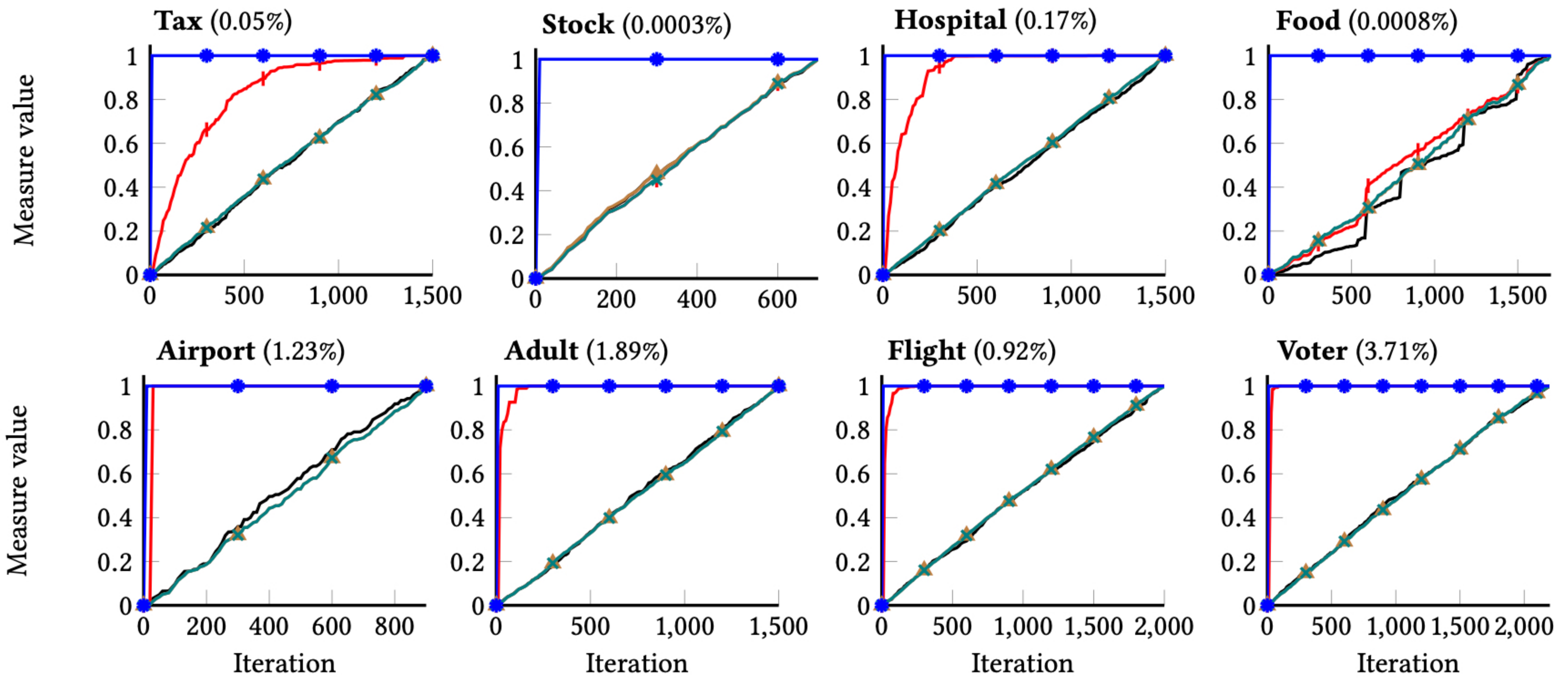}}
    \caption{The normalized values of $\Id$ (\ref{plot_id}), $\Imi$ (\ref{plot_imi}), $\Ip$ (\ref{plot_ip}), $\Imr$ (\ref{plot_ir}), and $\Ilmr$ (\ref{plot_ilinr}).}
    \label{fig:results_skew}
\end{figure*}

\begin{figure*}[t!]
\subfloat[\small \bf Noise added with \algname{RNoise} ($\alpha=0.01$ and $\beta=1$). Typo probability is $0.2$.\label{fig:prob1}]{
  \includegraphics[width=6in]{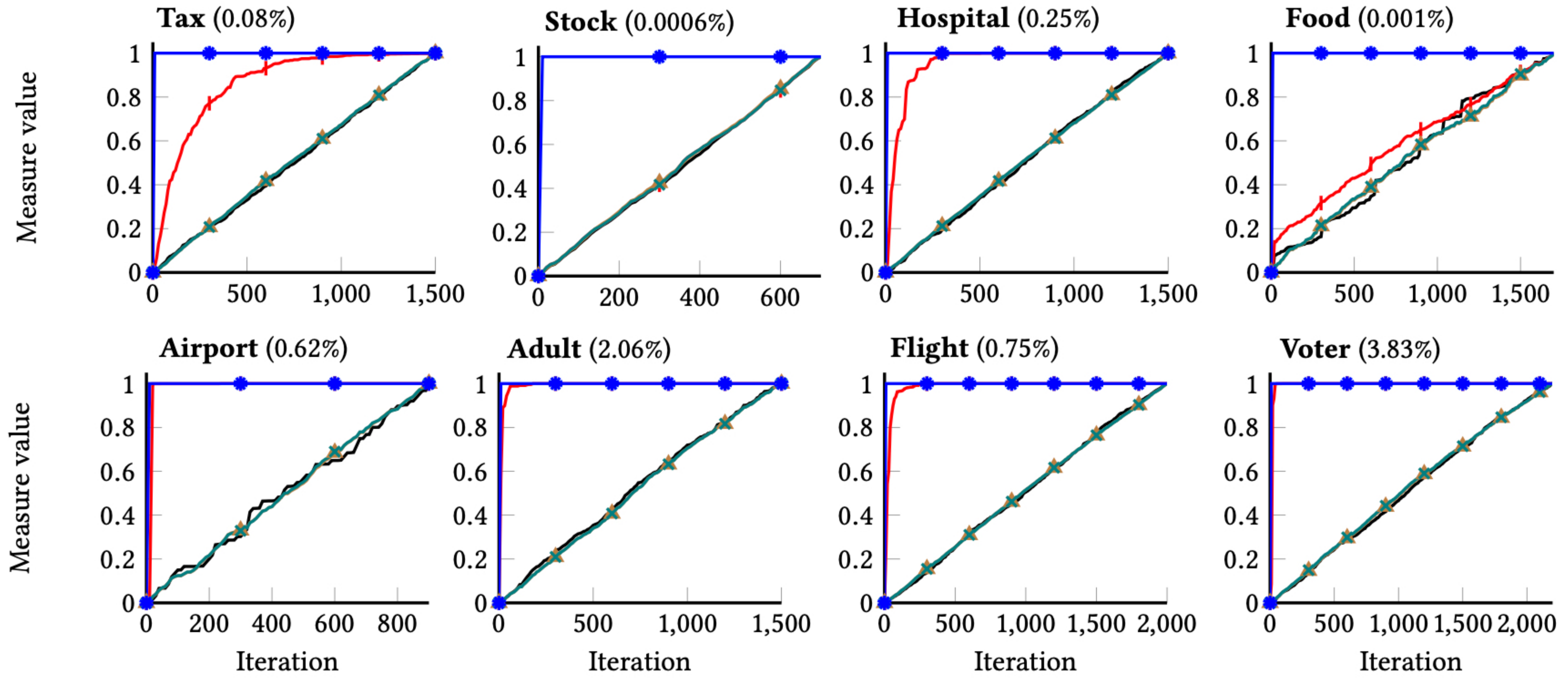}}\\
  \subfloat[\small \bf Noise added with \algname{RNoise} ($\alpha=0.01$ and $\beta=2$). Typo probability is $0.8$.\label{fig:prob2}]{
    \includegraphics[width=6in]{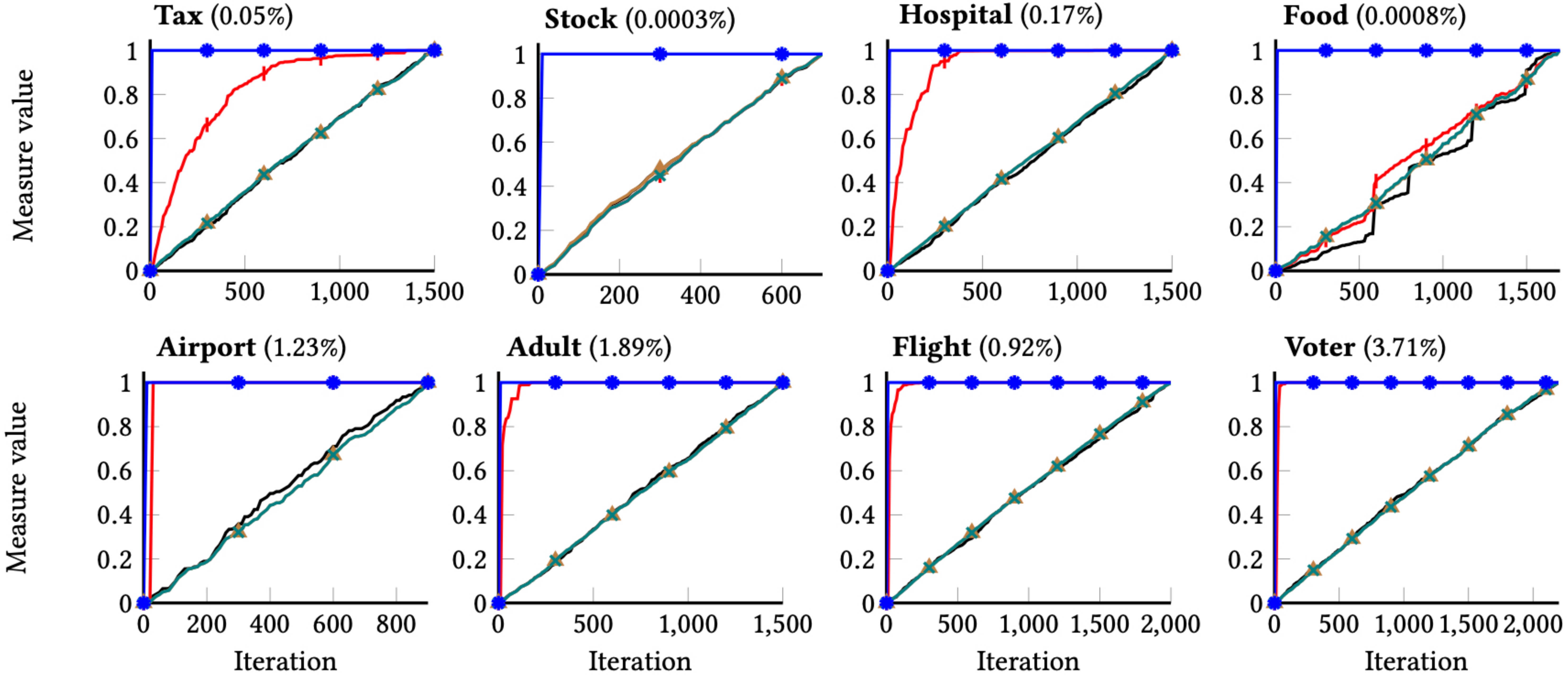}}
   \caption{The normalized values of $\Id$ (\ref{plot_id}), $\Imi$ (\ref{plot_imi}), $\Ip$ (\ref{plot_ip}), $\Imr$ (\ref{plot_ir}), and $\Ilmr$ (\ref{plot_ilinr}).}
    \label{fig:results_skew2}
\end{figure*}

\subsection{Running Times}

Figure~\ref{fig:runtime_er} depicts the running times of the different measures on samples of 10K tuples from each dataset. Using \algname{RNoise} with $\alpha=0.01$ and $\beta=0$, we added noise to these datasets, and computed the running times for all measures every ten iterations. We observe that the running times of $\Imi$ and $\Ip$ are only slightly affected by the error rate (that increases with the number of iterations), while the running time of $\Imr$ is affected the most and significantly increases with the error rate. On the Stock and Food datasets we do not see this trend, as the total number of violations in these datasets is low (see the percentage of violating pairs above the charts) even after modifying $1\%$ of the values in the dataset (i.e., at the end of the execution).

\begin{figure*}[t!]
  \includegraphics[width=6in]{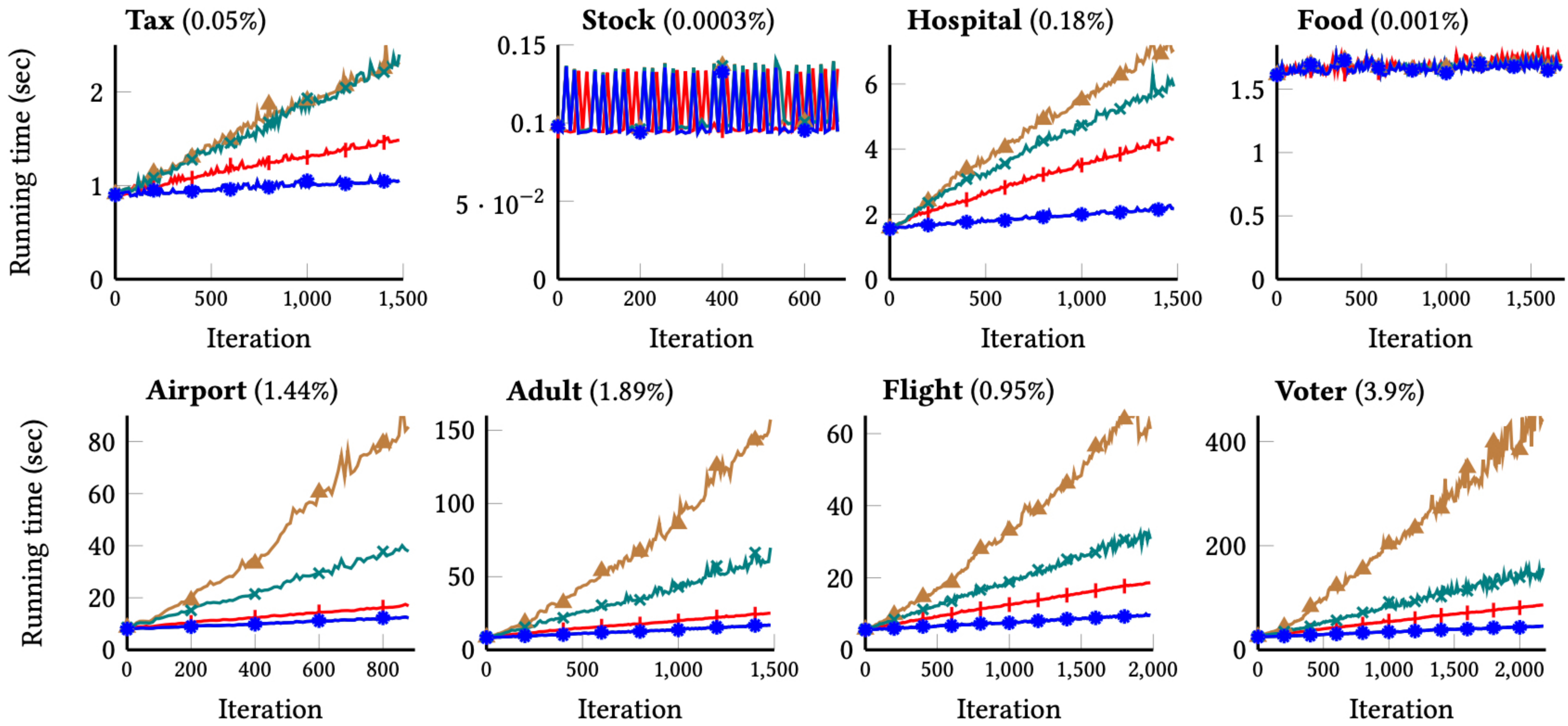}
     \caption{Running times in seconds of $\Id$ (\ref{plot_id}), $\Imi$ (\ref{plot_imi}), $\Ip$ (\ref{plot_ip}), $\Imr$ (\ref{plot_ir}), and $\Ilmr$ (\ref{plot_ilinr}) for varying error rates (the error rate increases with the number of iterations).}
    \label{fig:runtime_er}
\end{figure*}